\documentclass[12pt]{article}

\usepackage{natbib}
\usepackage{amsmath,amssymb,mathtools,commath,comment,lmodern,dsfont,booktabs}
\usepackage{float}
\counterwithin{equation}{section}
\usepackage[dvipsnames]{xcolor}
\definecolor{UBCblue}{RGB}{0, 0, 95} 
\definecolor{ForestGreen}{RGB}{34, 139, 34}
\definecolor{a2red}{RGB}{192,0,0}
\definecolor{a3sand}{RGB}{191,144,0}
\definecolor{a4green}{RGB}{0,204,0}

\usepackage{amssymb,amsmath,amsfonts,eurosym,geometry,ulem,graphicx,caption,color,setspace,sectsty,comment,footmisc,caption,natbib,pdflscape}
\usepackage{amsthm,subfigure,array,centernot,pgfplots}
\usepackage{algorithm}
\usepackage{algpseudocode}
\usepackage[listings,skins,breakable]{tcolorbox}
\usepackage[hidelinks]{hyperref}
\usepackage{tikz}
\usepackage{pifont}

\newtheorem{theorem}{Theorem} 
\newtheorem*{theorem*}{Theorem}
\newtheorem{corollary}{Corollary}  
\newtheorem{lemma}{Lemma}  

\newtheorem{assumption}{Assumption}[section]

\theoremstyle{definition}
\newtheorem{definition}{Definition}[section]

\newcolumntype{L}[1]{>{\raggedright\let\newline\\arraybackslash\hspace{0pt}}m{#1}}
\newcolumntype{C}[1]{>{\centering\let\newline\\arraybackslash\hspace{0pt}}m{#1}}
\newcolumntype{R}[1]{>{\raggedleft\let\newline\\arraybackslash\hspace{0pt}}m{#1}}

\newcommand{\CI}{\mathrel{\perp\mspace{-10mu}\perp}}
\newcommand{\nCI}{\centernot{\CI}}

\newcommand{\E}[1]{\operatorname{\mathbb{E}}\left[#1\right]}

\newcommand{\Var}[1]{\operatorname{Var}\left[#1\right]}
\newcommand{\Cov}[1]{\operatorname{Cov}\left[#1\right]}

\usetikzlibrary{shapes,decorations,arrows,calc,arrows.meta,fit,positioning}
\tikzset{
    Latex-Latex,auto,node distance =1 cm and 1 cm, thick,
    state/.style ={ellipse, draw, minimum width = 0.7 cm},
    point/.style = {circle, draw, inner sep=0.04cm,fill,node contents={}},
    bidirected/.style={Latex-Latex,dashed},
    el/.style = {inner sep=2pt, align=left, sloped},
    line/.style={draw, line width=1, -},  
    cross/.style={cross out, draw=black, minimum size=2*(#1-\pgflinewidth), inner sep=0pt, outer sep=0pt},
    cross/.default={1pt}
}

\pgfplotsset{ every non boxed x axis/.append style={x axis line style=-},
     every non boxed y axis/.append style={y axis line style=-}}

\begin{document}

\begin{titlepage}
\title{Relaxing Instrument Exogeneity with Common Confounders}
\author{{\scshape Christian Tien} \thanks{\href{mailto:ct493@cam.ac.uk}{ct493@cam.ac.uk}; Faculty of Economics, University of Cambridge} }
\date{\today}
\maketitle

\begin{abstract}
\noindent 
Instruments can be used to identify causal effects in the presence of unobserved confounding, under the famous relevance and exogeneity (unconfoundedness and exclusion) assumptions. As exogeneity is difficult to justify and to some degree untestable, it often invites criticism in applications. Hoping to alleviate this problem, we propose a novel identification approach, which relaxes traditional IV exogeneity to exogeneity conditional on some unobserved common confounders. 
We assume there exist some relevant proxies for the unobserved common confounders. Unlike typical proxies, our proxies can have a direct effect on the endogenous regressor and the outcome. We provide point identification results with a linearly separable outcome model in the disturbance, and alternatively with strict monotonicity in the first stage.
General doubly robust and Neyman orthogonal moments are derived consecutively to enable the straightforward $\sqrt{n}$-estimation of low-dimensional parameters despite the high-dimensionality of nuisances, themselves non-uniquely defined by Fredholm integral equations.
Using this novel method with NLS97 data, we separate ability bias from general selection bias in the economic returns to education problem.

\vspace*{\fill}

\noindent\textbf{Keywords:} \\ Causal Inference, Unobserved Confounding, Instrumental Variables, Control Function, Proximal Learning 
\end{abstract}
\setcounter{page}{0}
\thispagestyle{empty}
\end{titlepage}
\pagebreak \newpage

\section{Introduction}
Unobserved confounding complicates the identification of a causal effect of a regressor of interest on an outcome. Despite the endogeneity of a regressor of interest, instrumental variable (IV) approaches can identify their causal effect if the famous relevance and exogeneity (unconfoundedness and exclusion) assumptions hold for the instruments. These assumptions are strong and often invite criticism of IV estimates in practice. We propose a novel approach to relax exogeneity, in favour of exogeneity conditional on an unobserved common confounder, for which some relevant variables are observed.

Relaxing exogeneity is only possible when it is replaced by other strong assumptions. 
One way to identify causal effects without instrument exogeneity is from residual distributions, not variation in the explanatory variables \citep{heckman1979, millimet2013}. Very specific forms of heteroskedasticity across the first stage and outcome model can also be used to establish identification without an exclusion restriction \citep{klein2010, lewbel2012}. Others have suggested to use irrelevant variation in instruments to test for the exogeneity of the relevant variation in the instruments \citep{d2021}. A similar idea is followed when integrated conditional moments use nonlinear mean-dependence of endogenous variables on instruments, such that the instruments may violate the exclusion restrictions in pre-specified parametric ways. Despite recent advances in estimation with integrated conditional moments \citep{tsyawo2021}, the strong identifying assumptions render all these approaches difficult to justify in applications. Our solution differs significantly from these approaches as it only uses a relaxed exogeneity assumption and variation in explanatory variables to identify the causal effect of interest.

From the perspective of IV, we allow for some endogeneity in the instruments. That endogeneity originates from unobserved common confounders. We call these unobserved confounders \emph{common}, because there are some observed variables, which are relevant for them. These observed variables are called proxies. In other words, we assume there are some unobserved variables that explain all correlation (association) between the instruments and these proxies. This assumption is testable. Then, we need to argue for the exogeneity of the instruments, but only conditional on the unobserved common confounders (and observable variables). In general, this is a strong relaxation of instrument exogeneity conditional on observed variables only.

Another way to understand our proposal is as a solution to measurement error in observed confounders for IV. Residual bias is a well-known problem when confounders are measured with error. Proximal learning \citep{cui2020} is a solution to the problem, where observed variables measure all unobserved confounders with some error. In proximal learning, the proxies for the unobserved confounder may either be causes of the treatment or outcome variable. These proxies must be sufficiently relevant (e.g. complete) for all unobserved confounders. Separately developed from the proximal learning approach, a control function solution exists with identical conditional independence assumptions and mismeasured confounders \citep{nagasawa2018}. Our solution is different, as we do not assume the existence of measurements for all confounders. Instead, we use instruments and assume that measurements exist for all confounders, conditional on which the instruments would be exogenous. In this sense, our solution can be understood as IV with mismeasured confounders.

As it is standard in the control function literature, our approach will identify average causal (structural) quantities of interest. Unless the outcome model is fully linearly separable in the treatment and disturbance \citep{newey1999}, where in our case the disturbance includes the effect of the common confounder, we identify those average causal (structural) quantities of interest that integrate out the unobservables without dependence on the treatment using a control function \citep{imbens2009}. Our identification approach is most similar to recent advances in nonlinear panels \citep{liu2021}. In panel data, unobserved fixed effects are common to the same variables across time, in a similar way as the unobserved common confounders are common to the instruments and proxies in our setup. In \cite{liu2021}, identification stems from a parametric dimension reduction of the effect of observed variables on the outcome, and an index sufficiency assumption that renders the observed variables independent from the fixed effects conditional on an index of the observed variables. In our approach, identification stems from the existence of more instruments than treatments, and an index sufficiency assumption that renders the instruments independent from the unobserved common confounders conditional on an index of the instruments. Just like \cite{blundell2004}, \cite{liu2021} do not explain how to derive this crucial index function. One of our main contributions is the derivation of the index function, which arises naturally in the common confounding setup.

A motivating example for our proposal is the returns to college education identification problem. It features various biases, ability and selection, and we motivate pre-college test scores as instruments exogenous to selection, while clearly endogenous to ability. proxies are pre-college risky behaviour dummies, which appear to correlate negatively with ability. With NLS97 data, we show that selection bias is the much more economically relevant bias compared to ability bias in this problem.

\section{Setup}

The treatment (action) $A \in \mathcal{A}$ is discrete or continuous with base measure $\mu_A$ of $\mathcal{A} \subseteq \mathbb{R}^{d_A}$. $Y \in \mathcal{Y} \subseteq \mathbb{R}$ is the one-dimensional outcome variable. Other important variables are the instruments $Z \in \mathcal{Z} \in \mathbb{R}^{d_Z}$, the proxies $W \in \mathcal{W} \subseteq \mathbb{R}^{d_W}$, and the common confounders $U \in \mathcal{U} \subseteq \mathbb{R}^{d_W}$.

\begin{assumption}[Common Confounding IV Model] \label{a:riv}
\begin{enumerate}
  \item SUTVA: $Y=Y(A,Z)$.  \label{a:sutva}
  \item Instruments \label{a:iv}
  \begin{enumerate}
    \item Exogeneity: \label{a:iv-exog} $ Y(a, z) = Y(a) \CI Z \ | \ U. $
    \item Index sufficiency: For some $\tau \in \mathcal{L}_2(Z)$, where $T \coloneqq \tau(Z)$, $U \CI Z \ | \ T$. \label{a:iv-index}
    \item Relevance (completeness): \label{a:iv-rel} For any $g(A, T) \in \mathcal{L}_2(A, T)$,  \begin{align} \E{ g(A, T) | Z}  = 0 \text{ only when } g(A,T) = 0. \end{align}
  \end{enumerate}
  \item Proxies \label{a:pl}
  \begin{enumerate}
    \item Exogeneity: $ W \CI Z \ | \ U. $ \label{a:pl-excl}
    \item Relevance (completeness): \label{a:pl-rel} For any $g(U) \in \mathcal{L}_2(U)$, \begin{align} \E{g(U) | W} = 0 \text{ only when } g(U)=0. \end{align} 
  \end{enumerate} 
\end{enumerate}
\end{assumption}

Assumption \ref{a:riv}.\ref{a:sutva} is the standard stable unit treatment value assumption (SUTVA), which implies no interference across units. In assumption \ref{a:riv}.\ref{a:iv-exog} we capture the key relaxation of this model compared to standard IV. It states that the instruments are excluded and unconfounded, yet this exclusion and unconfoundedness may be conditional on an unobserved (vector-valued) random variable $U$. This is a significant relaxation of the standard exclusion restriction and unconfoundedness assumption, which is possible only with assumptions \ref{a:riv}.\ref{a:iv-index}, \ref{a:riv}.\ref{a:iv-rel}, and \ref{a:riv}.\ref{a:pl}. 
In assumption \ref{a:riv}.\ref{a:iv-index}, we introduce a control function $\tau \in \mathcal{T} \subseteq \mathcal{L}_2(Z)$ and a control variable $T \coloneqq \tau(Z)$ such that $T \in \mathcal{Z}^\tau \subseteq \mathcal{Z}$. Conditional on the control variable $T$, the instruments $Z$ are independent from the common confounders $U$. A simple example of such a function is the conditional density $f(U|Z)$. This assumption describing the existence the control function $\tau$ is often called index sufficiency, where $T$ is a (multiple) index of $Z$. In assumption \ref{a:riv}.\ref{a:iv-rel}, we require that conditional on the control variable $T$, the instruments $Z$ are complete for treatment $A$. This is a standard completeness condition. It simply means that keeping the variation of $Z$ described by $T$ fixed, the instruments must remain sufficiently relevant for $A$. In slightly different words, after conditioning on $T$, enough variation must be left in the instruments $Z$ to infer the effect of treatment $A$ on outcome $Y$. Intuitively, we are orthogonalising $Z$ with respect to some of its own variation, the variation in $T$. As in standard IV with observed confounders, this relevance requirement is typically testable. Assumption \ref{a:riv}.\ref{a:pl-excl} states that the proxies $W$ are independent from instruments $Z$ conditional on the common confounders $U$. The proxies $W$ must also be complete for the unobserved common confounders $U$, as stated in \ref{a:riv}.\ref{a:pl-rel}. Again, completeness means the proxies $W$ are sufficiently relevant for $U$. 

A different way to understand these assumptions is that the unobserved variable $U$, which explains all association (correlation) between the proxies $W$ and instruments $Z$, renders the instruments exogenous when observed. In this sense, $W$ can be (possibly quite poor) proxies for what we consider the unobserved common confounder $U$, as long as they are sufficiently relevant. Conditional on $W$, the instruments $Z$ are still endogenous. Common confounders $U$ are never observed, and $W$ could be quite poor proxies for it. Yet, we prove that conditioning on a control function $T$, which renders the instruments $Z$ and proxies $W$ independent, exogeneity of instruments $Z$ is restored. 

\section{Learning the Confounding Structure}   \label{sec:confounding}
In this section, we describe the main idea of the paper. Using only observable information, we find a control variable $T$, conditional on which the instruments $Z$ are independent from the unobserved common confounders $U$. We then explain what may be considered the optimal control variable $T$.

\subsection{Learning a Control Function}   \label{ssec:learning-control-function}

The control function $\tau \in \mathcal{L}_2(Z)$, described in lemma \ref{l:sc-1}, generates the control variable $T$. This control variable renders the instruments $Z$ independent from the unobserved common confounders $U$. Logically, if the instruments $Z$ and the proxies $W$ are independent conditional on $U$, it follows that any such control variable $T$ also renders $Z$ and $W$ independent conditional on $T$. 
\begin{lemma}[]  \label{l:sc-1}
Assume $W \CI Z \ | \ U$ (\ref{a:riv}.\ref{a:pl-excl}).
Take any $\tau \in \mathcal{L}_2(Z)$, where $T \coloneqq \tau(Z)$, such that $U \CI Z \ | \ T$. Then, also $W \CI Z \ | \ T$.
\end{lemma}
One possible such control variable is $T=Z$, yet this would leave no remaining variation in $Z$ to instrument for $A$ conditional on $T$. Also, lemma \ref{l:sc-1} does not provide a way to identify any control function $\tau$ apart from a function which captures the same information as $Z$ itself. For this purpose, we need lemma \ref{l:sc-2}. In this lemma, we establish that any $T = \tau(Z)$, conditional on which the instruments $Z$ and proxies $W$ are independent, also renders $Z$ conditionally independent from the unobserved common confounders $U$. 
\begin{lemma}[] \label{l:sc-2}
Assume $W \CI Z \ | \ U$ (\ref{a:riv}.\ref{a:pl-excl}), and for any $g(U) \in \mathcal{L}_2(U)$, $\E{g(U) | W} = 0$ only when $g(U)=0$ (\ref{a:riv}.\ref{a:pl-rel}).
Take any $\tau \in \mathcal{L}_2(Z)$, where $T \coloneqq \tau(Z)$, such that $W \CI Z \ | \ T$. Then, also $U \CI Z \ | \ T$.
\end{lemma}
Unlike lemma \ref{l:sc-1}, the conclusion of lemma \ref{l:sc-2} is not obvious and requires the completeness of proxies $W$ for the unobserved common confounders $U$. Again, completeness means that the proxies $W$ must be sufficiently relevant for $U$. If this were not the case, it would be impossible to keep all variation in $Z$ that is associated with $U$ fixed, using a control variable $T$ derived only using information about the association of $Z$ and $W$. We can interpret $U$ as all unobserved confounders that associate (correlate) $Z$ and $W$. 

Lemma \ref{l:sc-2} is an important result, because it allows the identification of a control function that does not capture all variation in instruments $Z$. For any $T$ that we identify conditional on which instruments $Z$ and proxies $W$ are independent, the instrument exogeneity assumption can be relaxed to exogeneity conditional on all unobservables $U$ that associate (correlate) $Z$ and $W$ (assumption \ref{a:riv}.\ref{a:iv-exog}). The parallels to standard IV are quite clear: The conditional exogeneity assumption \ref{a:riv}.\ref{a:iv-exog} is untestable, yet relaxed compared to standard IV. The relevance requirement of $Z$ for $A$ conditional on $T$ (assumption \ref{a:riv}.\ref{a:iv-rel}) is testable, yet stricter compared to standard IV. The requirement for relevance of $Z$ for $A$ conditional on $T$ implies that only a subset of control functions $\tau \in \mathcal{L}_2(Z)$, which leave enough relevant variation in $Z$ conditional on $T$, enable model identification under assumption \ref{a:riv}:
\begin{align}
{\mathcal{T}}_\text{valid} &\coloneqq \left\{ \tau \in \mathcal{L}_2(Z): \ \left( W \CI Z \ | \ \tau(Z) \right) \text{ and } \left( \E{g(A, \tau(Z) | Z} = 0 \text{ only when } g(A, \tau(Z))= 0 \right) \right\}  \label{eq:valid-cf}
\end{align}
As both defining relevance conditions of this set ${\mathcal{T}}_\text{valid}$ are testable, its non-emptiness is testable as well. 

\subsection{Optimal Control Function}   \label{ssec:optimal-control-function}
Under assumption \ref{a:riv}, the optimal control function $\tau^{*} \in {\mathcal{T}}_\text{valid}$ out of the set of valid control functions captures the minimum feasible information in $Z$ in a sense of minimising the variance of the asymptotically unbiased estimator $\hat{\theta}$ of some causal estimand $\theta$. Figure \ref{f:bias-variance} illustrates schematically how the degree of inconsistency and asymptotic variance of an IV estimator conditional on the control variable $T$ depend on the complexity of $\tau$. 

\begin{figure}[h]
\begin{center}
\caption{Implied typical estimator properties with control functions $\tau$ of varying complexity}  \label{f:bias-variance}
\includegraphics[width=0.618\textwidth]{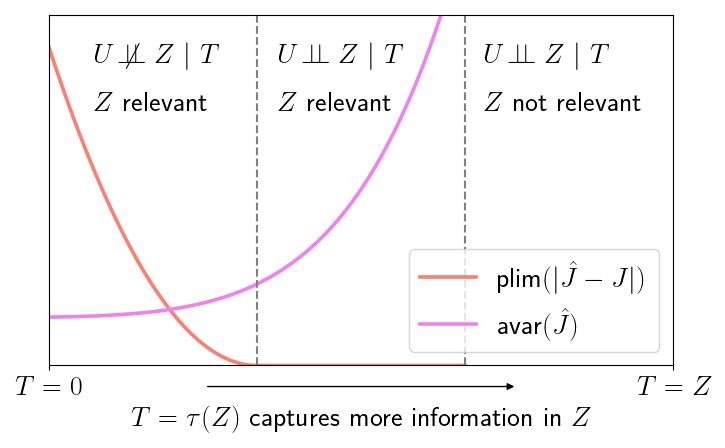}
\end{center}
\begin{footnotesize}
\textit{Notes}:
This figure illustrates some typical properties of an estimator $\hat{\theta}$ of the causal effect of a treatment $A$ on an outcome $Y$, using instruments $Z$ while conditioning on a control function $T = \tau(Z)$. Moving to the right on the x-axis, the control function captures more information in $Z$, starting with variation in $Z$ which correlates with the unobserved common confounders $U$. \\
In the left rectangle, the control function is too simple to render $Z$ exogenous conditional on $T$. Hence, the estimator is inconsistent, yet the degree of inconsistency decreases as the complexity of the control function $T$ increases. In the central rectangle, the control function captures enough information for $Z$ to be excluded conditional on $T$. So, the estimator is consistent. In the right rectangle, the control function is too complex. Conditional on $T$, the instruments $Z$ are no longer sufficiently relevant for $A$ and the estimator $\hat{\theta}$ asymptotically does not exist. In the first two rectangles, the asymptotic variance of the estimator $\hat{\theta}$ increases with the complexity of the control function, because conditional on $T$, less information in $Z$ is used to infer the effect of $A$ on $Y$.
\end{footnotesize}
\end{figure}

In figure \ref{f:bias-variance}, the complexity of $\tau$ on the x-axis increases from $\tau(Z)=0$ on the extreme left to $\tau(Z)=Z$ on the extreme right. Moving further to the right on the x-axis means that the control function captures more information in $Z$, starting with variation in $Z$ which correlates with the unobserved common confounders $U$.
In the left rectangle, the complexity of $\tau$ is low. $T=\tau(Z)$ does not capture all information in $Z$ that correlates with $U$, so even conditional on $T$ the instruments $Z$ remain endogenous, and the estimator $\hat{\theta}$ is inconsistent. However, as all information in $Z$ is used for inference, the asymptotic variance of the estimator $\hat{\theta}$ will be relatively small. As the complexity of $\tau$ increases towards the right in the left quadrant, more information about the elements of $Z$ which correlate with $U$ is captured in $T$. Increasing the complexity of $\tau$ increases the asymptotic variance of $\hat{\theta}$ as less information in $Z$ is used. Importantly, as this information corresponds to variation in $Z$ that correlates with $U$, inconsistency is being reduced. 

In the central rectangle of figure \ref{f:bias-variance}, the complexity of $\tau$ is sufficient for $Z$ to be exogenous conditional on $T$. Hence, the estimator $\hat{\theta}$ is consistent. However, as $\tau$ increases in complexity, we use less information in $Z$ to infer the causal estimand $\theta$. Hence, inevitably the asymptotic variance of the estimator $\hat{\theta}$ increases. Consequently, the optimal $\tau$ would be that of minimal complexity, such that $Z$ is excluded conditional on $T$. In practice, we do not know but can only estimate ${\mathcal{T}}_\text{valid}$, so the exact minimum complexity valid $\tau$ is unknown and estimated with sampling error. However, even if a $\tau$ is chosen with slightly too little complexity, the resulting inconsistency may still be small. The margin of sufficient complexity is at the border of the left and central rectangle. When a $\tau$ with slightly too little complexity is chosen, a small degree of inconsistency is incurred, but depending on the sample size possibly outweighed in terms of mean squared error contribution by the associated standard deviation reduction.
This is an example of a small in-sample bias-variance tradeoff, while we otherwise focus on identification to enable the construction of consistent estimators.

As the complexity of $\tau$ increases, at some point the instruments $Z$ are no longer relevant for treatment $A$ conditional on $T$. Asymptotically, the estimator $\hat{\theta}$ no longer exists. This is the case in the right rectangle of figure \ref{f:bias-variance}. In the extreme, $\tau$ is simply an identify function and $T=Z$. No variation in $Z$ remains to infer the effect of $A$ on $Y$. However, even in less extreme cases where there is some variation left in $Z$ conditional on $T$, it may simply be insufficient variation to be relevant for $A$.

\subsection{Specification Test}     \label{ssec:spec}
A straightforward way to ensure the sufficient complexity of some $\tau$ is to test $W \CI Z \ | \ T$. An alternative to this is a specification test, similar in spirit to specification testing in overidentified IV models. Consider the two control functions $\tau_1$ and $\tau_2$, such that $T_1 = \tau_1(Z)$ and $T_2 = \tau_2(Z)$. Without loss of generality, let $\tau_1$ be \emph{less} complex than $\tau_2$, i.e. $\mathcal{R}(\tau_1) \subset \mathcal{R}(\tau_2)$ where $\mathcal{R}(a)$ is the range space of function $a$. The null hypothesis is the conditional exogeneity of $Z$ given $T_1$,
\begin{align*}
H_0: Y(a) \CI Z \ | \ T_1, \text{ with alternative } H_a: Y(a) \nCI Z \ | \ T_1.
\end{align*}
Let $\hat{\theta}_1$ and $\hat{\theta}_2$ be the two causal estimators of estimand $\theta_0$ based on $\tau_1$ and $\tau_2$. Suppose that under both control functions, the instruments $Z$ remain conditionally relevant for treatment $A$, so that both estimators $\hat{\theta}_1$ and $\hat{\theta}_2$ have some probability limit. We also still assume that conditional on $U$, the instruments $Z$ are exogenous. The conditional exogeneity of $Z$ given $U$ is assumed, because here we only test for the sufficient complexity of $\tau$, i.e. $Z \CI U \ | \ T$.\footnote{To test whether some instruments $Z$ are exogenous conditional on $U$, we can use a standard specification test for different $Z$, if $\theta_0$ is overidentified conditional on $T$.} 

Under the null hypothesis $H_0$, both estimators $\hat{\theta}_1$ and $\hat{\theta}_2$ converge to the true causal effect $\theta_0$. However, the asymptotic variance of $\hat{\theta}_2$ with the more complex control function $\tau_2$ will be larger than that of $\hat{\theta}_1$, as $\hat{\theta}_2$ uses less variation in $Z$ than $\hat{\theta}_2$. Under the alternative $H_a$, the estimators do not have the same probability limit. If $\tau_2$ still captures enough variation $T_2$ in $Z$ for the instruments $Z$ to be conditionally exogenous, $\hat{\theta}_2$ still converges to $\theta_0$. $\hat{\theta}_1$ on the other hand will no longer converge to $\theta$. Generally, there is no guarantee that $T_2$ still renders $Z$ conditionally exogenous. In this case, $\hat{\theta}_2$ converges to some value other than $\theta_0$. However, unless the additional variation that we condition on in $T_2$ compared to $T_1$ is exogenous due to some particularly poor construction of $\tau_2$, $\hat{\theta}_1$ and $\hat{\theta}_2$ still have different probability limits. A specification test using this logic is generally possible for the sufficient complexity of a control function $\tau$.

\section{Point Identification}    \label{sec:id}

Without further parametric restrictions on the outcome or first stage model, at most set identification is possible. When the outcome model is linearly separable in the observables and unobservables, we show how to point-identify the model part relating to the observables \citep{newey2003}. If instead the first stage is monotonous, a control function approach can be used to point-identify average structural functions and thus causal effects with a common support assumption (instead of completeness) \citep{imbens2009}. We construct a control function for the endogenous variation in $A$ while already keeping the endogenous variation in $Z$ fixed.

\subsection{Linearly separable outcome model}    \label{ssec:ls}
An outcome model with linear separability in the treatment and a disturbance is one special case where point identification is possible. With a linearly separated disturbance $\varepsilon$, it is straightforward to represent the exogeneity of instrument $Z$ as mean-independence conditional on common confounders $U$. Assumption \ref{a:ls} fully describes this setting.
\begin{assumption}[Linearly separable outcome  model] \label{a:ls}
There exists some function $k_0 \in \mathcal{L}_2(A)$ 
such that
\begin{align}
Y = Y(A) &= k_0(A) + \varepsilon, & \E{\varepsilon | Z, U} &= \E{\varepsilon | U}. \label{eq:outcome}
\end{align}
\end{assumption}
The conditional moment describes the mean-independence of instruments $Z$ conditional on the unobserved common confounders $U$. 
\begin{theorem}[Identification in linearly separable model] \label{th:ls-id}
Let assumptions \ref{a:riv}.(\ref{a:sutva}/\ref{a:iv-index}/\ref{a:iv-rel}/\ref{a:pl}) and \ref{a:ls} hold.
Any $h \in \mathcal{L}_2(A, T)$ for which $\E{Y | Z} = \E{h(A, T) | Z}$, satisfies $h(A, T) = k_0(A) + \E{\varepsilon | T}$. Consequently, $\theta \coloneqq \int_\mathcal{A} Y(a) \pi(a) \mathrm{d}a = \int_\mathcal{T} \int_\mathcal{A} h(a, t) \pi(a) \mathrm{d}a f(t) \mathrm{d}t$.
\end{theorem}
Theorem \ref{th:ls-id} establishes point identification of the function $k_0$ of the effect of the observable treatment $A$ on outcome $Y$. While we do not make this explicit, $k_0$ may also be a function of the proxies $W$ or other observed covariates $X$. The linear separability in combination with the completeness assumption leads to a straightforward identification in the linearly separable model. Unlike in \cite{tien2022icc}, identification of an average structural function when there are interactions of the observables and unobservable $U$ is much more difficult in  this model where the proxies $W$ may also have a direct effect on treatment $A$.

We could have considered other model specifications or versions of completeness to establish identification \citep{d2011complete}. For now, we leave this exercise for future work.

\subsection{First stage monotonicity}    \label{ssec:mono}
If the outcome model is not linearly separable in treatment and disturbance, monotonicity in the first stage reduced form is an alternative assumption to identify average causal (structural) effects \citep{imbens2009}. If the common confounders $U$ were observed, there would be a simple control function for the endogenous variation in $A$ due to monotonicity. Assumption \ref{a:monotonicity} describes the necessary first stage reduced form monotonicity. 
\begin{assumption}[Monotonicity] \label{a:monotonicity} 
\begin{align}
A &= h(Z, \eta)
\end{align}
\begin{enumerate}
\item $h(Z, \eta)$ is strictly monotonic in $\eta$ with probability 1.  \label{a:mon-1}
\item $\eta$ is a continuously distributed scalar with a strictly increasing conditional CDF $F_{\eta | U}$ on the conditional support of $\eta$.  \label{a:mon-2}
\item $Z \CI \eta \ | \ U$.  \label{a:mon-3}
\end{enumerate}
\end{assumption}
Assumption \ref{a:monotonicity}.\ref{a:mon-1} describes the strict montonicity of $A$ in the disturbance $\eta$. This disturbance $\eta$ is scalar and continuously distributed conditional on the unobservable $U$, with a strictly increasing conditional CDF according to assumption \ref{a:monotonicity}.\ref{a:mon-2}. Jointly, these two assumptions ensure that for any given $Z$, any $A$ is associated with a unique $\eta$. The unobserved confounders $U$ may affect $A$, but only through their effect on $\eta$. This restriction keeps the model monotonous in the unobservables to ensure point identification. Finally, assumption \ref{a:monotonicity}.\ref{a:mon-3} requires full independence of instruments $Z$ and $\eta$ conditional on the common confounders $U$.

The above setup does not immediately help with identification, because the common confounders $U$ are always unobserved. In lemma \ref{l:7}, we establish a few useful facts about the conditional distribution of the scalar disturbance $\eta$ given $T$. Notably, this conditional distribution $F_{\eta | T}$ is also strictly increasing on the conditional support of $\eta$, and unsurprisingly the instruments $Z$ are independent from $\eta$ conditional on $T$. 
\begin{lemma} \label{l:7}
\begin{align*}
F_{\eta | T} &\coloneqq \int_{\mathcal{U}} F_{A | Z, U}(A, Z, u) f_{U | T}(u, T) \dif \mu_U(u)
\end{align*}
 is a strictly increasing CDF on the conditional support of $\eta$, and $Z \CI \eta \ | \ T$.
\end{lemma}
The above lemma \ref{l:7} implies that $F_{\eta | T}(\eta)$ is a one-to-one function of $\eta$ conditional on $T$, just like $F_{\eta | U}(\eta)$ conditional on $U$. This fact is useful, because it is no longer necessary to condition on the unobservable $U$ to identify the endogenous variation in $A$, which is $\eta$, exactly. Instead, if we can identify $F_{\eta | T}(\eta)$, $\eta$ is held fixed as long as $(F_{\eta | T}(\eta), T)$ is held fixed. The remaining difficulty is to identify $F_{\eta | T}(\eta)$. In this regard, theorem \ref{th:8} states that $F_{\eta | T}(\eta)$ is equal to the conditional CDF of $A$ given $Z$. This conditional CDF is defined as $V_T$ in equation \ref{eq:vt}. 
\begin{theorem} \label{th:8}
Let 
\begin{align}
V_T &\coloneqq F_{A | Z; T}(A, Z).   \label{eq:vt}
\end{align}
Under assumption \ref{a:monotonicity}, $V_T = F_{\eta | T}(\eta)$, and 
\begin{align} 
A \CI Y(a) \ | \ (V_T, T), \text{ for all } a \in \mathcal{A}.
\end{align}
\end{theorem}
Theorem \ref{th:8} states that despite the unobservable common confounder $U$, there exist the observable control functions $V_T$ and $T$ conditional on which we retrieve unconfoundedness. Specifically, we retrieve unconfoundedness because all variation in treatment $A$ stems from instruments $Z$ once we condition on $V_T$ and $T$. Fortunately, conditional on $T$, the instruments $Z$ are fully exogenous. So far, our arguments have only been with respect to exogeneity, not yet relevance.

To describe relevance, we use a common support assumption \ref{a:common-support}, with focus on a causal effect of interest 
\begin{align*}
\theta_0 \coloneqq  \int_\mathcal{A} Y(a) \pi(a) \dif \mu_A(a)
.\end{align*}
This common support assumption requires the sufficient relevance of instruments $Z$ for treatment $A$, and sufficient variation in $Z$, both conditional on $T$. In slightly different words, after holding all variation in $Z$ associated with $U$ fixed through $T$, the variation in $Z$ must still be sufficiently rich and relevant for $A$.
\begin{assumption}[Common Support] \label{a:common-support}
For all $a \in \mathcal{A}$, where the contrast function is non-zero ($\pi(a) \neq 0$), the support of $(V_T, T)$ equals the support of $(V_T, T)$ conditional on $A$.
\end{assumption}
With the common support assusmption \ref{a:common-support}, average causal quantities $\theta_0$ in our model are identified under monotonicity (assumption \ref{a:monotonicity}). In theorem \ref{th:average-id}, we explicitly replace the completeness assumption in assumption \ref{a:riv}.\ref{a:iv-rel} by the common support assumption \ref{a:common-support}, which is the correct relevance requirement with a control function. 
\begin{theorem}[Average causal quantity identification]  \label{th:average-id}
Suppose assumption \ref{a:riv}.(\ref{a:sutva}/\ref{a:iv-exog}/\ref{a:iv-index}/\ref{a:pl}) [relaxed IV model], \ref{a:monotonicity} [monotonicity], and \ref{a:common-support} [common support] hold. Then, any $\theta_0 \coloneqq \int_{\mathcal{A}} Y(a) \pi(a) \dif \mu_A(a)$ is identified as 
\begin{align*}
\theta_0 = \int_{\mathcal{V}_T, {\mathcal{Z}^\tau}} \int_{\mathcal{A}} \E{Y | A=a, (V_T, T)=(v_T, t)} \pi(a) \dif \mu_A(a) \dif F_{V_T, T}(v_T, t).
\end{align*}
\end{theorem}
We simply integrate out the control functions $(V_T, T)$ without dependence on treatment $A$ to obtain the causal quantity of interest $\theta_0$. Often, $\theta_0$ will be some form of average treatment effect. 

Other functions of interest than the above (weighted) averages of potential outcomes, e.g. quantile structural functions, are also identified as a consequence of theorem \ref{th:8}, but require corresponding common support assumption which will differ from assumption \ref{a:common-support}.

\section{Linear Model}   \label{sec:linear}
In this section, we explain identification in the common confounding model in linear terms. Apart from the illustrative purpose of linear models, their tractability and ease of use make them attractive. For the common confounding IV approach, the linear model provides useful intuition regarding the relevance and exogeneity assumptions.

First, we describe the model assumptions in linear form. The outcome variable $Y \in \mathbb{R}$ is one-dimensional. For ease of notation, we let $A \in \mathbb{R}$ be one-dimensional too. All other variables $X$ are of some general dimensions $d_X$, i.e. $Z \in \mathbb{R}^{d_Z}$, $W \in \mathbb{R}^{d_W}$, and $U \in \mathbb{R}^{d_U}$. As previously, instruments are called $Z$, proxies $W$, and the unobserved common confounders $U$. 
\begin{align}
Y &= A \beta + U \gamma_{Y} + W \upsilon_{Y} + \varepsilon_Y, & \E{\varepsilon_Y^\intercal (Z, U, W)} &= \mathbf{0}   \label{eq:lin-y} \\  
A &= Z \zeta + U \gamma_{A} + W \upsilon_{A} + \varepsilon_{A}, & \E{\varepsilon_A^{\intercal} (Z, U, W)} &= \mathbf{0}  \label{eq:lin-a}
\end{align}
Equation \ref{eq:lin-y} expresses $Y$ as a linear function of $A$, $U$, $W$ and a disturbance $\varepsilon_Y$. The disturbances $\varepsilon_Y$ is uncorrelated from the instruments $Z$. With this moment equation, the model parameters would be identified under a conditional relevance requirement of $Z$ for $A$ if $U$ were observed. The parameter vector of interest in this model is $\beta$, the effect of treatment $A$ on $Y$. In equation \ref{eq:lin-a}, $A$ is a linearly projected on $Z$, $U$, $W$, and a disturbance $\varepsilon_A$. The $(d_Z \times d_A)$-dimensional parameter matrix $\zeta$ describes the marginal linear effect of $Z$ on $A$. Equation \ref{eq:lin-a} for $A$ is the model's first stage. The conditional relevance requirement of $Z$ for $A$ would simply be $\operatorname{rank}(\zeta) = d_A$, if $U$ were observed. With observable $U$, the model would be sufficiently described at this point to point identify $\beta$. As the common confounders $U$ are never observed, the model requires further assumptions.
\begin{align}
Z &= U \gamma_{Z} + \varepsilon_{Z}, & \E{\varepsilon_Z^\intercal (U, W)} &= \mathbf{0}  \label{eq:lin-z} \\
W &= U \gamma_{W} + \varepsilon_W, & \E{\varepsilon_W^\intercal (U, Z)} &= \mathbf{0}  \label{eq:lin-w}
\end{align}
Equations \ref{eq:lin-z} and \ref{eq:lin-w} imply that all correlation between $Z$ and $W$ stems from the unobserved common confounders $U$. There is no direct effect from either on the other. If there were, we could model this by increasing the dimension of $U$ by the corresponding element of $Z$ or $W$ until $Z$ and $W$ are uncorrelated conditional on $U$. Consequently, the covariance matrix for $Z$ and $W$ is 
\begin{align*}
\Sigma_{ZW} &= \gamma_Z^\intercal \Sigma_U \gamma_W, & \operatorname{rank}(\Sigma_{ZW}) &= d_U.
\end{align*}

\subsection{First projection}
To deal with the unobservedness of $U$, we use two additional projections compared to standard IV. First, we project $W$ on $Z$ as 
\begin{align*}
W &= Z \Sigma_Z^{-1} \Sigma_{ZW} + \tilde{\varepsilon}_W = \gamma_Z^\intercal \Sigma_U \gamma_W + \tilde{\varepsilon}_W, & \E{\tilde{\varepsilon}_W Z} &= \mathbf{0}.
\end{align*}
Then, choose any $C_Z \in \mathbb{R}^{d_Z \times d_U}$ and $C_W \in \mathbb{R}^{d_W \times d_U}$ such that $C_Z C_W^\intercal = \Sigma_{ZW}$ to create the control variable 
\begin{align*}
T &\coloneqq Z \Sigma_Z^{-1} C_Z, & \Sigma_{ZT} &= C_Z, & \Sigma_T &= C_Z^\intercal \Sigma_{Z}^{-1} C_Z.
\end{align*}
The control variable $T$ is a lower-dimensional representation of $Z$, which captures all correlation between $Z$ and $W$.

\subsubsection{Why can this projection help deconfound $Z$?}
All endogeneity in instruments $Z$ stems from their correlation with $U$. Hence, if the linear projection of $U$ on $Z$, $Z \Sigma_Z^{-1} \Sigma_{Z U}$, is unchanged for some linear combinations of $Z$, those linear combinations of $Z$ become excluded instruments. Assumption \ref{a:lin-rank-w} describes the necessary relevance condition for proxies $W$ with respect to $U$. 
\begin{assumption}[Relevance conditios for $\gamma_W$]  \label{a:lin-rank-w}
\begin{align}
d_W &\geq d_U   \label{eq:lin-rank-w}
\end{align}
\end{assumption}
As long as the proxies $W$ are relevant for $U$ in form of assumption \ref{a:lin-rank-w}, the endogeneity inducing unobserved covariance $\Sigma_{Z U}$ is proportional to the observable covariance $\Sigma_{Z W}$. As the rank of $\Sigma_{Z W}$ is $d_U$, keeping the $d_U$-dimensional control variable $T = Z \Sigma_Z^{-1} C_Z$ fixed restores instrument exogeneity in $Z$. Of course, $d_U$ is the minimal sufficient dimension to construct an exogeneity restoring control variable $T$. Larger dimensions of $T$ are similarly admissible choices to restore instrument exogeneity, but might complicate relevance.

\subsection{Second projection}
In the second projection step, the instruments are orthogonalised with respect to the control variable $T$ and thus with respect to the linear projection of $U$ on $Z$. The orthogonalised instrument $\tilde{Z}$ is constructed as
\begin{align*}
\tilde{Z} &= Z \underbrace{\left( I_{d_Z} - \Sigma_Z^{-1} C_Z \left( C_Z^\intercal \Sigma_Z^{-1} C_Z \right)^{-1} C_Z^\intercal \right)}_{\coloneqq M_{ZW}, \text{ note that } C_Z^\intercal M_{ZW} = \mathbf{0}} D_Z,
\end{align*}
for some $D_Z \in \mathbb{R}^{d_U \times d_{\tilde{Z}}}$ such that $\operatorname{rank} \left( \left( I_{d_Z} - \Sigma_T^{-1} \Sigma_{TZ} \right) D_Z \right) = d_{\tilde{Z}} \geq d_A$. Importantly, the orthogonalised instrument $\tilde{Z}$ is uncorrelated with $U$ and $W$:
\begin{align*}
\Sigma_{W \tilde{Z}} = \Sigma_{WZ} M_{ZW} D_Z &= C_W \underbrace{C_Z^\intercal M_{ZW}}_{=0} D_Z = 0, \\
\Sigma_{U \tilde{Z}} = \Sigma_{UZ} M_{ZW} D_Z &= \left( \gamma_W \gamma_W^\intercal \right)^{-1} \gamma_W \Sigma_{WZ} M_{ZW} D_Z \\
&= \left( \gamma_W \gamma_W^\intercal \right)^{-1} \gamma_W C_W \underbrace{C_Z^\intercal M_{ZW}}_{=0} D_Z = 0.
\end{align*}
Finally, using equations \ref{eq:lin-y} and \ref{eq:lin-a} it is easy to show that $\beta$ can be written in a familiar form:
\begin{align}
\beta &= \E{ A^\intercal P_{\tilde{Z}} A}^{-1} \E{A^\intercal P_{\tilde{Z}} Y}, \label{eq:lin-beta} \\
P_{\tilde{Z}} &\coloneqq \tilde{Z} \left( \tilde{Z}^\intercal \tilde{Z} \right)^{-1} \tilde{Z}^\intercal = Z M_{ZW} D_Z \left( D_Z^\intercal \Sigma_Z M_{ZW} D_Z \right)^{-1} D_Z^\intercal M_{ZW}^\intercal Z^\intercal.
\end{align}



\subsection{Enhanced instrument relevance condition}
With the conditional exogeneity of $Z$ sufficiently discussed, the focus shifts to the conditional relevance of $Z$ for $A$. 

Given that at least $d_U$ dimensions of variation in $Z$ are held fixed, there can be spare variation in $Z$ to identify the linear parameter vector $\beta$. A relevant condition for this is $d_Z - d_U \geq d_A$. After keeping the $d_U$-dimensional variation in $\E{W | Z}$ fixed, the expected predicted values of treatment $A$ given instruments $Z$, $\E{Z \zeta | \E{W | Z}}$, must be non-degenerate. The rank condition for any $d_A$ is described in assumption \ref{a:lin-rank-z}.
\begin{assumption}[Rank condition]  \label{a:lin-rank-z}
For some $C_Z \in \mathbb{R}^{d_Z \times d_U}$ and $C_W \in \mathbb{R}^{d_W \times d_U}$ such that $C_Z C_W^\intercal = \Sigma_{ZW}$, and $D_Z \in \mathbb{R}^{d_U \times d_{\tilde{Z}}}$,
\begin{align}
\operatorname{rank}\left(\E{A^\intercal P_{\tilde{Z}} A }\right) = d_A,  \label{eq:lin-rank-z}
\end{align}
where 
\begin{align*}
P_{\tilde{Z}} &\coloneqq Z M_{ZW} D_Z \left( D_Z^\intercal \Sigma_Z M_{ZW} D_Z \right)^{-1} D_Z^\intercal M_{ZW}^\intercal Z^\intercal, \\
M_{ZW} &\coloneqq I_{d_Z} - \Sigma_Z^{-1} C_Z \left( C_Z^\intercal \Sigma_Z^{-1} C_Z \right)^{-1} C_Z^\intercal.
\end{align*}
\end{assumption}
$d_U$ dimensions of variation in $Z$ are lost by conditioning on $\E{W | Z}$. The remaining variation in $Z$ after this conditioning step must be sufficiently relevant for $A$.

\section{Semiparametric estimation}  \label{sec:semiparametric}
With its transparency, the linear model sheds light on the assumptions in IV with common confounders. The intuition for relevance and exogeneity assumptions in the linear model carries on to nonlinear models, specifically the idea of a pre-IV control function. In the linear model, instrument variation is used conditional on the $d_U$-dimensional control variable $T$, which is derived from the correlation of $Z$ and $W$. In nonlinear settings, the control function is a general $\tau \in \mathcal{T} \subseteq \mathcal{L}_2(Z)$, which serves the same purpose: to render the instruments $Z$ conditionally independent from the proxies $W$ and hence from the unobserved common confounders $U$. Exogeneity of the instruments $Z$ conditional on this control function is the desired consequence.

Generally, we would want to identify control function $\tau$ via the integral equation
\begin{align}  
\E{g_0(Z) | W} &= \E{\tau(Z) | W}  \label{eq:nonlinear-bridge}
\end{align}
for some function $g_0 \in \mathcal{G} \subseteq \mathcal{L}_2(Z)$ such that $g_0(Z) \in \mathcal{Z}^g_0$. It is worth considering the behaviour of equation \ref{eq:nonlinear-bridge} for two extreme cases. 

First, consider $Z \CI W$, the unconditional independence of instruments and proxies. In the unconditional independence setting, no information in $Z$ must be held fixed to recover exogeneity. Formally, it follows that $\E{g_0(Z) | W} = \E{g_0(Z)}$. The function $\tau$ does not require its input $Z$, and can simply be set to $\tau = \E{g_0(Z)}$. 

Secondly, one may note that one extreme solution always solves equation \ref{eq:nonlinear-bridge}, that solution being $\tau(Z) = g_0(Z)$. Intuitively, the remaining variation in $Z$ has to be exogenous when the entirety of $Z$ is held fixed, because there simply is no remaining variation in $Z$.

Any interesting solution will be non-extreme, in the sense that $Z \CI W \ | \ U$ for unobservable $U$, and $\tau \neq g_0$. 
The relevance of that spare information for treatment $A$ is generally testable, as in standard IV conditional on covariates.

Let $\mathcal{T}_\text{exog}$ be the set of solutions to \ref{eq:nonlinear-bridge} for a given function $g_0 \in \mathcal{G}$ , i.e.
\begin{align}
\mathcal{T}_\text{exog} &= \left\{ \tau \in \mathcal{T}: \E{g_0(Z) | W} = \E{\tau(Z) | W} \right\}.  \label{eq:t-exog}
\end{align}
As discussed, the $\tau = g_0$ is always in $\mathcal{T}_\text{exog}$, but is incompatible with instrument relevance. In a sense, $\mathcal{T}_\text{exog}$ is the set of exogeneity restoring control functions. A further relevance assumption will lead to the set of valid control functions $\mathcal{T}_\text{valid} \subset \mathcal{T}_\text{exog}$. As the required relevance assumptions always depend on the estimand, we will leave this general on purpose, for now.

\subsection{Continuous linear functionals in $\tau$}  \label{ssec:linear-in-tau}

Subsequently, we broadly follow the semiparametric estimation setup in \cite{bennett2022}. Consider a general semiparametric model where the target estimand $\theta_0$ is given by an unconditional moment as
\begin{align}
\theta_0 &= \E{m(O; \tau_0)}, \ \forall \tau_0 \in \mathcal{T}_\text{valid},  \label{eq:nonlinear-estimand}
\end{align}
where $m$ is a known function, and $O = (Y, A, Z, W)$ a vector of all observable variables. Non-uniqueness of valid  control functions $\tau_0 \in \mathcal{T}_\text{valid}$ is not a problem here, because the target estimand $\theta_0$ remains uniquely identified (as in overidentified IV models).

To facilitate inference, we assume that $\tau \mapsto \E{m(O; \tau)}$ is a continuous linear functional over $\mathcal{T}$. Continuity and linearity will permit using the Riesz representation theorem \citep{luenberger1997optimization}. Specifically, there exists a Riesz representer $\alpha$ of our functional such that
\begin{align}
\E{m(O; \tau)} &= \E{\alpha(Z) \tau(Z)} \ \forall \tau \in \mathcal{T}.
\end{align}
Then, we can use ideas from automated debiased machine learning to find a quantity that we can use in place of the Riesz representer to debias the plug-in estimator $\mathbb{E}_n\left[ m(O; \hat{\tau}) \right]$ and retrieve $\sqrt{n}$-convergence along with asymptotic normality \citep{autodml2022}. These desired asymptotic properties require functional strong identification \ref{d:nonlinear-strong-id}. 

Operator language facilitates writing the strong identification definition \ref{d:nonlinear-strong-id}, so let $P_{\mathcal{L}_2(W)}^{Z, \mathcal{T}}(g): \mathcal{T} \mapsto \mathcal{L}_2(W)$ be the linear operator given by $[P_{\mathcal{L}_2(W)}^{Z, \mathcal{T}} \tau](W) = \E{\tau(Z) | W}$.
Let $P^{\mathcal{L}_2(W)}_{Z, \mathcal{T}}: \mathcal{L}_2(Z) \mapsto \mathcal{T}$ be the adjoint of $P_{\mathcal{L}_2(W)}^{Z, \mathcal{T}}$ given by $[P^{\mathcal{L}_2(W)}_{Z, \mathcal{T}} q](Z) = \Pi_{\mathcal{T}} \E{q(W) | Z} = \Pi_\mathcal{T} \left[ q(W) | Z \right]$ for any $q \in \mathcal{L}_2(W)$ and projection operator $\Pi_\mathcal{T}[.|Z]$ onto $\mathcal{T}$. 

\begin{definition}[Strong identification of $\theta_0$]  \label{d:nonlinear-strong-id}
$\theta_0$ is strongly identified if $\alpha \in \mathcal{N}^\perp(P^{\mathcal{L}_2(W)}_{Z, \mathcal{T}} P_{\mathcal{L}_2(W)}^{Z, \mathcal{T}})$, i.e.
\begin{align}
\Xi_0 &\neq \emptyset, & \text{ where } & &  \Xi_0 &\coloneqq \arg \min_{\xi \in \mathcal{T}} \left( \frac{1}{2} \E{\E{\xi(Z) | W}^2} - \E{m(O;\xi)} \right) \\
& & & & &= \left\{ \xi \in \mathcal{T}: P^{\mathcal{L}_2(W)}_{Z, \mathcal{T}} P_{\mathcal{L}_2(W)}^{Z, \mathcal{T}} \xi = \alpha \right\}.
\end{align}
\end{definition}

\begin{theorem}[Strong identification of $\theta_0$] \label{th:nonlinear-strong-id}
Suppose $\Pi_\mathcal{T} \left[ q(W) | Z \right] = \Pi_\mathcal{T} \left[ \E{q(W) | U} | Z \right]$ for any $q \in \mathcal{L}_2(W)$ and $\mathcal{T} \subseteq \mathcal{L}_2(Z)$, and assumption \ref{a:riv}.\ref{a:pl-rel} (completeness of $W$ for $U$) hold. Also, suppose that instruments $Z$ satisfy relevance conditional on any $\tau_0 \in \mathcal{T}_\text{valid} \neq \emptyset$ such that $\theta_0 = \E{m(O; \tau_0)}$, where $\tau \mapsto \E{m(O;\tau)}$ is a continuous and linear functional over $\mathcal{T}$. Then, $\theta_0$ is strongly identified with 
\begin{align}
\alpha \in \mathcal{N}^\perp(P^{\mathcal{L}_2(U)}_{Z, \mathcal{T}} P_{\mathcal{L}_2(U)}^{Z, \mathcal{T}}) = \mathcal{N}^\perp(P^{\mathcal{L}_2(W)}_{Z, \mathcal{T}} P_{\mathcal{L}_2(W)}^{Z, \mathcal{T}}).
\end{align}
\end{theorem}

\begin{definition}[Debiased representation of $\theta_0$]  \label{d:nonlinear-debiasing-nuisance}
The debiasing nuisance is
\begin{align}
q_0(W) &= \E{\xi_0(Z) | W} \ \forall \xi_0 \in \Xi_0.
\end{align}
The set of valid debiasing nuisances is 
\begin{align}
\mathcal{Q}_0 &\coloneqq \left\{ q \in \mathcal{L}_2(W): \left[ P^{\mathcal{L}_2(W)}_{Z, \mathcal{T}} q \right](Z) = \alpha(Z) \right\}.
\end{align}
The debiased (Neyman orthogonal) moment $\psi$ with $\theta_0 = \E{\psi(O; \tau_0, q_0)}$ for any $\tau_0 \in \mathcal{T}_\text{valid}$ is
\begin{align}
\psi(O; \tau, q) &\coloneqq m(O; \tau) + q(W) \left( g_0(Z) - \tau(Z) \right)
\end{align}
\end{definition}

\begin{theorem}  \label{th:nonlinear-error}
Suppose $\theta_0$ is strongly identified. Then, $\theta_0 = \E{\psi(O; \tau_0, q_0)}$ for any $\tau_0 \in \mathcal{T}_\text{valid}$ and $q_0 \in \mathcal{Q}_0$. For any $\tau \in \mathcal{T}$, $q \in \mathcal{L}_2(W)$, $\tau_0 \in \mathcal{T}_\text{valid}$, $q_0 \in \mathcal{Q}_0$,
\begin{align*}
\E{\psi(O; \tau, q)} - \theta_0 &= - \E{(q(W) - q_0(W)) (\tau(Z) - \tau_0(Z))}.
\end{align*}
This implies error bound 
\begin{align}
\left| \E{\psi(O; \tau, q)} - \theta_0 \right| 
&= \left| \left\langle P_{\mathcal{L}_2(W)}^{Z, \mathcal{T}} \left( \tau - \tau_0 \right), q - q_0 \right\rangle \right| 
= \left| \left\langle   \tau - \tau_0 , P^{\mathcal{L}_2(W)}_{Z, \mathcal{T}} \left( q - q_0 \right) \right\rangle \right| \nonumber \\
&\leq \min \left\{ \left\lVert P_{\mathcal{L}_2(W)}^{Z, \mathcal{T}} \left( \tau - \tau_0 \right) \right\rVert_2 \left\lVert q - q_0 \right\rVert_2, \  \left\lVert \tau - \tau_0 \right\rVert_2 \left\lVert P^{\mathcal{L}_2(W)}_{Z, \mathcal{T}} \left( q - q_0 \right) \right\rVert_2  \right\}
\end{align}
Double robustness is satisfied as $\theta_0 = \E{\psi(O; \tau_0, q)} = \E{\psi(O; \tau, q_0)}$ for any $\tau_0 \in \mathcal{T}_\text{valid}$, $\tau \in \mathcal{T}$, $q_0 \in \mathcal{Q}_0$, $q \in \mathcal{L}_2(W)$. \\
Neyman orthogonality is satisfied as $\tau_0 \in \mathcal{T}_\text{valid}$ and $q_0 \in \mathcal{Q}_0$ if and only if
\begin{align*}
0 &= \left. \frac{\partial}{\partial t}  \E{\psi(O; \tau_0 + t \tau, q_0)} \right|_{t=0} = \left. \frac{\partial}{\partial t}  \E{\psi(O; \tau, q_0 + t q)} \right|_{t=0} \ \forall \tau \in \mathcal{T}, q \in \mathcal{L}_2(W).
\end{align*}
\end{theorem}

If we let $\mathcal{T} = \mathcal{L}_2(Z)$ as we usually would to allow for general $\tau$, the condition $P^{\mathcal{L}_2(W)}_{Z, \mathcal{T}} P_{\mathcal{L}_2(W)}^{Z, \mathcal{T}} \xi_0 = \alpha$ for any $\xi \in \Xi_0$ simplifies to $\E{\E{\xi_0(Z) | W} | Z} = \alpha(Z)$. The Riesz representer is therefore closely related to the debiasing nuisance $\xi_0 \in \Xi_0$. More precisely, any $\xi_0 \in \Xi_0$ is equal to the Riesz representer after two projection steps. This double projection property reflects our intuition from the linear model in section \ref{sec:linear}, where similarly two projection steps were required to restore instrument exogeneity: First proxies $W$ were projected onto instruments $Z$ to retrieve control $T$, then $Z$ was orthogonalised with respect to $T$.

The simpler estimation problem as a result of \ref{th:nonlinear-strong-id} has the additional advantage that estimation error in $\tau$, identified by a possibly ill-posed inverse problem, no longer determines asymptotic behaviour of an estimator dependent on it. Instead, asymptotic behaviour is determined by the estimation error of the projection of $\tau$ on $W$. The projection eliminates the problem from ill-posedness, which can imply a discontinuous dependence of solutions to equation \ref{eq:nonlinear-bridge} on the data distribution \citep{carrasco2007linear}.

\subsection{Consecutive debiasing of more general functionals of the treatment} \label{ssec:consecutive-debiasing}
Most functionals of interest will not be linear functionals of the debiasing function $\tau$. However, semiparametric estimation of such functionals often remains possible. For example, let the target estimand $\theta_0$ be given by
\begin{align}
\theta_0 &= \E{m(O; k_0(\tau_0, g_0))}, \ \forall k_0(\tau_0, g_0) \in \mathcal{K}_0, \ \tau_0(g_0) \in \mathcal{T}_\text{valid}.
\end{align}
Let the set of functions $\mathcal{K}_0(\tau, g)$ be defined as the set of all functions $k_0(\tau, g): \mathcal{A} \times \mathcal{Z}^\tau \mapsto \mathbb{R}$ that satisfy the first order (Fredholm) integral equation
\begin{align}
\E{k_0(A; \tau, g) | Z} &= g(Z) - \tau(Z) \text{ for any given } \tau \in \mathcal{T} \text{ and } g \in \mathcal{G},  \label{eq:consecutive-debiasing-1}
\end{align}
where the dependence of $k_0$ on is $\tau$ and $g$ is made explicit.
Then, let $\mathcal{K}_0 \coloneqq \mathcal{K}_0(\tau_0, g_0)$ for any $\tau_0 \in \mathcal{T}_\text{valid}$ and a given $g_0$ be the set of functions $k_0$ that identify $\theta_0$.
Remember that $\mathcal{T}_\text{valid} \subset \mathcal{T}_\text{exog}$, where the set $\mathcal{T}_\text{exog}$ itself is defined all solutions to the first order integral equation \ref{eq:t-exog}.
In many cases, instead of being a known function, $g_0$ will also need to be estimated, e.g. if $g_0(Z) = \E{g_Y(Y)|Z}$ for some known function $g_Y \in \mathcal{L}_2(Y)$. 

\subsubsection{Common nuisance of type \ref{eq:consecutive-debiasing-1} for causal questions}
With averages over outcomes being the most common estimand, it is worth understanding \ref{eq:consecutive-debiasing-1} better. Indeed, consider $g_0(Z) = \E{Y|Z}$ and note that $\E{Y|Z} = \E{\mathbb{E}^{C(U | \tau_0(Z))} \left[ Y | A, \tau_0(Z) \right] | Z}$, where $\mathbb{E}^{C(U | \tau_0(Z))} \left[ Y | A, \tau_0(Z) \right] \coloneqq \int_\mathcal{U} \E{Y | A, \tau_0(Z), U=u} d(F_{U | \tau_0(Z)}(u | \tau_0(Z)))$, due to the conditional completeness of $Z$ for $(A, \tau(Z))$ and $\tau_0(Z) = \E{Y | \tau_0(Z)}$, such that
\begin{align}
g_0(Z) - \tau_0(Z) &= \E{\mathbb{E}^{C(U | \tau_0(Z))} \left[ Y | A, \tau_0(Z) \right] - \E{Y | \tau_0(Z)} | Z}.  \label{eq:ex-estimand}
\end{align}

Now, we perform a hypothetical exercise. Decompose variation in $Z$ into two independent components $\tau_0(Z)$ and $\tilde{\tau_0}(Z)$, i.e. $\tau_0(Z) \CI \tilde{\tau_0}(Z)$. Neither must $\tilde{\tau}$ be estimated in practice, nor is it is not necessary to be able to reconstruct $Z$ from $\tau_0(Z)$ and $\tilde{\tau_0}(Z)$ exactly. In this hypothetical exercise, we should ensure that no spare variation in $Z$ is predictive of $\mathbb{E}^{C(U | \tau_0(Z))} \left[ Y | A, \tau_0(Z) \right]$, i.e. 
\begin{align}
\E{\mathbb{E}^{C(U | \tau_0(Z))} \left[ Y | A, \tau_0(Z) \right] | Z} = \E{\mathbb{E}^{C(U | \tau_0(Z))} \left[ Y | A, \tau_0(Z) \right] | \tau_0(Z), \tilde{\tau_0}(Z) }.
\end{align}
In other words, $\tilde{\tau_0}(Z)$ should capture all variation in $Z$ that predicts $\mathbb{E}^{C(U | \tau_0(Z))} \left[ Y | A, \tau_0(Z) \right]$, while being independent of $\tau_0(Z)$.

Note that any $k_0(\tau_0, g_0) \in \mathcal{K}_0(\tau_0, g_0)$ that satisfies $\E{k_0(A; \tau_0, g_0) | Z} = g_0(Z) - \tau_0(Z)$ will also satisfy
\begin{align}
\E{k_0(A) | \tilde{\tau_0}(Z)} &= \E{g_0(Z) - \tau_0(Z) | \tilde{\tau_0}(Z)}.  \label{eq:also-satisfies}
\end{align}
Using the independence of $\tau_0(Z)$ and $\tilde{\tau_0}(Z)$, the right-hand side of the above equation reduces to a form of counterfactual mean deviation
\begin{align}
\mathbb{E}^{C(U, \tau_0(Z))} \left[ Y - c \mid A \right] &\coloneqq \int_{\mathcal{Z}^\tau} \left( \mathbb{E}^{C(U | \tau_0(Z))} \left[ Y | A, T=t \right] - \E{Y | T=t} \right) d(F_{\tau_0(Z)}(t)) \\
&= \int_{\mathcal{Z}^\tau \times \mathcal{U}} \left(  \E{Y | A, T=t, U=u} - \E{Y | T=t, U=u} \right) d(F_{\tau_0(Z), U}(t, u)),
\end{align}
where $c \in \mathbb{R}$ is some constant.
Indeed, this counterfactual mean deviation integrates out all endogenous variation independently from treatment $A$. This exact desired behaviour allows the identification of typical causal effects, like average treatment effects, making $\mathbb{E}^{C(U, \tau_0(Z))} \left[ Y - c \mid A \right]$ an interesting nuisance function to identify causal effects. Given that any $k_0$ that satisfies \ref{eq:consecutive-debiasing-1} also satisfies \ref{eq:also-satisfies}, the not necessarily unique $k_0$ to satisfy \ref{eq:consecutive-debiasing-1} is $k_0(A) = \mathbb{E}^{C(U, \tau_0(Z))} \left[ Y - c \mid A \right]$. Despite identification only up to the constant $c \in \mathbb{R}$, causal effects based on this nuisance function are often unique, e.g. if $m_0(O; k) = k_0(a') - k_0(a)$ for two distinct $a, a' \in \mathcal{A}$.

\subsubsection{Consecutive debiasing}
Consecutive debiasing will allow us to address each first order bias in successive order. Again, assume that $k \mapsto \E{m(O; k)}$ for $k \in \mathcal{K}$ is a continuous linear functional over $\mathcal{K}$, such that by the Riesz representation theorem
\begin{align}
\E{m_0(O; k)} &= \E{\alpha_{k, 0}(A) k(A)} \ \forall k \in \mathcal{K}.
\end{align}

\begin{assumption}[Strong instrument relevance]  \label{a:iv-strong-relevance}
$\alpha_{k, 0} \in \mathcal{N}^\perp \left( P_{A, \mathcal{K}}^{\mathcal{L}_2(Z)} P^{A, \mathcal{K}}_{\mathcal{L}_2(Z)} \right)$, i.e. 
\begin{align}
\Xi_{k, 0} &\neq \emptyset, & \text{ where } & &  \Xi_{k, 0} &\coloneqq \arg \min_{\xi_k \in \mathcal{K}} \left( \frac{1}{2} \E{\E{\xi_k(A) | Z}^2} - \E{m_0(O;\xi_k)} \right) \\
& & & & &= \left\{ \xi_k \in \mathcal{K}: P_{A, \mathcal{K}}^{\mathcal{L}_2(Z)} P^{A, \mathcal{K}}_{\mathcal{L}_2(Z)} \xi_k = \alpha_{k, 0} \right\}.
\end{align}
\end{assumption}
The strong instrument relevance assumption \ref{a:iv-strong-relevance} imposes a stronger regularity condition on the smoothness of $\alpha_{k, 0}$ than needed for the simple identification of $\theta_0$, because $\theta_0 = \E{q_{k, 0}(Z) k_0(A; \tau_0, g_0)}$ where $\E{q_{k, 0}(Z)} = \alpha_{k, 0}(A)$ only requires $\alpha_{k, 0} \in \mathcal{R} \left( P^{A, \mathcal{K}}_{\mathcal{L}_2(Z)} \right)$ as opposed to $\alpha_{k, 0} \in \mathcal{R} \left( P^{A, \mathcal{K}}_{\mathcal{L}_2(Z)} P_{A, \mathcal{K}}^{\mathcal{L}_2(Z))}  \right) \subseteq \mathcal{R} \left( P^{A, \mathcal{K}}_{\mathcal{L}_2(Z)} \right)$. The strong instrument relevance in assumption \ref{a:iv-strong-relevance} ensures $\sqrt{n}$-estimability of the target functional while allowing for ill-posedness in the inverse problem \ref{eq:consecutive-debiasing-1} which defines any valid $k_0$ \citep{bennett2022}. 

Strong instrument relevance as stated in \ref{a:iv-strong-relevance} requires the remaining variation in $Z$ to be relevant enough with respect to outcome-relevant variation in treatment $A$ after integrating out any of the outcome-relevant variation in $A$ associated with the endogenous variation $\tau(Z)$.
Hence, the function space $\mathcal{K}$ contains functions of $A$, where any outcome-relevant variation associated with the possibly endogenous instrument variation $\tau(Z)$ is explicitly subtracted and integrated out without dependence on treatment $A$. We can understand the function spaces $\mathcal{K}$ and $\mathcal{T}$ (as well as the $\mathcal{K}_0$ and $\mathcal{T}_\text{exog}$) as jointly defined:
\begin{align}
\mathcal{K} \times \mathcal{T} &= \left\{ (k, \tau): \left( \E{k(A) | Z} = g(Z) - \tau(Z) \right) \land \left( \E{g(Z) | W} = \E{\tau(Z) | W} \right) \ \forall g \in \mathcal{G} \right\} \label{eq:joint-def} \\
\mathcal{K}_0 \times \mathcal{T}_\text{exog} &= \left\{ (k, \tau): \left( \E{k(A) | Z} = g_0(Z) - \tau(Z) \right) \land \left( \E{g_0(Z) | W} = \E{\tau(Z) | W} \right) \right\}
\end{align}
The joint definition in \ref{eq:joint-def} illustrates the close relation between the function spaces $\mathcal{K}$ and $\mathcal{T}$.
Assumptions like \ref{a:iv-strong-relevance} are typically interpreted as smoothness restrictions on the Riesz representer $\alpha_{k, 0}$. Here, smoothness means that the nuisance $k \in \mathcal{K}$ may only be a function of treatment $A$ to the degree that remaining instrument variation after integrating out $\tau(Z)$ remains strongly relevant, hence the assumption's name \emph{strong instrument relevance}. 
Consequently, strong instrument relevance as stated in \ref{a:iv-strong-relevance} always implies a restriction on the complexity of $\mathcal{T}$. If $\mathcal{T}$ is too encompassing, $\alpha_{k, 0} \in \mathcal{N}^\perp \left( P_{A, \mathcal{K}}^{\mathcal{L}_2(Z)} P^{A, \mathcal{K}}_{\mathcal{L}_2(Z)} \right)$ is less likely to be satisfied. In the extreme, $\mathcal{G} =  \mathcal{T}$ implies $\E{k(A; \tau, g) | Z} = 0$, requiring $\alpha_{k, 0} = 0$ under strong instrument relevance \ref{a:iv-strong-relevance}, which means that $\theta_0$ is identifiable only if $k = 0$. Clearly, zero nuisances can only correspond to trivial identification problems. For any non-trivial identification problem, the strong instrument relevance assumption \ref{a:iv-strong-relevance} always implies some restriction on the complexity of $\mathcal{T}$.
A violation of strong instrument relevance \ref{a:iv-strong-relevance} can be interpreted in two ways: The first interpretation is that remaining instrument variation after integrating out $\tau(Z)$ is not sufficiently relevant for treatment $A$, the alternative interpretation is that too much variation in instruments $Z$ is associated with the common confounders $U$ and thus proxies $W$. Both of these interpretations, however, are different ways to phrase the same problem. Instruments are not strongly relevant, as required by assumption \ref{a:iv-strong-relevance}.

With assumption \ref{a:iv-strong-relevance}, it is possible to perform the first debiasing step and construct a first step debiased moment as
\begin{align}
m_1(O; k, \tau, g, q_k) &= m_0(O; k) + q_k(Z) \left( g(Z) - \tau(Z) - k(A; \tau, g) \right) \\
\mathcal{Q}_k &\coloneqq \left\{ q_k \in \mathcal{L}_2(Z): q_k(Z) = \E{\xi_k(A) | Z} \ \forall \xi_k \in \mathcal{K} \right\}  \\
\mathcal{Q}_{k, 0} &\coloneqq \left\{ q_k \in \mathcal{L}_2(Z): q_k(Z) = \E{\xi_{k, 0}(A) | Z} \ \forall \xi_{k, 0} \in \Xi_{k, 0} \right\}
\end{align}


Having taken care of the debiasing by $k$, we direct our attention to debiasing by $\tau$. The functional $\tau \mapsto \E{m_1(O; k, \tau, g, q_k)}$ is indeed straigtforwardly linear and continuous in $\tau$. Any dependence on $\tau$ trough $k$ cancels out in expectation since $m_1$ is already debiased with respect to $k$. Thus, $m_1$ only depends on $\tau$ via $- q_k(Z) \tau(Z)$. The Riesz representer of $\tau \mapsto \E{m_1(O; k, \tau, g, q_k)}$ is the same as the Riesz representer of the simpler functional $\tau \mapsto \E{q_k(Z) \tau(Z)}$. Consequently, the results in section \ref{ssec:linear-in-tau} apply, such that 
\begin{align*}
\E{q_k(Z) \tau(Z)} = \E{\alpha_{\tau}(Z; q_k) \tau(Z)}
\end{align*}
for some $\alpha_{\tau}(q_k) \in \mathcal{N}^\perp \left( P_{Z, \mathcal{T}}^{\mathcal{L}_2(U)} P^{Z, \mathcal{T}}_{\mathcal{L}_2(U)} \right) = \mathcal{N}^\perp \left( P_{Z, \mathcal{T}}^{\mathcal{L}_2(W)} P^{Z, \mathcal{T}}_{\mathcal{L}_2(W)} \right)$. $\E{q_k(Z) \tau(Z)}$ already appears to be in some sort of Riesz representer form. However, the smoothness of $q_k$ is not restricted, while the smoothness of $\alpha_\tau(q_k)$ implied by theorem \ref{th:nonlinear-strong-id} is essential to achieve robustness with respect to misestimation in $\tau$. As $\alpha_\tau(q_k)$ depends on $q_k$, let $\alpha_{\tau, 0} \coloneqq \alpha_\tau(q_{k, 0})$.

The second step debiased moment is
\begin{align}
m_2(O; k, \tau, g, q_k, q_\tau) &= m_1(O; k, \tau, g, q_k) + q_\tau(W; q_k) \left( \tau(Z) - g(Z) \right) \\
\mathcal{Q}_\tau(q_k) &\coloneqq \left\{ q_\tau \in \mathcal{L}_2(W): q_\tau = \E{\xi_\tau(Z; q_k) | W}, \ \forall \xi_\tau(q_k) \in \mathcal{T}, q_k \in \mathcal{Q}_k \right\} \\
\mathcal{Q}_{\tau, 0}(q_k) &\coloneqq \left\{ q_\tau \in \mathcal{L}_2(W): q_\tau = \E{\xi_{\tau, 0}(Z; q_k) | W}, \ \forall \xi_{\tau, 0}(q_{k, 0}) \in \Xi_{\tau, 0}(q_k), q_k \in \mathcal{Q}_k \right\} \\
\Xi_{\tau, 0}(q_k) &\coloneqq \arg \min_{\xi_\tau(q_k) \in \mathcal{T}} \left( \frac{1}{2} \E{\E{\xi_\tau(Z; q_k) | W}^2} - \E{q_k(Z) \xi_\tau(Z; q_k)} \right).
\end{align}

The only remaining task is debiasing by $g$, which requires defining $g_0$ in some form. Typically, we are interested in some sorts of averages over the outcome variable $Y$, hence for some known function $g_Y \in \mathcal{L}_2(Y)$ let
\begin{align}
g_0(Z) &\coloneqq \E{g_Y(Y) | Z}.
\end{align}
Again, the Riesz representer of $g \mapsto \E{m_2(O; k, \tau, g, q_k, q_\tau)}$ is the same as the Riesz representer of a simpler functional, in this case $g \mapsto \E{\left( q_k(Z) - q_\tau(W; q_k) \right) g(Z)}$. While this functional is not immediately in Riesz representer form, it only takes an inner expectation to show that $\E{\left( q_k(Z) - q_\tau(W; q_k) \right) g(Z)} = \E{\alpha_g(Z; q_k, q_\tau) g(Z)}$ with $\alpha_g(Z; q_k, q_\tau) = \E{q_k(Z) - q_\tau(W; q_k) | Z}$. Hence, the third step debiased moment is 
\begin{align}
m_3(O; k, \tau, g, q_k, q_\tau, \alpha_g) &= m_2(O; k, \tau, g, q_k, q_\tau) + \alpha_g(Z; q_k, q_\tau) \left( g_Y(Y) - g(Z) \right) \\
\Xi_{g, 0}(q_k, q_\tau) &\coloneqq \arg \min_{\xi_g(q_k, q_\tau) \in \mathcal{L}_2(Z)} \left( \frac{1}{2} \E{\xi_g(Z; q_k, q_\tau)^2} - \E{\left( q_k(Z) - q_\tau(W; q_k) \right) \xi_g(Z; q_k, q_\tau)} \right).
\end{align}

The third step debiased moment can be used construct a doubly robust estimator. First, note that each of the moments $m_j$ for $j \in \{ 1, 2, 3 \}$ identify $\theta_0$.
\begin{lemma}
Assume $\theta_0 = \E{m_0(O; k_0(\tau_0, g_0))}$ and \ref{a:iv-strong-relevance}. Then,
\begin{align*}
\theta_0 = \E{m_0(O; k_0)}  &= \E{m_1(O; k_0, \tau_0, g_0, q_{k, 0})} \\
 &= \E{m_2(O; k_0, \tau_0, g_0, q_{k, 0}, q_{\tau, 0})} \\
 &= \E{m_3(O; k_0, \tau_0, g_0, q_{k, 0}, q_{\tau, 0}, \alpha_{g, 0})}.
\end{align*}
\end{lemma}
Each of $m_j$ for $j \in \{ 1, 2, 3 \}$ identify $\theta_0$ because in each consecutive debiasing step a conditionally mean-zero term is added to the previous moment.

More importantly, the final debiased moment $m_3$ allows for a robust error decomposition.
\begin{theorem}[Robust error decomposition]  \label{th:error-decomposition}
Suppose $\theta_0 = \E{m_0(O; k_0(\tau_0, g_0))}$, assumption \ref{a:iv-strong-relevance}, and the conditions for theorem \ref{th:nonlinear-strong-id} hold for $q = q_\tau$. Then,
\begin{align*}
\E{m_3(O; k, \tau, g, q_k, q_\tau, \alpha_g)} - \theta_0 
&= \E{\left( q_{k, 0}(Z) - q_k(Z) \right) \left( k(A; \tau, g) - k_0(A; \tau_0, g_0) \right)} \\
& \ \ \ \ + \E{\left( q_{\tau, 0}(W; q_k) - q_\tau(W; q_k) \right) \left( \tau_0(Z; g_0) - \tau(Z; g) \right)} \\
& \ \ \ \ + \E{\left( \alpha_{g, 0}(Z; q_k, q_\tau) - \alpha_g(Z; q_k, q_\tau) \right) \left( g(Z) - g_0(Z) \right)}.
\end{align*}
\end{theorem}
This robust error decomposition admits double robustness in corollary \ref{co:double-robustness} and Neyman orthogonality in corollary \ref{co:neyman-orthogonality}. Double robustness here means that $\theta_0$ is still correctly estimated when \emph{either} the nuisances $(k, \tau, g)$ \emph{or} the debiasing functions $q_k, q_\tau, \alpha_g$, which are closely related to Riesz representers, are estimated correctly. Interestingly, the debiasing functions depend on each other in opposite order compared to the ordinary nuisances. Specifically, while nuisance $k$ depends on $(g, \tau)$, and nuisance $\tau$ depends on $g$, the debiasing function (Riesz representer) $\alpha_g$ depends on $(q_k, q_\tau)$, and the debiasing function $q_\tau$ depends on $q_k$. 
\begin{corollary}[Double robustness]  \label{co:double-robustness}
Suppose $\theta_0 = \E{m_0(O; k_0(\tau_0, g_0))}$, assumption \ref{a:iv-strong-relevance}, and the conditions for theorem \ref{th:nonlinear-strong-id} hold for $q = q_\tau$. Then, if \\
$\left( (k = k_0(\tau_0, g_0)) \lor (q_k = q_{k, 0}) \right) \land \left( (\tau = \tau_0(g_0)) \lor (q_\tau = q_{\tau, 0}) \right) \land \left( (g = g_0) \lor (\alpha_g = \alpha_{g, 0}) \right)$,
\begin{align*}
\theta_0 
 &= \E{m_3(O; k, \tau, g, q_k, q_\tau, \alpha_g)}.
\end{align*}
\end{corollary}
A special form of double robustness and Neyman orthogonality holds in the sequential debiasing setting, where at each debiasing step, either the ordinary nuisance or the debiasing function must be correct, but never both. For example, when $\tau \notin \mathcal{T}_\text{valid}$, $q_\tau \in \mathcal{Q}_{\tau, 0}(q_k)$ is necessary to still correctly estimate $\theta_0$. However, it suffices that the ordinary nuisance is correct with $k \in \mathcal{K}_0$, while the debiasing function at the same stage may be incorrect as $q_k \notin \mathcal{Q}_{k, 0}$. Specifically, double robustness and Neyman orthogonality is not restricted to all ordinary nuisances being jointly correct when a debiasing function is incorrect. Instead, the corresponding ordinary nuisance for the incorrect debiasing function must be correct. It suffices if at all other steps in the sequential debiasing that a given nuisance function affects either the ordinary nuisance or debiasing function is correct. This logic is captured by the conditions in corollary \ref{co:neyman-orthogonality}.

\begin{corollary}[Neyman orthogonality]  \label{co:neyman-orthogonality}
Suppose $\theta_0 = \E{m_0(O; k_0(\tau_0, g_0))}$, assumption \ref{a:iv-strong-relevance}, and the conditions for theorem \ref{th:nonlinear-strong-id} hold for $q = q_\tau$. Let $k \in \mathcal{K}$, $\tau \in \mathcal{T}$, $g \in \mathcal{G}$, $k_0 \in \mathcal{K}_0$, $\tau_0 \in \mathcal{T}_\text{valid}$, $g_0(Z) = \E{g_Y(Y) | Z}$, and $q_k \in \mathcal{Q}_k$, $q_\tau \in \mathcal{Q}_\tau(q_k)$, $\alpha_g \in \mathcal{L}_2(Z)$, $q_{k, 0} \in \mathcal{Q}_{k, 0}$, $q_{\tau, 0} \in \mathcal{Q}_{\tau, 0}(q_k)$, $\alpha_{g, 0} \in \Xi_{g, 0}(q_k, q_\tau)$. Let $\partial_g f(X; g, h) \coloneqq \left. \frac{\partial}{\partial t} f(X; g_0 + t g, h) \right|_{t=0}$ for any functions $f$, $g$, and $h$. Then:

Under no further conditions,
\begin{align*}
0 &=\left. \frac{\partial}{\partial t}  \E{m_3(O; k_0 + t k, \tau, g, q_k, q_{\tau, 0}, \alpha_g)} \right|_{t=0} = \left( q_{k, 0}(Z) - q_k(Z) \right).
\end{align*}
If $\left( \left( \partial_\tau k(\tau, g) = 0 \right) \lor \left( q_k = q_{k, 0} \right) \right)$,
\begin{align*}
0 &=\left. \frac{\partial}{\partial t}  \E{m_3(O; k, \tau_0 + t \tau, g, q_k, q_{\tau, 0}, \alpha_g)} \right|_{t=0}.
\end{align*}
If $\left( (\partial_g k(\tau, g) = 0) \lor (q_k = q_{k, 0}) \right) \land \left( (\partial_g \tau(g) = 0) \lor (q_\tau = q_{\tau, 0}) \right)$,
\begin{align*}
0 &=\left. \frac{\partial}{\partial t}  \E{m_3(O; k, \tau, g_0 + t g, q_k, q_\tau, \alpha_{g, 0})} \right|_{t=0}.
\end{align*}
If $\left( (\tau = \tau_0(g_0)) \lor (q_\tau = q_{\tau, 0}) \right) \land \left( (g = g_0) \lor (\alpha_g = \alpha_{g, 0}) \right)$,
\begin{align*}
0 &=\left. \frac{\partial}{\partial t}  \E{m_3(O; k_0(\tau_0, g_0), \tau, g, q_{k, 0} + t q_k, q_\tau, \alpha_g)} \right|_{t=0}.
\end{align*}
If $\left( (g = g_0) \lor (\alpha_g = \alpha_{g, 0}) \right)$,
\begin{align*}
0 &=\left. \frac{\partial}{\partial t}  \E{m_3(O; k, \tau_0(g_0), g, q_k, q_{\tau, 0} + t q_\tau, \alpha_g)} \right|_{t=0}.
\end{align*}
Under no further conditions,
\begin{align*}
0 &=\left. \frac{\partial}{\partial t}  \E{m_3(O; k, \tau, g_0, q_k, q_\tau, \alpha_{g, 0} + t \alpha_g)} \right|_{t=0}.
\end{align*}
\end{corollary}

Error bounds for estimation in corollary \ref{co:error-bounds} explicitly reflect the sequential nature of debiasing function estimation. However, it is worth recalling that also the ordinary nuisances must be estimated sequentially, implicitly implying the same sequentiality in the opposite direction. As usual in debiased machine learning, error bounds are products of estimation errors from ordinary nuisance functions and debiasing functions. This product form of error bounds allows $\sqrt{n}$-estimation of $\theta_0$ under the standard assumptions in debiased machine learning on faster than $n^{1/4}$-estimation of ordinary nuisances and debiasing functions and appropriate sample splitting. Notably, when nuisances are defined via Fredholm integral equations as they are here, the error rates are generally rates of projections of nuisances. Crucially, having error bounds based on projections of nuisances eliminates the problem from ill-posedness \citep{bennett2022}, which is a common feature of inverse problems.
\begin{corollary}[Error bounds] \label{co:error-bounds}
Suppose $\theta_0 = \E{m_0(O; k_0(\tau_0, g_0))}$, assumption \ref{a:iv-strong-relevance}, and the conditions for theorem \ref{th:nonlinear-strong-id} hold for $q = q_\tau$. Then,
\begin{footnotesize}
\begin{align*}
&\left| \E{m_3(O; \hat{k}_n, \hat{\tau}_n, \hat{g}, \hat{q}_{k, n}, \hat{q}_{\tau, n}, \hat{\alpha}_{g, n})} - \theta_0 \right| \\
&\ \  \leq  \min \left\{ \left\lVert P^{A, \mathcal{K}}_{\mathcal{L}_2(Z)} \left(k_0(\tau_0, g_0) - \hat{k}_n(\hat{\tau}_n, \hat{g}_n)\right) \right\rVert_2 \left\lVert q_{k, 0} - \hat{q}_{k, n} \right\rVert_2, \ \left\lVert k_0(\tau_0, g_0) - \hat{k}_n(\hat{\tau}_n, \hat{g}_n) \right\rVert_2 \left\lVert P_{A, \mathcal{K}}^{\mathcal{L}_2(Z)} \left( q_{k, 0} - \hat{q}_{k, n} \right) \right\rVert_2  \right\}  \\
&\ \ \  + \min \left\{ \left\lVert P^{\mathcal{L}_2(W)}_{Z, \mathcal{T}} \left(\tau_0(g_0) - \hat{\tau}_n(\hat{g}_n) \right) \right\rVert_2 \left\lVert q_{\tau, 0}(\hat{q}_{k, n}) - \hat{q}_{\tau, n}(\hat{q}_{k, n}) \right\rVert_2, \ \left\lVert \tau_0(g_0) - \hat{\tau}_n(\hat{g}_n) \right\rVert_2 \left\lVert P_{\mathcal{L}_2(W)}^{Z, \mathcal{T}}  \left( q_{\tau, 0}(\hat{q}_{k, n}) - \hat{q}_{\tau, n}(\hat{q}_{k, n}) \right) \right\rVert_2  \right\} \\
&\ \ \  + \left\lVert \left(\hat{g}_n - g_0 \right) \right\rVert_2 \left\lVert \alpha_{g, 0}(Z; \hat{q}_{k, n}, \hat{q}_{\tau, n}) - \hat{\alpha}_{g, n}(Z; \hat{q}_{k, n}, \hat{q}_{\tau, n}) \right\rVert_2.
\end{align*}
\end{footnotesize}
\end{corollary}

\begin{definition}[Debiased Machine Learning Estimator]  \label{d:dml}
Take $K$ folds with index set $\mathcal{I}_j$ for $k = 1, 2, \hdots, K$ of (approximately) equal size. Construct all nuisance estimators $K$ nuisance estimators $\eta^{(j)}$ with all data except data in $\mathcal{I}_j$ for $\eta \in \{ k, \tau, g, q_k, q_\tau, \alpha_g \}$. The debiased machine learning estimator is
\begin{align}
\hat{\theta}_n = \frac{1}{K} \sum_{k=1}^K \frac{1}{|\mathcal{I}_j|} \sum_{i \in \mathcal{I}_j} m_3(O; \hat{k}^{(j)}, \hat{\tau}^{(j)}, \hat{g}^{(j)}, \hat{q}_k^{(j)}, \hat{q}_\tau^{(j)}, \hat{\alpha}_g^{(j)}).
\end{align}
\end{definition}

\begin{definition}[L2 Convergence Rates of Nuisance Estimators]  \label{d:l2-convergence-rates}
Direct convergence rates:
\begin{align*}
\left\lVert k_0(\tau_0, g_0) - \hat{k}_n(\hat{\tau}_n, \hat{g}_n) \right\rVert_2 &\leq \epsilon_{n, k} \\
\left\lVert \tau_0(g_0) - \hat{\tau}_n(\hat{g}_n) \right\rVert_2 &\leq \epsilon_{n, \tau} \\
\left\lVert \hat{g}_n - g_0 \right\rVert_2 &\leq \epsilon_{n, g} \\
\left\lVert q_{k, 0} - \hat{q}_{k, n} \right\rVert_2 &\leq \epsilon_{n, q_k} \\
\left\lVert q_{\tau, 0}(q_k) - \hat{q}_{\tau, n}(q_k) \right\rVert_2 &\leq \epsilon_{n, q_\tau} \text{ for any } q_k \in \mathcal{Q}_k \\
\left\lVert \alpha_{g, 0}(q_k, q_\tau) - \hat{\alpha}_{g, n}(q_k, q_\tau) \right\rVert_2 &\leq \epsilon_{n, \alpha_g} \text{ for any } q_k \in \mathcal{Q}_k, q_\tau \in \mathcal{Q}_\tau
\end{align*}
Projected convergence rates:
\begin{align*}
\left\lVert P^{A, \mathcal{K}}_{\mathcal{L}_2(Z)} \left(k_0(\tau_0, g_0) - \hat{k}_n(\hat{\tau}_n, \hat{g}_n)\right) \right\rVert_2 &\leq \epsilon_{n, k, \mathcal{L}_2(Z)} \\
\left\lVert P_{A, \mathcal{K}}^{\mathcal{L}_2(Z)} \left( q_{k, 0} - \hat{q}_{k, n} \right) \right\rVert_2 &\leq \epsilon_{n, q_k, \mathcal{K}} \\
\left\lVert P_{\mathcal{L}_2(W)}^{Z, \mathcal{T}} \left(\tau_0(g_0) - \hat{\tau}_n(\hat{g}_n) \right) \right\rVert_2 &\leq \epsilon_{n, \tau, \mathcal{L}_2(W)} \\
\left\lVert P^{\mathcal{L}_2(W)}_{Z, \mathcal{T}} \left( q_{\tau, 0}(q_k) - \hat{q}_{\tau, n}(q_k) \right) \right\rVert_2 &\leq \epsilon_{n, q_\tau, \mathcal{T}} \text{ for any } q_k \in \mathcal{Q}_k \\
\end{align*}
\end{definition}

\begin{assumption}[Rate conditions on nuisance estimators]  \label{a:rate-conditions}
\begin{enumerate}
  \item $\lVert \hat{\eta}_n - \eta_0 \rVert_2^2 = o_p(1)$ for any $\eta \in \left\{k, \tau, g, q_k, q_\tau, \alpha_g\right\}$.
  \item The following requirements on the convergence rates in L2 norm hold:
\begin{align*}
o_p \left( n^{-1/2} \right) &= 
 \min \left\{ \epsilon_{n, k, \mathcal{L}_2(Z)} \epsilon_{n, q_k}, \epsilon_{n, k} \epsilon_{n, q_k, \mathcal{K}}  \right\} \\
o_p \left( n^{-1/2} \right) &=  \min \left\{ \epsilon_{n, \tau, \mathcal{L}_2(W)} \epsilon_{n, q_\tau}, \epsilon_{n, \tau} \epsilon_{n, q_\tau, \mathcal{T}}  \right\} \\
o_p \left( n^{-1/2} \right) &=  \epsilon_{n, g} \epsilon_{n, \alpha_g}
\end{align*}
\end{enumerate}
\end{assumption}

Assumption \ref{a:rate-conditions} is useful, because it can circumvent problems with possibly poor convergence rates in ill-posed inverse problems as long as the projections of nuisance functions are estimated with sufficient convergence rates, which may be much easier to satisfy. Firstly, we require mean square convergence in probability of each nuisance function estimator. In a more general framework, \citep{bennett2022} in assumptions 12-21b provide extensive conditions under which mean square convergence in probability is satisfied for minimum-norm nuisance functions identified as solutions to Fredholm integral equations of the first kind. The requirements on convergence rates in L2 norm in assumption \ref{a:rate-conditions} are doubly robust, and allow for convergence rates of projections instead of the nuisance functions themselves. As a result, estimation is robust not only in the typical debiased machine learning sense to slower than parametric estimation of nuisance parameters, but also to the ill-posedness of inverse problems that define some nuisances by only requiring L2 norm convergence in more favourable projections. For example, it suffices if all projections of nuisances identified by integral equations satisfy $o_p \left( n^{-1/4} \right)$ convergence in L2 norm.

\begin{theorem}[Debiased Estimation of Target Parameter]  \label{th:dr-est}
Let $\hat{\theta}_n$ be a debiased machine learning estimator as in definition \ref{d:dml}. 
Suppose $\theta_0 = \E{m_0(O; k_0(\tau_0, g_0))}$, assumption \ref{a:iv-strong-relevance}, assumption \ref{a:rate-conditions}, and the conditions for theorem \ref{th:nonlinear-strong-id} hold for $q = q_\tau$. Then,
\begin{align*}
\hat{\theta}_k - \theta_0 &= \mathbb{E}_{n, k} \left[ m_3(O; k_0, \tau_0, g_0, q_{k, 0}, q_{\tau, 0}, \alpha_{g, 0}) - \theta_0 \right] + o_p \left( n^{-1/2} \right) \\
\sqrt{n} \left( \hat{\theta}_k - \theta_0 \right) &= \frac{1}{\sqrt{n}} \sum_{i=1}^n \left( m_3(O; k_0, \tau_0, g_0, q_{k, 0}, q_{\tau, 0}, \alpha_{g, 0}) - \theta_0 \right) + o_p(1) \\
 &\rightarrow N(0, \sigma_0^2), \ \ \ \ \sigma_0^2 = \E{\left( m_3(O; k_0, \tau_0, g_0, q_{k, 0}, q_{\tau, 0}, \alpha_{g, 0}) - \theta_0 \right)^2} \\
 \frac{\sqrt{n}}{\hat{\sigma}^2_n} \left( \hat{\theta}_k - \theta_0 \right) &\rightarrow N(0, 1) \ \text{ for any } \ \hat{\sigma} \ \text{ s.t. } \ \hat{\sigma}^2_n \rightarrow \sigma^2_0 \ \text{ in probability}.
\end{align*}
\end{theorem}

\section{Practical Guide}  \label{sec:guide}
Straightforward testing and discussion of model assumptions is key in any application. In this section, we provide a practical guide to identification with this approach. Each of the four steps we describe has its own subsection. As in standard IV, the relevance assumptions generally remain testable. The conditional exogeneity assumption is not testable up to a specification test.

\subsection{Find $T$ and test relevance of $Z, W$ for $U$ } \label{ssec:guide-1}
In this step, we test the relevance of $W$ for $U$ (assumption \ref{a:riv}.\ref{a:pl}), as well as the relevance of $Z$ for $U$. The latter is a necessary condition for the relevance of $Z$ for $A$ conditional on $T$ (assumption \ref{a:riv}.\ref{a:iv-rel}), which is tested explicitly in subsection \ref{ssec:guide-2}.

First, we find some $T = \tau(Z)$ such that $Z \CI W \ | \ T$ is satisfied. As long as some $T=\tau(Z)$ leads to the conditional independence of $Z$ and $W$, this control function $\tau \in \mathcal{L}_2(Z)$ also renders $Z$ and $U$ conditionally independent. To test for the sufficient relevance of $Z$ and $W$ with respect to $U$, we can use $T$. A sufficient condition for relevance of $Z$ and $W$ for $U$ is that both $Z$ and $W$ contain spare information conditional on $T$. This motivates the choice of a valid $\tau \in {\mathcal{T}}_\text{valid}$, which captures the least information about $Z$ while ensuring the conditional exogeneity of $Z$ conditional on $T$, as discussed in section \ref{sec:confounding}. To simplify the argument, suppose that our model is linear and the dimensions for $(Z, T, W, U)$ are $(d_Z, d_T, d_W, d_U)$. Often, relevance of $Z$ and $W$ for $U$ imply $\min \{d_Z, d_W\} \geq d_U$. The minimum dimension of $T$ to ensure conditional independence of $Z$ and $W$ is $d_U$, so we know that any $T = \tau(Z)$ such that $Z \CI W \ | \ T$ must satisfy $d_T \geq d_U$. If we can reject a test with the null hypothesis 
\begin{align*}
H_0: \min \{d_Z, d_W\} \leq d_T, \text{ and alternative } H_a: \min \{d_Z, d_W\} > d_T,
\end{align*} 
this implies $\min \{d_Z, d_W\} > d_U$. Thus, there is a test whether $Z$ and $W$ are relevant for $U$. However, $\min \{d_Z, d_W\} > d_T$ is a sufficient, not a necessary condition for the relevance of $Z$ and $W$ for $U$. $Z$ and $W$ are still relevant for $U$ when $\min \{d_Z, d_W\} = d_U$. Unfortunately, this hypothesis is not testable with unobserved $U$. So, how should an applied researcher proceed when $\min \{d_Z, d_W\} = d_T$ (which could mean $\min \{d_Z, d_W\} = d_U$)? Here, we need to distinguish $d_Z = d_T$ from $d_W = d_T$. 
    \begin{enumerate}
      \item[$d_Z = d_T$] When $T$ contains as much information as $Z$, there is no point in moving to step 2. The instruments $Z$ contain no variation conditional on $T$, so $Z$ cannot be relevant for treatment $A$ conditional on $T$. 
      \item[$d_W = d_T$] We know for sure that $d_W \leq d_U$. If $d_W = d_U$, the proxies $W$ are exactly relevant for $U$ without spare information. If we are willing to assume $d_W = d_U$, we could move forward to step 2. The variation in $U$ associated with $Z$ would still be held fixed with $T$ in this case. Yet, $d_W = d_U$ is not testable. It may well be that $d_W < d_U$. Then, the variation in $U$ associated with $Z$ is \emph{not} held fixed with $T$. We can never test the completeness of $W$ for $U$ when $d_W = d_T$. Accordingly, we do not generally suggest to move to step 2 by relying on the assumption $d_W = d_U$, when $d_W = d_T$ is observed.
    \end{enumerate}

\subsection{Test relevance of $Z$ for $A$ given $T$}   \label{ssec:guide-2}
In step 1, the relevance of $Z, W$ for $U$ was confirmed. Now, we test the relevance of the instruments $Z$ for treatment $A$ conditional on $T$ (assumption \ref{a:riv}.\ref{a:iv-rel}). 

Depending on the additional model assumptions, this is either a test of completeness (\ref{a:riv}.\ref{a:iv-rel}), or common support (\ref{a:common-support}). As any $T=\tau(Z)$ is simply the control function $\tau \in \mathcal{L}_2(Z)$ applied to instruments $Z$, the test of relevance of $Z$ for $A$ given $T$ is straightforward for any given $\tau$. It is as simple as a test of instrument relevance with observed confounders. Let us return to a linear model. If all components of the instrument vector $Z$ are correlated with both $U$ and $A$, the conditional relevance requirement simplifies to $(d_Z - d_T) \geq d_A$. After conditioning on all variation in $Z$ which correlates with $U$ by holding $T$ fixed, $d_Z - d_T$ dimensions of instruments $Z$ remain to infer the causal effect of treatment $A$ on outcome $Y$. The remaining instrument variation of dimension $d_Z - d_T$ is relevant for treatment $A$ only if the treatment's dimension $d_A$ is smaller than or equal to $d_Z - d_T$. 

If $Z$ is found to be relevant for $A$ given $T$, it implies that $Z$ is relevant for both $U$ and $A$. Interpreting $U$ as any source of unobserved variation which associates instruments $Z$ and proxies $W$, $Z$ would have no variation conditional on $T$ if they were not sufficiently relevant for $U$ ($d_Z < d_U$). So, if $Z$ is relevant for $A$ given $T$, it implies that $Z$ already had to be relevant for $U$.

\subsection{Exogeneity of $Z$ conditional on $U$}   \label{ssec:guide-3}
In step 1 and 2 we tested all relevance assumptions in this model. The conditional exogeneity assumption for $Z$ conditional on $U$ remains untestable. To be precise, $Y \CI Z \ | \ (A, U)$ (assumption \ref{a:riv}.\ref{a:iv-exog}) can only be justified on theoretic grounds, not observed data. In order to justify conditional exogeneity theoretically, $T$ can be used to understand the unobserved common confounders $U$. As $T$ captures all variation in $Z$ associated with $U$, $T$ immediately explains $U$ in terms of its association with $Z$ and $W$. From the association of $T$ and $W$, we can interpret $U$ even better. For example: If subject-specific pre-college GPA measures are used as instruments $Z$, and $T$ turns out to capture average GPA, then $U$ could be interpreted as general ability. Suppose $W$ contains dummies capturing whether someone has engaged in risky behaviour, including drugs and illegal activity, while in high school. From theory and empirical evidence we would expect high ability to lead to less risky behaviour. Thus, if $T$ is an average GPA, it would be expected to negatively correlated with $W$.

Once we used $T$ to understand the variation reflected by the unobserved confounders $U$, we can construct a theoretic argument with respect to the conditional exogeneity of instruments $Z$. In our example, the common confounder $U$ reflected general ability. The conditional exogeneity assumption reduces to whether conditional on general ability $U$, the subject-specific pre-college GPA measures $Z$ are excluded. This argument clearly depends on the respective choice of treatment $A$ and outcome $Y$.

As in standard IV, specification tests, which can be revealing about the exogeneity of $Z$, are possible if the model is overidentified. If a specification test suggests that different subset of instruments $Z$ conditional on $T$ result in estimators with different probability limits, we reject that all instruments $Z$ are excluded conditional on $T$ (unless the estimand is the local average treatment effect which can vary across subpopulations). Necessary for any such test is model overidentification for the causal effect $\theta_0$ conditional on $T$. In a simple linear model, overidentification would e.g. mean $(d_Z - d_T) > d_A$. After keeping $d_T$ dimensions of $Z$ fixed, the instruments must still contain overidentifying information. 

Ultimately, just as in standard IV, the conditional exogeneity assumption for $Z$ remains largely untestable. Therefore, it is crucial to better understand $U$ from the control variable $T$.

\subsection{Estimation}     \label{ssec:guide-4}
In the final fourth step, use $Z$ to instrument for treatment $A$ conditional on control variable $T$, to identify (and estimate) the structural function or average causal effect $\theta_0$ of $A$ on $Y$. Having established that all necessary relevance and exogeneity requirements hold, an estimator $\hat{\theta}$ can be formulated for the causal effect of interest $\theta_0$. The form of this estimator depends on the type of parametric model assumptions made. 


\section{Example: Linear Returns to Education}   \label{sec:ex-educ}
Interested in the returns to education, we use data from the National Longitudinal Survey of Youth 1997 \citep{nls97}. The variables of interest are introduced below.
\begin{itemize}
  \item[$Y$] Household net worth at 35: continuous variable, in USD.
  \item[$A$] BA degree: $1$ if individual $i$ obtained a BA degree, $0$ otherwise.
  \item[$Z$] Pre-college test results: subject-specific and overall GPA; ASVAB percentile.
  \item[$W$] Risky behaviour dummies: whether $i$ drank, smoked, or engaged in other behaviours considered risky by the age of 17.    
  \item[$U$] Ability: Unmeasured intellectual capacity.
  \item[$\tilde{U}$] Other biases: Selection on unobservables into obtaining a BA degree (at least in part result of optimising individuals).
  \item[$X$] Covariates: sex, college GPA, parental education/net worth, siblings, region, etc.
\end{itemize}


A review of the vast literature on returns to education is far beyond the scope of this paper \citep{psacharopoulos2018}. Instead, we focus on estimation of a very specific return to education: The causal effect of obtaining a bachelor's degree $A$ on household net worth at 35. Even in a simple linear model like 
\begin{align}
Y &= \alpha_Y + A \beta + U \gamma_Y + W \upsilon_Y + X \eta_Y + \varepsilon_Y,  \label{eq:ex-educ-lin}
\end{align}
two distinct potential sources of confounding are easily identified via the unobservable components of \ref{eq:ex-educ-lin}:
\begin{enumerate}
  \item[$U \gamma_Y$] Ability $U$ likely has a positive effect on household net worth $Y$, by means of salary and non-salary based net worth accumulation \citep{griliches1977}. The vector-valued linear parameter $\gamma_Y$ captures this positive linear effect of ability on net worth. 
  \item[$\varepsilon_Y$] The disturbance $\varepsilon_Y$ captures all variation in $Y$, which is jointly unexplained by $(A, U, W, X)$. This can be understood as individual-specific, heterogeneous characteristics, and chance. 
\end{enumerate}
Any correlation of $A$ with either of these terms leads to biased estimates of $\beta$. 

\subsubsection*{How does obtaining a BA degree $A$ correlate with ability $U$ and a general disturbance $\varepsilon_Y$?} 
In this identification problem, selection bias in inherent.
At least to some degree, individuals choose whether to obtain a BA degree as a result of an optimisation problem of expected utility subject to an information set $\mathcal{I}$:
\begin{align}
A &= \underset{a \in \{0,1\}}{\arg \max{}}\left( \E{u(Y(a)) - c(a) | A=a, \mathcal{I}} \right), \label{eq:ex-educ-opt}
\end{align}
where $u: \mathcal{Y} \rightarrow \mathbb{R}$ is a utility function for net worth with diminishing returns, and $c: \{0,1\} \rightarrow \mathbb{R}$ a cost function for obtaining a BA degree $A$. Both utility and cost function can be individual-specific.
 
For ease of illustration, suppose individuals are perfectly informed with $\mathcal{I} = (A, U, W, X, \varepsilon_Y)$. Then, each individual chooses $a \in \{0,1\}$ to maximise the utility associated with potential outcome $Y(a)$ minus cost $c(a)$. In this case, there is an easy decision rule to determine optimal $A$:
\begin{align*}
A &= \underset{a \in \{0,1\}}{\arg \max{}}  u(Y(a)) - c(a), \\
&=  \mathds{1} \left[ u(Y(1)) - u(Y(0)) \ > \ c(1) - c(0) \right].
\end{align*}

\begin{enumerate}
  \item[$U \gamma_Y$] Ceteris paribus, an increase in ability $U$ equally increases $Y(0)$ and $Y(1)$ according to model \ref{eq:ex-educ-lin}. Due to diminishing returns in the utility function $u$, $u(Y(1)) - u(Y(0))$ decreases. However, also $c(1)$ decreases as higher-ability individuals experience a lower utility cost of obtaining a BA degree. The overall effect on the choice of $A$ is ambiguous and depends on the utility functions of the individual. 
  \item[$\varepsilon_Y$] The effect of $\varepsilon_Y$ on the choice of $A$, on the other hand, is unambiguous. An increase in $\varepsilon_Y$ reduces $u(Y(1)) - u(Y(0))$ due to the diminishing returns of $u$. Cost $c$, however, is unaffected by $\varepsilon_Y$. Hence, $A$ inevitably negatively correlates with $\varepsilon_Y$.
\end{enumerate}

This logic regarding negative selection bias when treatment is chosen by utility-maximising individuals is by no means novel \citep{heckman2006}, or unique to the returns to education identification problem. Negative selection bias is inherent to the treatment variable when it is at least in part the result of optimising behaviour by utility-maximising heterogeneous individuals. Novel in our approach is the ability to explicitly account for certain biases, in this case ability bias, when proxies for them exist. Finding excluded instruments can be much more straightforward when pertinent biases, like ability bias, have already been taken care of.

In our identification approach, instruments $Z$ are pre-college test results. These results are strongly correlated with ability $U$. Yet, conditional on ability, 
and some other covariates, pre-college test results contain random variation, which is excluded with respect to household net worth $Y$ (at age 35).
Concurrently, even random variation in pre-college test results is a strong predictor of obtaining a BA degree. Hence, instrument relevance likely holds. The proxies $W$ are dummies for whether an individual engaged in risky behaviours at high school age. Among others, the risky behaviour dummies include drinking, smoking (marijuana), selling drugs and stealing. Theory and empirical evidence suggest the correlation of low intelligence and risky behaviour \citep{loeber2012}. Therefore, ability $U$ both causes instruments $Z$ and proxies $W$ in our data. Ability $U$ is the common confounder in this causal question. Clearly, additional covariates are necessary to justify instrument exogeneity. These include sex, college GPA, parental education and net worth, the number of siblings, region of residence, etc. 

\subsection{Assumptions}  \label{ssec:ex-educ-a}
The linear equivalent to the general common confounding IV model in assumption \ref{a:riv} is described as assumption \ref{a:livcc}. Again, for ease of notation assume $d_A = 1$, just as in this returns to education identification problem.
\begin{assumption}[Linear IV Model with Common Confounding] \label{a:livcc} \ \ \ \ \ \ \ \

\begin{enumerate}
  \item Linear outcome model projection:  
          \begin{align}
          Y &= \alpha_Y + A \beta + U \gamma_Y + W \upsilon_Y + X \eta_Y + \varepsilon_Y
          \end{align}
  \item Instruments \label{a:iv}
  \begin{enumerate}
    \item Exogeneity: \label{a:livcc-iv-exog} $ \E{\varepsilon_Y (Z, U, W, X)} = \pmb{0}. $
    \item Relevance: \label{a:livcc-iv-rel}  For the linear projection of $A$ on $(Z, U, W, X)$,          
    \begin{align}
          A &= \alpha_{A} + Z \zeta + U \gamma_{A} + W \upsilon_{A} + X \eta_{A} + \varepsilon_A, 
          \ \E{\varepsilon_A (Z, U, W, Z)} = \pmb{0} \\
          \operatorname{rank} & \left(\E{\left. \left( Z \zeta \right) A \ \right| \ T, X}\right) = d_A. \label{eq:livcc:iv-rel-A}
          \end{align}
  \end{enumerate}
  \item Proxies \label{a:livcc-pl}
  \begin{enumerate}
    \item Exogeneity: \label{a:livcc-pl-excl} For the linear projection of $W$ on $(Z, U, X)$ and $(Z, X)$, 
    \begin{align} 
       W &= \alpha_{W} + U \gamma_{W} + X \eta_W + \varepsilon_W, 
       &  \E{\varepsilon_W (Z, U, X)} &= \pmb{0}, \\
       W &= \tilde{\alpha}_W + Z \tilde{\gamma}_W + X \tilde{\eta}_W + \tilde{\varepsilon}_W, 
       & \E{\tilde{\varepsilon}_W (Z, X)} &= \pmb{0}.  \label{eq:w-lp}
    \end{align}
    with $T \coloneqq Z \tilde{\gamma}_W + X \tilde{\eta}_W$.
    \item Relevance: \label{a:livcc-pl-rel} $\operatorname{rank}(\gamma_W) \geq  d_U$
    \begin{align} . \end{align} 
  \end{enumerate} 
\end{enumerate}
\end{assumption}
To simplify notation, let $Z_{|X}$ be the true residual of a projection of $Z$ onto $X$. The linearity of the outcome model implies that the covariance $\Cov{(Z_{|X} \zeta), Y}$ is
\begin{align*}
\Cov{(Z_{|X} \zeta), Y} &=  \Cov{(Z_{|X} \zeta), A} \beta + \Cov{(Z_{|X} \zeta), U} \gamma_{Y} + \Cov{(Z_{|X} \zeta), W} \upsilon_{Y}.
\end{align*}
The above expression uses the uncorrelatedness of $Z$ and $\varepsilon_Y$ in assumption \ref{a:livcc}.\ref{a:livcc-iv-exog}. If it were not for the linear confounding from the unobserved common confounders $U$ and proxies $W$, $Z$ would be excluded. Next, we demonstrate how to use the proxies $W$ to keep $\Cov{(Z \zeta) U}$ fixed. 
\begin{align*}
\Cov{(Z_{|X} \zeta), U} &= \Cov{(Z_{|X} \zeta), W} \gamma_W^\intercal \left( \gamma_W \gamma_W^\intercal \right)^{-1}
\end{align*} 
The inverse $\left( \gamma_W \gamma_W^\intercal \right)^{-1}$ exists under assumption \ref{a:livcc}.\ref{a:livcc-pl} that the rank of $\gamma_W$ is at least $d_U$. Then, slightly rewriting $\Cov{(Z \zeta) Y}$ as
\begin{align*}
\Cov{(Z_{|X} \zeta), Y} &=  \Var{Z_{|X} \zeta} \beta + \Cov{(Z_{|X} \zeta), W} \tilde{\upsilon}_W, \\
\tilde{\upsilon}_Y &\coloneqq \upsilon_Y + \upsilon_A \beta +  \gamma_W^\intercal \left( \gamma_W \gamma_W^\intercal \right)^{-1} \left( \gamma_Y + \gamma_A \beta \right), 
\end{align*}
implies that any endogeneity of residualised instruments $Z_{|X}$ is controlled for by conditioning on $Z \tilde{\gamma}_W$ from the linear projection \ref{eq:w-lp}. To be precise, 
\begin{align*}
\Cov{(Z_{|X} \zeta), Y | Z \tilde{\gamma}_W} &=  \Var{Z_{|X} \zeta | Z \tilde{\gamma}_W} \beta + \underbrace{\Cov{(Z_{|X} \zeta) W | Z \tilde{\gamma}_W}}_{=0} \tilde{\upsilon}_Y.
\end{align*}
The covariance of the first stage can be rewritten as
\begin{align*}
\Cov{(Z_{|X} \zeta), A | Z \tilde{\gamma}_W} &= \Var{Z_{|X} \zeta | Z \tilde{\gamma}_W} + \underbrace{\Cov{(Z_{|X} \zeta) W | Z \tilde{\gamma}_W}}_{=0} \tilde{\upsilon}_A.
\end{align*}
Using both of these results, and one-dimensional treatment $A$ to simplify notation, a simple ratio form for the linear effect of $A$ on outcome $Y$ is 
\begin{align}
\beta &= \frac{ \Cov{\left(Z_{|X} \zeta \right) Y \ | \ Z \tilde{\gamma}_W} }{ \Cov{\left(Z_{|X} \zeta \right) A \ | \ Z \tilde{\gamma}_W}} = \frac{ \Cov{\left(Z \zeta \right) Y \ | \ T, X} }{ \Cov{\left(Z \zeta \right) A \ | \ T, X} }, \text{ with } T = Z \tilde{\gamma}_W.
\end{align}
Hence, the estimator differs from standard IV based estimation only by also holding a linear prediction $T$ of $W$ fixed as the partial predicted values $Z \zeta$ for $A$ change. Thus, the relevance requirement \ref{a:livcc}.\ref{a:livcc-iv-rel} for the instruments $Z$ is conditional on $T$ and $X$. $T$ can be represented by a $d_U$-dimensional linear function of $Z$, $\E{U | Z, X}$, multiplied by $\gamma_W$. Hence, a simpler way to understand the relevance requirement \ref{a:livcc}.\ref{a:livcc-iv-rel} is as
\begin{align}
\operatorname{rank}(\zeta) \geq (d_U + d_A)  \label{eq:livcc-simple-rel}
\end{align}
A total of $d_U$ dimensions of variation in $Z$ are typically needed to account for the $d_U$-dimensional confounding effect of $U$ via $\E{U | Z, X}$, while the remaining variation in $Z$ still needs to be relevant for $A$. Other than in trivial cases\footnote{e.g. when $Z$ contains perfectly collinear variation conditional on $(U, W, X)$.}, equation \ref{eq:livcc-simple-rel} describes this relevance requirement satisfactorily as a rank condition on $\zeta$, the partial linear projection effect of $Z$ on $A$ conditional on $(U, W, X)$.


\subsection{Find $T$ and test relevance of $Z$, $W$ for $U$}     \label{sssec:ex-lin-educ-1}
A valid control function is the linear prediction $T = Z \tilde{\gamma}_W$ under assumption \ref{a:livcc}.\ref{a:livcc-pl}, meaning that conditional on $(T, X)$, instruments $Z$ are still relevant for $A$. However, its OLS estimate $T = Z \hat{\tilde{\gamma}}_W$ generally is not a valid control function, because $T$ and $Z$ are perfectly correlated due to sampling variation, unless $d_Z > d_W$. However, even when $d_Z > d_W$, the true $\tilde{\gamma}_W$ will have rank $d_U \leq d_W$, while its OLS estimate $\hat{\tilde{\gamma}}_W$ always has possibly larger than necessary rank $d_W$ due to sampling variation. Ultimately, the estimate $\hat{\tilde{\gamma}}_W$ should at best have exactly rank $d_U$. A test is needed for the rank $r_0$ of matrix $\tilde{\gamma}_W$. If $\E{U | Z, X} = Z \gamma_Z + X \gamma_X$, then $\tilde{\gamma}_W = \gamma_Z \gamma_W$. Sufficient for the rank condition in assumption \ref{a:livcc}.\ref{a:livcc-pl} is $r_0 < \min\{d_Z, d_W\}$. This condition means that an unobservable variable of smaller dimension than both $W$ and $Z$ can explain all correlation between $W$ and $Z$ conditional on $X$. This unobserved variable is the common confounder $U$. By the definition of $U$ as the (minimum information) unobserved variable which renders $W$ and $Z$ mean-independent, $\gamma_Z$ has $d_U \leq d_Z$ linearly independent rows ($\operatorname{rank}(\gamma_Z) = d_U$). As $\gamma_W$ has dimensions $d_U \times d_W$ and $d_U \leq d_W$, $\operatorname{rank}(\gamma_W) \leq d_U$ and thus $r_0 = \operatorname{rank}(\gamma_Z \gamma_W) = \operatorname{rank}(\gamma_W)$.

While $r_0 < d_W$ suffices to confirm the relevance of $W$ for $U$ in assumption \ref{a:livcc}.\ref{a:livcc-pl}, $r_0 < d_Z$ is necessary for $Z$ to be relevant for treatment $A$ in assumption \ref{a:livcc}.\ref{a:livcc-iv-rel}. A suitable test for some $r < \min\{d_Z, d_W\}$ has null hypothesis
\begin{align}
H_0: r_0 \leq r, \text{ and alternative } H_a: r_0 > r.  \label{eq:test-r}
\end{align}
With the OLS estimator $\hat{\tilde{\gamma}}_W$, we apply a bootstrap based test for its rank. First, write the singular value decomposition as
\begin{align}
\tilde{\gamma}_W &= \underset{d_Z \times d_Z}{P_0} \underset{d_Z \times d_W}{\Pi_0} \underset{d_W \times d_W}{Q_0^\intercal}
\end{align}
Then, let $\phi_r(A) \coloneqq \sum_{j=r+1}^{m_A} \pi_j^2(A)$ be the sum of squared singular values of $A$ from the $(r+1)$ largest to the smallest singular value, which is the $m_A$-th singular value, where $m_A$ is the minimum across $A$'s number of rows and columns. Then, an equivalent test to \ref{eq:test-r} is a test with null hypothesis
\begin{align}
 H_0: \phi_r \left( \tilde{\gamma}_W \right) = 0, \text{ and alternative } H_a: \phi_r \left( \tilde{\gamma}_W \right) > 0. \label{eq:test-phi}
\end{align}

The bootstrap procedure is as follows:
\begin{enumerate}
  \item For each binary proxy $W_j \in W$, calculate the probability $p_j \coloneqq \Pr{(W_j = 1)}$
    under $H_0$ as  $p_{j, 0} = \text{Logit}\left( \left( Z \underset{d_Z \times r}{P_{0, r}} \underset{r \times r}{\Pi_{0, r}} \underset{r \times 1}{Q_{0, r, j}^\intercal} + X \underset{d_X \times 1}{\tilde{\eta}_{W, j}} \right) \beta_0 + \alpha_0 \right)$, where
        \begin{enumerate}
          \item $P_{0,r}$ corresponds to the first $r$ columns of $P_0$,
          \item $Q_{0,j}$ corresponds to the first $r$ entries of the $j$-th row of $Q_0$
          \item $\Pi_{0, r}$ corresponds to the $r \times r$ matrix of of the first $r$ rows and columns of $\Pi_0$,
          \item $\tilde{\eta}_{W, j}$ corresponds to $j$-th column of $\tilde{\eta}_W$,
          \item $\beta_0$ and $\alpha_0$ are univariate coefficients, which need to be estimated.
        \end{enumerate}
  \item Draw 1000 new bootstrap samples $b \in \mathcal{B}$ of binary proxies as $W^b_0$ using the $n \times d_W$ probability matrix $(p_{0, 0}, p_{1, 0}, \hdots, p_{d_W, 0})$.
  \item For each bootstrap sample $b \in \mathcal{B}$: Calculate the sample projection coefficient $\hat{\tilde{\gamma}}^b_{W, 0}$ by projecting $W^b_0$ onto $(Z, X)$ (all demeaned), and the sum of its smallest squared singular values starting at the $(r+1)$ largest as $\phi_{r, 0, b} \coloneqq \phi_r \left( \hat{\tilde{\gamma}}^b_{W, 0} \right)$. \\
  \item Obtain the $p$-value as $1 - \frac{1}{\left| \mathcal{B} \right|} \sum_{b \in \mathcal{B}} \mathds{1} \left(\phi_{r, 0, b} < \phi_r(\tilde{\gamma}_W) \right) $.
\end{enumerate}


\begin{figure}[h]
\begin{center}
\caption{Bootstrap based test for $H_0: \operatorname{rank}(\tilde{\gamma}_W) = r_0 \leq r$}  \label{f:boot-chen-lin}
\includegraphics[width=0.786\textwidth]{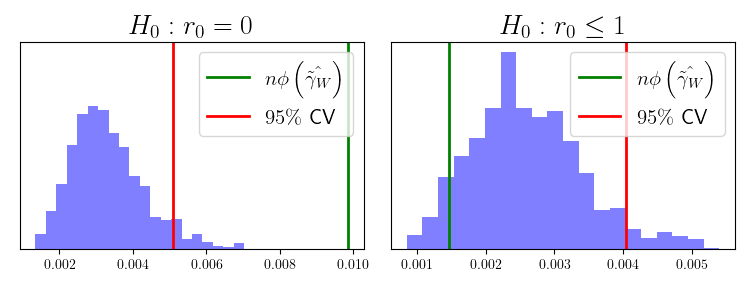}
\end{center}
\begin{footnotesize}
\textit{Notes}:
This figure illustrates the bootstrap distribution of the test statistic $n \phi_r \left( \hat{\tilde{\gamma}}_{W} \right)$ in the test with null hypothesis $H_0: \operatorname{rank}(\tilde{\gamma}_W) = r_0 \leq r$. The left figure depicts the test statistic distribution under $r=0$, and the left under $r=1$. For $r=0$, the $p$-value is zero. $H_0: r_0 = 0$ is strongly rejected. For $r=1$, the $p$-value is $0.933$. $H_0: r_0 \leq 1$ cannot be rejected at any meaningful level of significance.
\end{footnotesize}
\end{figure}

In figure \ref{f:boot-chen-lin}, the bootstrapped distributions of the test statistic $n \phi_r \left( \hat{\tilde{\gamma}}_{W} \right)$ are depicted under two different null hypotheses: $r_0 = 0$ and $r_0 \leq 1$. Non-rejection of the test is evidence in favour of the low rank $r_0$ of $\tilde{\gamma}_W$. In the left diagram of figure \ref{f:boot-chen-lin}, where the test concerns $r_0 = 0$, the $p$-value is at zero. The test provides strong evidence against $r_0 = 0$, which indicates some correlation between $W$ and $Z$ conditional on $X$. The right diagram of figure \ref{f:boot-chen-lin} depicts the test statistic bootstrap distribution for $H_0: r_0 \leq 1$, and provides strong evidence against rejection. The associated $p$-value is 93.3\%. Thus, we can conclude that the rank of $\gamma_W$ is at most one. In the NLS97 data, pre-college test results $Z$ and risky behaviour dummies have dimensions $d_Z=7$ and $d_W=9$. Thus, $r_0 \leq 1$ allows the conclusion that the common confounder dimension is small: $d_U \leq 1$. Conditional on covariates $X$, all covariance between $Z$ and $W$ is explained by a one-dimensional unobserved $U$. Successfully, the proximal assumption \ref{a:livcc}.\ref{a:livcc-pl} was tested. In addition, the necessary $d_Z > d_U$ condition for conditional instrument relevance (assumption \ref{a:livcc}.\ref{a:livcc-iv-rel}) was confirmed. 

\subsection{Test relevance of $Z$ for $A$ given $T$}     \label{sssec:ex-lin-educ-2}
Despite satisfying the necessary $d_Z > d_U$ condition for IV relevance (assumption \ref{a:livcc}.\ref{a:livcc-iv-rel}), a proper test for the conditional relevance of $Z$ for $A$ given the control function $T$ is still missing. In this step, we first explain how to construct the here one-dimensional control variable $T$ after having conducted the tests in section \ref{sssec:ex-lin-educ-2}. Then, we test for the conditional relevance of instrument $Z$ for treatment $A$ given this control function $T$.

Given the statistical evidence in favour of $d_U \leq 1$, we construct the variable 
\begin{align}
\underset{N \times 1}{T} &\coloneqq \underset{N \times d_Z}{Z} \underset{d_Z \times 1}{\hat{P}_{0, 1}} \underset{1 \times 1}{\hat{\Pi}_{0, 1}},  \label{eq:control-variable}
\end{align}
with the singular value decomposition of the OLS estimator $\hat{\tilde{\gamma}}_W = \hat{P}_{0} \hat{\Pi}_0 \hat{Q}_0^\intercal$. $\hat{P}_{0, 1}$ is the first column of $\hat{P}_{0}$, and $\hat{\Pi}_{0, 1}$ is the top-left entry of $\hat{\Pi}_0$. Aside from sampling error, proxies $W$ are mean-independent from instruments $Z$ conditional on $(T, X)$. 

With the control $T$ now defined, we can use a bootstrap based test to confirm the relevance of instruments $Z$ for $A$ conditional on $(T, X)$. The null hypothesis can be formulated as 
\begin{align}
H_0: \operatorname{rank}\left( \E{ (Z {\zeta}) A | T, X} \right) < d_A, \text{ with alternative } H_a: \operatorname{rank}\left( \E{ (Z {\zeta}) A | T, X} \right) = d_A.
\end{align}
Importantly, under $H_0$ the effect of $Z$ on $A$ (given $X$) would be fully described by a one-dimensional $T$, as the dimension of $U$ was found to be $r_0 \leq 1$ in section \ref{sssec:ex-lin-educ-1}. When treatment $A$ is one-dimensional, a simple test for this null hypothesis compares the $R^2$ of an unrestricted (\ref{eq:reduced-ur}) and restricted regression (\ref{eq:reduced-r}). 
\begin{align}
A &= \tilde{\alpha}_{A, ur} + Z \tilde{\zeta} + X \tilde{\eta}_{A, ur} + \tilde{\varepsilon}_{A, ur}, & & & \E{\tilde{\varepsilon}_{A, ur} | Z, X} &= 0,  \label{eq:reduced-ur} \\
A &= \tilde{\alpha}_{A, r} + T \tilde{\gamma}_{A} + X \tilde{\eta}_{A, r} + \tilde{\varepsilon}_{A, r}, & & & \E{\tilde{\varepsilon}_{A, r} | T, X} &= 0.   \label{eq:reduced-r}
\end{align}
Under $H_0$, both regressions would predict $A$ equally well, despite the dimension reduction on $Z$ in the second regression, \ref{eq:reduced-r}. 
With the uncertainty in estimated $T$, we use a simple bootstrap-based test. With 1000 bootstrap samples $b_t \in \mathcal{B}_t$, we obtain a bootstrap distribution of $R^2_{r}$. Under $H_0$, $R^2_{ur}$ is (asymptotically) distributed as $R^2_r$. 

\begin{figure}[h]
\begin{center}
\caption{Test for Relevance of $Z$ for $A$ given $T$ ($H_0: \operatorname{rank}\left( \E{ Z {\zeta} | T, X} \right) < d_A$)}  \label{f:boot-rel-lin}
\includegraphics[width=0.655\textwidth]{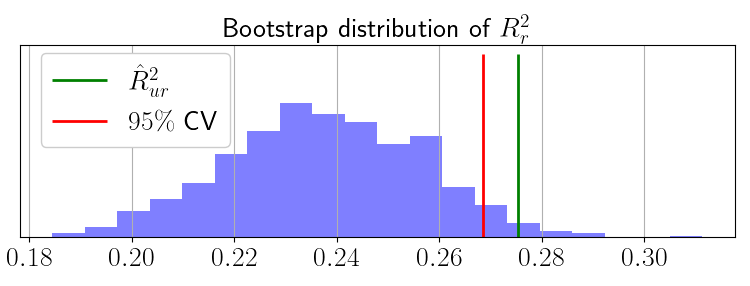}
\end{center}
\begin{footnotesize}
\textit{Notes}:
This figure illustrates the bootstrap distribution of the restricted $R^2_r$ in regression \ref{eq:reduced-r}. The dimension of $T$,$d_T=1$, is based on the test in section \ref{sssec:ex-lin-educ-2}.
\end{footnotesize}
\end{figure}

Figure \ref{f:boot-rel-lin} depicts the bootstrap distribution of $R^2_r$ based on the restricted regression \ref{eq:reduced-r}. The control variable $T$ is constructed for each bootstrap sample as described in \ref{eq:control-variable}. The unrestricted $R^2_{ur}$ based on the unrestricted regression \ref{eq:reduced-ur} fits the data significantly better, which indicates rejection of $H_0$. The $p$-value is 0.023. There is predictive information in $Z$ for $A$, beyond that controlled for in $(T, X)$. In other words, $Z$ satisfies the conditional instrument relevance requirement \ref{a:livcc}.\ref{a:livcc-iv-rel}.

\subsection{Exogeneity of $Z$ conditional on $U$}     \label{ssec:ex-lin-educ-3}
While both relevance assumptions \ref{a:livcc}.\ref{a:livcc-pl} and \ref{a:livcc}.\ref{a:livcc-iv-rel} could be tested successfully, the exogeneity of instrument $Z$ conditional on $U$ in assumption \ref{a:livcc}.\ref{a:livcc-iv-exog} remains generally untestable. 

To argue whether $Z$ is exogenous conditional on $U$, it is worth asking: Which information is being held fixed in $T$, and what does this imply about $U$? The linear construction of $T$ from $Z$ is illustrated in table \ref{t:hat-t}, where the instruments have been normalised to standard deviation one. 
$Z$ is mean-independent from $W$ conditional on $(T, X)$. $T$ mostly consists of an average of subject-specific pre-college GPA measures. In this sense, $T$ closely measures academic ability, as captured by pre-college GPA measures. Without the transcript GPA and ASVAB percentile, the subject-specific GPA measures describe 94.5\% of variation in $T$. Despite the negative dependence of $T$ on ASVAB percentile in its construction, $T$ positively correlates with ASVAB percentile unconditional on the GPA measures with a 0.31 correlation coefficient. The interpretation of $T$ and consequently $U$ is pretty straightforward: It positively reflects (academic) ability.
\begin{table}[h]
\centering
\caption{Construction of $T = Z \hat{P}_{0,1} \hat{\Pi}_{0,1}$}  \label{t:hat-t}
          \begin{tabular}{c | c c c c c} 
                                                         & \multicolumn{4}{c}{GPA} & ASVAB \\ 
                                                         & English & Math & SocSci & LifeSci & percentile \\ 
                                                         \midrule
                                                         $T$ & 0.537 & 0.216 & 0.280 & 0.214 & -0.262 \\ 
          \end{tabular}
\end{table}

As $U$ reflects (academic) ability, an increase in $T$ is expected to result in a reduction of risky behaviour \citep{loeber2012}. Indeed, a one standard deviation increase in $T$ reduces the probability of having engaged in risky behaviour by the age of 17 between 3\% and 9\%, as illustrated in table \ref{t:hat-t-w}. All effects have strong statistical and economic significance. Compared to the average probability of engaging in risky behaviour, the estimated effect of a one standard deviation change in $T$ is largest for some of the riskiest behaviour we considered: selling drugs (-54\%), running away (-46\%), and attacking someone (-41\%).

\begin{table}[h]
\centering
\caption{Effect of $T$ on $W$}  \label{t:hat-t-w}
          \begin{tabular}{c | c c c c c c c c c} & & & try & run & attack & sell & destroy & steal & steal \\ 
                                                         & drink & smoke & marijuana & away & someone & drugs & property & $<50$\$ & $>50$\$ \\ 
                                                         \midrule
                                                         $\Pr$   &  65.5\% &  47.1\% &         29.8\% &    10.9\% &          19.0\% &       9.1\% &            32.0\% &       38.1\% &        8.0\% \\
                                                         $T$ &  -7.9\% &  -9.5\% &         -9.6\% &    -5.0\% &          -7.8\% &      -4.9\% &            -3.5\% &       -4.7\% &       -3.2\% \\
\multicolumn{10}{c}{ } \\[-1ex]
\multicolumn{10}{p{16cm}}{
\begin{minipage}{1 \linewidth}
\begin{footnotesize}
\textit{Notes}:
The table contains sample probabilities for engaging in risky behaviour by the age of 17 in the $\Pr$ row. The estimated decrease in the probability of engaging in risky behaviour from a linear probability model for a one standard deviation increase in $T$ is noted in the row corresponding to $T$.
\end{footnotesize}
\end{minipage}}
          \end{tabular}
\end{table}

$T$ captures the information we expected based on our suspicion about the unobserved confounder ability. $T$ closely reflects (academic) ability as measured by high-school GPA measures, which reduces the probability of engaging in risky behaviours during high-school. Thus, we can conclude that the common confounder $U$ contains the unobserved variable ability.

Now, an argument is required for the conditional exogeneity of instruments $Z$ given unobserved ability $U$ and observed covariates $X$:
\begin{align*}
Y &= \alpha_Y + A \beta + U \gamma_Y + W \upsilon_Y + X \eta_Y + \varepsilon_Y, & & &  \E{\left. \varepsilon_Y \right| Z, X} &= 0.
\end{align*}
While ability is the obvious confounder of the effect of pre-college GPA measures on net worth later in life, there are other possible confounders. Among everyone who goes to college, those with higher pre-college GPA are likely to also have a higher college GPA. Even conditional on whether someone obtained a BA degree, a higher college GPA likely leads to higher earnings later in life. Thus, college GPA is an important observed confounder. Family and individual net worth at young age can affect pre-college GPA measures as more learning resources are available. Their effect on net worth later in life is undeniable. Apart from net worth, other family background characteristics likely affect both pre-college test scores and net worth later in life. We include parental education, maternal age at first birth and the individual's birth, as well as the number of siblings to capture family background characteristics. Individual-specific characteristics are other important confounders. We include sex and citizenship status based on birth 
as further covariates. Conditional on this rich set of covariates $X$, and the unobserved variable ability $U$, there is no reason to believe that pre-college test scores $Z$ would affect or be correlated with post-college earnings through any other channel than obtaining a BA degree $A$. Despite our best efforts in explaining $U$, and the provided arguments in favour of assumption \ref{a:livcc}.\ref{a:livcc-iv-exog}, a test or conditional instrument exogeneity is not possible. A specification test is not feasible, because in this example the model is not overidentified.

\subsection{Estimation}     \label{sssec:ex-lin-educ-4}
Estimation of the fully linear model is now straightforward. As in \cite{tien2022icc}, we call the estimator an \emph{instrumented common confounding} (ICC) estimator.
\begin{align}
\hat{\beta}_{ICC} &= \left( A^\intercal P_Z M_{T, X} A \right)^{-1} \left( A^\intercal P_Z M_{T, X} Y \right).
\end{align}
Here, $P_Z = Z \left( Z^\intercal Z \right)^{-1} Z^\intercal$ is the projection matrix of $Z$, and $M_{T, X} = I_n - P_{T, X}$ is the annihilator matrix of $(T, X)$. In table \ref{t:est-lin}, the estimates of four major methods are compared: ordinary least squares (OLS), instrumental variables (IV), proximal learning (PL), and the here suggested ICC estimator. The row corresponding to $T$ describes the estimated partial effect of $T$ (normalised to standard deviation one) on net worth $Y$ (at 35) in the respective regressions. $T$ is only used in proximal learning and ICC, but derived from the covariation of $(Z, A)$ and $W$ in negative control \citep{cui2020}, as opposed to $Z$ and $W$ in our approach. The row corresponding to $A$ contains estimates for $\beta$, the causal effect of obtaining a BA degree $A$ on net worth $Y$ (at 35). Their unit is US Dollar. 

\begin{table}[h]
\caption{Estimates with different estimators ($Y$ in thousands (k))}
\label{t:est-lin}
\begin{center}
\begin{tabular}{l @{\hspace{2.618\tabcolsep}} | @{\hspace{2.618\tabcolsep}} c @{\hspace{4.236\tabcolsep}} c @{\hspace{4.236\tabcolsep}} c @{\hspace{4.236\tabcolsep}} c}
\toprule
                            & OLS       & PL        & IV        & ICC        \\
\midrule
$A$                       & 59.18***  & 30.90***  & 222.97*** & 125.15**   \\
                              & (9.12)    & (10.40)   & (34.74)   & (52.93)    \\
\midrule
$T$                    &           & 27.76***  &           & 16.05**    \\
                              &           & (4.82)    &           & (7.37)     \\
\bottomrule
\multicolumn{5}{c}{ } \\[-1ex]
\multicolumn{5}{p{11cm}}{
\begin{minipage}{1 \linewidth}
\begin{footnotesize}
\textit{Notes}:
The table contains estimates and their standard errors (in parantheses) for $\beta$ in the $A$ row, and the linear parameter on $T$ if used in the method from four estimators: Ordinary Least Squares (OLS), Proximal Learning (NC), Instrumental Variables (IV), and Instrumented Common Confounding (ICC). Asterisks indicate significance at the 1\% (***), 5\% (**) and 10\% (*) level.
\end{footnotesize}
\end{minipage}}
\end{tabular}
\end{center}
\end{table}

OLS estimates that obtaining a BA degree increases net worth at 35 by 59k\$. 
The proximal learning estimator conditions on its own $T$, so implicitly anything fixed that covaries $(Z, A)$ and proxies $W$. The proximal learning estimate at 31k\$, is indeed economically significantly smaller than the OLS estimate. As hypothesised, this might indicate that unobserved ability, which correlates $(Z, A)$ and $W$, is a confounder which biases the estimated effect of education on net worth upwards.
In contrast, the IV estimate is much larger at 223k\$. The inherent negative selection bias may thus be quite large. However, the IV estimator ignores the strong correlation of the pre-college test score instruments $Z$ with ability $U$, which may lead to an accentuated ability bias compared to that in OLS. As the estimator should be robust to ability bias, we condition on $T$ and obtain the ICC estimator at 125k\$. Indeed, conditioning on $T$ attenuates the estimate by the expected ability bias. As relevance is not strongly satisfied for the instruments $Z$ conditional on $T$ in the ICC estimator, the standard error is expectedly large for this method. Still, both general selection bias and ability bias appear to be strong confounders in this difficult identification problem.

Quantitatively separating ability and general selection bias helps add the necessary credibility to IV, which misses under the original IV exogeneity assumption.

\section{Conclusion}   \label{sec:conclusion}
In this work, we relax instrument exogeneity in the presence of mismeasured confounders. Other observed variables, the proxies, must be relevant for the unobserved confounders, which cause endogeneity in the instruments. The mild parametric index sufficiency assumption is also required. Importantly, the proxies can be economically meaningful variables, with their own effects on treatment and outcome. This method can be useful in various causal identification problems with observational data, where the unobserved confounders are otherwise unrestricted observed variables. The linear returns to education identification problem illustrates how this method can identify causal effects when instrument exogeneity, as often in practice, is a strong and hardly testable assumption.

This paper established two point identification results. When point identification is impossible, this approach can still identify informative bounds on causal effects. This set identification exercise is left to future work. Further, we have not demonstrated how to construct estimators other than in the linear example. Uncertainty in the control function estimation will be reflected in the performance of any estimator using this identification approach. The integration of this approach, which at best identifies averages of causal effects across unobservables, with marginal treatment effects, is another remaining task.

\clearpage

\bibliographystyle{plainnat}
\bibliography{references}

\clearpage

\section{Proofs} \label{sec:proofs}
\subsection*{Section \ref{sec:confounding} proofs}   \label{ssec:proofs-confounding}
\begin{proof}[Proof of lemma \ref{l:sc-1}]
Let $t \in {\mathcal{Z}^\tau}$. 
\begin{align*}
f_{W, Z | T}(W, Z | T) &= \int_{\mathcal{U}} f_{W, Z | U, T}(W, Z | u, T) f_{U | T}(u, T) \dif \mu_U(u)  \\
&= \int_{\mathcal{U}} f_{W | Z, U, T}(W | Z, u, T) f_{Z | U, T}(Z | u, T) f_{U | T}(u, T) \dif \mu_U(u)  \\
&= \int_{\mathcal{U}} f_{W | U}(W | u) f_{Z | T}(Z | T) f_{U | T}(u, T) \dif \mu_U(u)  \\
&= f_{Z | T}(Z | T) \int_{\mathcal{U}} f_{W | U}(W | u) f_{U | T}(u, T) \dif \mu_U(u)  \\
&= f_{Z | T}(Z | T) f_{W | T}(W | T) 
\end{align*}
From line two to three we use $W \CI Z \ | \ U$ (assumption \ref{a:riv}.\ref{a:pl-excl}) and $U \CI Z \ | \ T$ (by construction of $T$). From line four to five we again use $W \CI \tau(Z) \ | \ U$ (assumption \ref{a:riv}.\ref{a:pl-excl}).
\end{proof}

\begin{proof}[Proof of lemma \ref{l:sc-2}]
We write $f_{W | Z}(W, Z)$ in two separate ways using $T$ and relate those.
\begin{enumerate}
\item
\begin{align*}
f_{W | Z}(W, Z) &= \int_{\mathcal{U}} f_{W | U}(W, u) f_{U | Z}(u, Z) \dif \mu_U(u) \\
&= \int_{\mathcal{U}} f_{W | U}(W | u)  f_{U | T, Z}(u, t, z)  \dif \mu_U(u) 
\end{align*}
The first line follows from $W \CI Z \ | \ U$ (assumption \ref{a:riv}.\ref{a:pl-excl}). 
\item
\begin{align*}
f_{W | Z}(W, Z) &= f_{W | T}(W, T) \\
&= \int_{\mathcal{U}} f_{W | U}(W, u) f_{U | T}(u, t) \dif \mu_U(u) 
\end{align*}
The first line follows from the construction of $T=\tau(Z)$, $\tau \in \mathcal{L}_2(Z)$, such that $W \CI Z \ | \ T$.
\end{enumerate}
By the equality of the expressions in steps 1 and 2 above, for all $z \in \mathcal{Z}$,
\begin{align*}
\int_{\mathcal{U}} f_{W | U}(W, u) f_{U | T, Z}(u, t, z) \dif \mu_U(u) &= \int_{\mathcal{U}} f_{W | U}(W, u) f_{U | T}(u, t) \dif \mu_U(u), \\
\int_{\mathcal{U}} f_{U | T, Z}(u, t, z) \frac{f_W(W)}{f_U(u)} f_{U | W}(u, W) \dif \mu_U(u) &= \int_{\mathcal{U}} f_{U | T}(u, t) \frac{f_W(W)}{f_U(u)} f_{U | W}(u, W) \dif \mu_U(u).
\end{align*}
Let $g_{t,z}(U) \coloneqq \frac{f_{U | T,Z}(U, t, z)}{f_U(U)}$ and $g_t(U) \coloneqq \frac{f_{U | T}(U, t)}{f_U(U)}$ for any $z \in \mathcal{Z}$. Then,
\begin{align*}
\E{g_{t, z}(U) f_W(W) | W} &= \E{g_t(U) f_W(W) | W} \\
\E{\left. \left(g_{t, z}(U) - g_t(U) \right) \right| W} f_W(W) &= 0
\end{align*}
By completeness of $W$ for $U$ (assumption \ref{a:riv}.\ref{a:pl-rel}), the above only holds if
\begin{align*}
g_{t, z}(U) &= g_t(U).
\end{align*}
This implies $$f_{U | T,Z}(U, t, z) = f_{U | T}(U, t),$$ and thus $U \CI Z \ | \ T$ for any $T=\tau(Z)$, $\tau \in \mathcal{L}_2(Z)$ such that $W \CI Z \ | \ T$.
\end{proof}

\subsection*{Section \ref{ssec:ls} proofs}   \label{ssec:proofs-ls}
\begin{proof}[Proof of Theorem \ref{th:ls-id}]

\begin{align*}
\E{Y | Z} &= \E{\left. k_0(A) + \varepsilon \right| Z} \\
&= \E{k(A) | Z} + \E{ \E{\varepsilon | U} | Z} \\
&= \E{k(A) | Z} + \E{ \E{\varepsilon | U} | T} \\
&= \E{k(A) | Z} + \E{\varepsilon | T} \\
&= \E{k(A) + \E{\varepsilon | T} | Z} \\
h(A, T) &= k(A) + \E{\varepsilon | T}
\end{align*}
From line one to two we use the moment $\E{\varepsilon | Z, U} = \E{\varepsilon | U}$ formulated in assumption \ref{a:ls}. 
From line two to three we use $U \CI Z \ | \ T$ and $T = \tau(Z)$ (assumption \ref{a:riv}.\ref{a:iv-index}/\ref{a:iv-rel}). Required for this step is the identifiability of $T$, given by lemma \ref{l:sc-1} subject to assumption \ref{a:riv}.\ref{a:pl}.
Completeness of $Z$ for $A$ given $T$ (assumption \ref{a:riv}.\ref{a:iv-rel}) is used to obtain the final line given that $\E{Y | Z} = \E{h(A, T) | Z}$. It follows that
\begin{align*}
\int_\mathcal{A} Y(a) \pi(a) \mathrm{d}a &= \int_\mathcal{A} k(a) \pi(a) \mathrm{d}a \\
&= \int_\mathcal{A} \left( h(a, t) - \E{\varepsilon | t} \right) \pi(a) \mathrm{d}a \text{ for any } t \in \mathcal{T} \\
&= \int_\mathcal{T} \int_\mathcal{A} \left( h(a, t) - \E{\varepsilon | t} \right) \pi(a) \mathrm{d}a f(t) \mathrm{d}t \\
&= \int_\mathcal{T} \int_\mathcal{A} h(a, t) \pi(a) \mathrm{d}a f(t) \mathrm{d}t.
\end{align*}
In line one we use $\int_\mathcal{\varepsilon} \varepsilon f(\varepsilon) \mathrm{d}\varepsilon = 0$. From line one to two we use $h(A, T) = k(A) + \E{\varepsilon | T}$. From line two to three we simply extend what holds for any given $t$ to all $t \in \mathcal{T}$. From line three to four we use $\int_\mathcal{T} \E{\varepsilon | t} f(t) \mathrm{d}t = \E{\varepsilon} = 0$.

\end{proof}

\subsection*{Section \ref{ssec:mono} proofs}  \label{ssec:proofs-mono}
\begin{proof}[Proof of lemma \ref{l:7}]
\begin{align*}
F_{\eta | T}(\eta) &= \int_{\mathcal{U}} F_{\eta | u}(\eta) f_{U | T}(u, T) \dif \mu_U(u), \\
\left. \frac{ \partial F_{\eta | T}(e) }{ \partial e } \right|_{e = \eta} &= \int_{\mathcal{U}} \underbrace{\left. \frac{ \partial F_{\eta | u}(e) }{ \partial e } \right|_{e = \eta}}_{>0 \ \forall u, \eta} f_{U | T}(u, T) \dif \mu_U(u) > 0,
\end{align*}
 so $F_{\eta | T}(\eta)$ is strictly increasing on the conditional support of $\eta$. \\
Now prove $Z \CI \eta \ | \ T$.
\begin{align*}
f_{Z, \eta | T} &= \int_{\mathcal{U}} f_{Z, \eta | U, T} f_{U | T} \dif \mu_U(u) \\
 &= \int_{\mathcal{U}} f_{\eta | Z, U, T} f_{Z | U, T} f_{U | T} \dif \mu_U(u) \\
 &= \int_{\mathcal{U}} f_{\eta | U, T} f_{Z | T} f_{U | T} \dif \mu_U(u) \\
 &= f_{Z | T} \int_{\mathcal{U}} f_{\eta | U, T} f_{U | T} \dif \mu_U(u) \\
 &= f_{Z | T} f_{\eta | T}
\end{align*}
From second to third line we use $Z \CI \eta \ | \ U$ (assumption \ref{a:monotonicity}.3) and $Z \CI U \ | \ T$ (assumption \ref{a:riv}.4). We have shown that $Z \CI \eta \ | \ T$.
\end{proof}

\begin{proof}[Proof of theorem \ref{th:8}] 
First, we derive a preliminary result as if $U$ were observed.
\begin{align*}
 F_{A | Z, U}(a, z, u) &= \Pr{ \left( A \leq a | Z = z, U = u \right) } = \Pr{ \left( h(z, \eta) \leq a | Z=z, U=u \right) }  \\
&= \Pr{ \left( \eta \leq h^{-1}(a, z) | Z=z, U=u \right) }  \\
&= \Pr{ \left( \eta \leq h^{-1}(a, z) | U=u \right) }  \\
&= F_{\eta | u} \left( h^{-1}(a, z) \right) \\
&= F_{\eta | u} \left( \eta \right).
\end{align*}
From line one to two we use the invertibility of $h(z, \eta)$ (assumption \ref{a:monotonicity}.(1/2)). From line two to three we use $Z \CI \eta \ | \ U$ (assumption \ref{a:monotonicity}.3). 
Then,
\begin{align*}
V_T &\coloneqq F_{A | Z}(A, Z) = \int_{\mathcal{U}} F_{A | Z, U}(A, Z, u) f_{U | T}(u, T) \dif \mu_U(u) \\
&= \int_{\mathcal{U}} F_{\eta | U} \left( \eta \right) f_{U | T}(u, T) \dif \mu_U(u) \\
&=  F_{\eta | T} (\eta)
\end{align*}
On line one we use $Z \CI U \ | \ T$ (assumption \ref{a:riv}.\ref{a:iv-rel}). From line one to two, we use the result derived at the beginning of this proof. The final line again follows from $\tau(Z) \CI \eta \ | \ U$ (assumption \ref{a:monotonicity}.3).
Hence,
\begin{align*}
V_T &= F_{A | Z}(A, Z) = F_{\eta | T}(\eta).
\end{align*}
By lemma \ref{l:7} (strictly increasing $F_{\eta | T}$ on the support of $\eta$), $(\eta, T)$ and $(V_T, T)=(F_{\eta | T}(\eta), T)$ are associated with the same sigma algebra. \\
We show $A \CI Y(a) \ | \ (V_T, T)$.
\begin{align*}
&f_{Y(a), A | V_T, T}(Y(a), A | V_T, T) \\
&= \int_{\mathcal{U}} \underbrace{f_{Y(a) | A, V_T, T, U}(Y(a) | h(Z, \eta), V_T, T, u)}_{=f_{Y(a) | V_T, T, u}(Y(a) | V_T, T, u)} \underbrace{f_{A | T, V_T, U}(h(Z, \eta) | V_T, T, u)}_{=f_{A | V_T, T}(h(Z, \eta) | V_T, T)} f_{U | T}(u | T) \dif \mu_U(u)  \\
 &= f_{A | V_T, T}(h(Z, \eta) | V_T, T) \underbrace{\int_{\mathcal{U}} f_{Y(a) | V_T, T, u}(Y(a) | V_T, T, u) f_{U | T}(u | T) \dif \mu_U(u)}_{=f_{Y(a) | V_T, T}(Y(a) | V_T, T)}  \\
 &=  f_{A | V_T, T}(A | V_T, T)  f_{Y(a) | V_T, T}(Y(a) | V_T, T) \ \implies \ A \CI Y(a) \ | \ (V_T, T)
\end{align*}
\end{proof}

\begin{proof}[Proof of theorem \ref{th:average-id}]
$\theta_0 \coloneqq \E{ \int_{\mathcal{A}} Y(a) \pi(a) \dif \mu_A(a) }$ is identified as 
\begin{align*}
&\int_{\mathcal{V}_T, {\mathcal{Z}^\tau}} \int_{\mathcal{A}} \E{Y | A=a, (V_T, T)=(v_T, t)} \pi(a) \dif \mu_A(a) \dif F_{V_T, T}(v_T, t) \\
&\E{ \int_{\mathcal{A}} \E{Y(a) | A=a, V_T, T} \pi(a) \dif \mu_A(a) } \\
&\E{ \int_{\mathcal{A}} \E{g(A, \varepsilon) | A=a, V_T, T} \pi(a) \dif \mu_A(a) } \\
&\E{ \int_{\mathcal{A}} \int_{\mathcal{E}} g(A, \varepsilon) \dif F_{\varepsilon | A, V_T, T}(\varepsilon, a, V_T, T) \pi(a) \dif \mu_A(a) } \\
&\E{ \int_{\mathcal{A}} \int_{\mathcal{E}} g(A, \varepsilon) \dif F_{\varepsilon | V_T, T}(\varepsilon, V_T, T) \pi(a) \dif \mu_A(a) } \\
&\E{ \int_{\mathcal{E}} \int_{\mathcal{A}} g(A, \varepsilon) \pi(a) \dif \mu_A(a) \dif F_{\varepsilon | V_T, T}(\varepsilon, V_T, T) } \\
&\E{ \int_{\mathcal{A}} g(A, \varepsilon) \pi(a) \dif \mu_A(a) } \\
&\E{ \int_{\mathcal{A}} Y(a) \pi(a) \dif \mu_A(a) }  =  \theta_0.
\end{align*}
We let $Y(a) = g(a, \varepsilon)$ for some $\varepsilon \in \mathcal{E}$ and $g \in \mathcal{L}_2(A, \varepsilon)$ from line two to three. Then, $A \CI Y(a) \ | \ (V_T, T)$ from theorem \ref{th:8} implies $A \CI \varepsilon \ | \ (V_T, T)$, which we use from line four to five. All other steps are algebra.
\end{proof}

\subsection*{Section \ref{sec:semiparametric} proofs}   \label{ssec:proofs-semiparametric}

\begin{proof}[Proof of theorem \ref{th:nonlinear-strong-id}]
For any $\tau_0 \in \mathcal{T}_\text{valid}$,
\begin{align*}
\E{m(O; \tau_0 + \eta) - m(O; \tau_0)} &= \E{m(O; \eta)} \\
 &= \E{\alpha(Z) \eta(Z)}.
\end{align*}
The first line holds by linearity of $\tau \mapsto \E{m(O; \tau)}$, the second by the Riesz representation theorem.

For any $\tau_0 \in \mathcal{T}_\text{valid}$, consider any $\eta \in \mathcal{N}(P^{\mathcal{L}_2(U)}_{Z, \mathcal{T}} P_{\mathcal{L}_2(U)}^{Z, \mathcal{T}})$, i.e.
\begin{align*}
\Pi_{\mathcal{T}} \left[ {\E{\tau_0(Z) + \eta(Z) | U} | Z} \right] &= 0.
\end{align*}
Completeness of $Z$ for $U$ by construction implies $\E{h(\eta(Z); g) | U} = 0$.

For any $\eta \in \mathcal{N}(P^{\mathcal{L}_2(U)}_{Z, \mathcal{T}} P_{\mathcal{L}_2(U)}^{Z, \mathcal{T}})$ instruments $Z$ remain exogenous, which implies
\begin{align*}
0 &= \E{m(O; \eta)} 
 = \E{\alpha(Z) \eta(Z)}.
\end{align*}
The above condition holds for any $\eta \in \mathcal{N}(P^{\mathcal{L}_2(U)}_{Z, \mathcal{T}} P_{\mathcal{L}_2(U)}^{Z, \mathcal{T}})$, therefore it must be that $\alpha \in \mathcal{N}^\perp(P^{\mathcal{L}_2(U)}_{Z, \mathcal{T}} P_{\mathcal{L}_2(U)}^{Z, \mathcal{T}})$ when $\theta_0$ is identified as $\theta_0 = \E{m(O; \tau_0)}$ for any $\tau_0 \in \mathcal{T}_\text{valid}$.

Note that by $\E{g(Z) | W} = \E{\E{g(Z) | U} | W}$ for any $g \in \mathcal{G}$ and $\Pi_\mathcal{T} \left[ q(W) | Z \right] = \Pi_\mathcal{T} \left[ \E{q(W) | U} | Z \right]$ for $\mathcal{T} \subseteq \mathcal{L}_2(Z)$ for any $q \in \mathcal{L}_2(W)$ (relaxed version of assumption \ref{a:riv}.\ref{a:pl-excl}),
\begin{align*}
P^{\mathcal{L}_2(W)}_{Z, \mathcal{T}} P_{\mathcal{L}_2(W)}^{Z, \mathcal{T}} &= \Pi_\mathcal{T} \left[ {\E{\tau_0(Z) + \eta(Z) | W} | Z} \right] \\
 &= \Pi_\mathcal{T} \left[ { \E{ \E{ \E{\tau_0(Z) + \eta(Z) | U} | W} | U} | Z} \right] = P_{Z, \mathcal{T}}^{\mathcal{L}_2(U)} P_{\mathcal{L}_2(U)}^{\mathcal{L}_2(W)} P_{\mathcal{L}_2(W)}^{\mathcal{L}_2(U)} P_{\mathcal{L}_2(U)}^{Z, \mathcal{T}}.
\end{align*}
Then, by completeness of $W$ for $U$ (assumption \ref{a:riv}.\ref{a:pl-rel}),
\begin{align*}
P_{\mathcal{L}_2(U)}^{\mathcal{L}_2(W)} P_{\mathcal{L}_2(W)}^{\mathcal{L}_2(U)} P_{\mathcal{L}_2(U)}^{Z, \mathcal{T}} &= \E{ \E{ \E{\tau_0(Z) + \eta(Z) | U} | W} | U} \\
 &= \E{\tau_0(Z) + \eta(Z) | U} = P_{\mathcal{L}_2(U)}^{Z, \mathcal{T}}.
\end{align*}
Completeness of $W$ for $U$ (assumption \ref{a:riv}.\ref{a:pl-rel}) implies the above for the following reason: Only if completeness is violated can it be that $0 = \E{g(U) | W}$ (and thus $0 = \E{\E{g(U) | W} | U}$) for some $g(U) \neq 0$. Consequently, $\gamma(U) = \E{\E{\gamma(U) + g(U) | W} | U}$ only if $g(U) = 0$ when completeness holds, which implies $\gamma(U) = \E{\E{\gamma(U) | W} | U}$ when completeness holds.
 
Jointly, this concludes the proof as
\begin{align*}
\alpha \in \mathcal{N}^\perp(P^{\mathcal{L}_2(U)}_{Z, \mathcal{T}} P_{\mathcal{L}_2(U)}^{Z, \mathcal{T}}) = \mathcal{N}^\perp(P^{\mathcal{L}_2(W)}_{Z, \mathcal{T}} P_{\mathcal{L}_2(W)}^{Z, \mathcal{T}}).
\end{align*}
\end{proof}

\begin{proof}[Proof of theorem \ref{th:nonlinear-error}]
$\theta_0$ is strongly identified, so
\begin{align*}
\E{q_0(W) \tau(Z)} = \E{\alpha(Z) \tau(Z)} = \E{m(O; \tau)} \ \text{for any} \ h \in \mathcal{H}.
\end{align*}
Use this, and $\E{g_0(Z) | W} - \E{\tau_0(Z) | W}$ below:
\begin{align*}
\E{\psi(O; \tau, q)} - \theta_0 &= \E{m(O; \tau) + q(W)(g_0(Z) - \tau(Z))} - \theta_0 \\
&= \E{m(O; \tau) + q(W)(g_0(Z) - \tau(Z))} - \E{m(O; \tau_0)} - \E{q(W)(g_0(Z) - \tau_0(Z))} \\
&= \E{\alpha(Z) \tau(Z)} - \E{\alpha(Z) \tau_0(Z)} + \E{q(W)(\tau_0(Z) - \tau(Z))} \\
&= \E{\alpha(Z) (\tau(Z) - \tau_0(Z))} - \E{q(W)(\tau(Z) - \tau_0(Z))} \\
&= \E{\E{q_0(W) | Z} (\tau(Z) - \tau_0(Z))} - \E{q(W)(\tau(Z) - \tau_0(Z))} \\
&= \E{q_0(W) (\tau(Z) - \tau_0(Z))} - \E{q(W)(\tau(Z) - \tau_0(Z))} \\
&= - \E{\left( q(W) - q_0(W) \right) (\tau(Z) - \tau_0(Z))}
\end{align*}

By projecting onto $Z$ and $W$ respectively, the error rates follow:
\begin{align*}
\left| \E{\psi(O; \tau, q)} - \theta_0 \right| &= \left| \E{ \E{ q(W) - q_0(W) | Z} (\tau(Z) - \tau_0(Z))} \right| = \left| \left\langle P^{\mathcal{L}_2(W)}_{Z, \mathcal{T}} (q - q_0), \tau - \tau_0 \right\rangle \right| \\
&= \left| \E{ \left( q(W) - q_0(W) \right) \E{\tau(Z) - \tau_0(Z) | W}} \right| = \left| \left\langle q - q_0, P_{\mathcal{L}_2(W)}^{Z, \mathcal{T}} (\tau - \tau_0) \right\rangle \right|
\end{align*}

Then use the Cauchy-Schwartz inequality to obtain
\begin{align*}
\left| \E{\psi(O; \tau, q)} - \theta_0 \right| &\leq \min \left\{ \left\lVert P_{\mathcal{L}_2(W)}^{Z, \mathcal{T}} \left( \tau - \tau_0 \right) \right\rVert_2 \left\lVert q - q_0 \right\rVert_2, \  \left\lVert \tau - \tau_0 \right\rVert_2 \left\lVert P^{\mathcal{L}_2(W)}_{Z, \mathcal{T}} \left( q - q_0 \right) \right\rVert_2  \right\}
\end{align*}

Double robustness is satisfied as
\begin{align*}
\E{\psi(O; \tau, q_0)} &= \theta_0 - \E{\left( q_0(W) - q_0(W) \right) (\tau(Z) - \tau_0(Z))} = \theta_0 \\
\E{\psi(O; \tau_0, q)} &= \theta_0 - \E{\left( q(W) - q_0(W) \right) (\tau_0(Z) - \tau_0(Z))} = \theta_0.
\end{align*}

Neyman orthogonality is satisfied as $\tau_0 \in \mathcal{T}_\text{valid}$ and $q_0 \in \mathcal{Q}_0$ if and only if
\begin{align*}
\left. \frac{\partial}{\partial t}  \E{\psi(O; \tau_0 + t \tau, q_0)} \right|_{t=0} = \E{\left( \alpha(Z) - q_0(W) \right) \tau(Z)} = \E{\E{ \alpha(Z) - q_0(W) | Z} \tau(Z)} &= 0 \\
\left. \frac{\partial}{\partial t}  \E{\psi(O; \tau, q_0 + t q)} \right|_{t=0} = \E{q(W) (g_0(Z) - \tau_0(Z))} = \E{q(W) \E{g_0(Z) - \tau_0(Z) | W}} &= 0.
\end{align*}
\end{proof}

\begin{proof}[Proof of theorem \ref{th:error-decomposition}]
Want to show that:
\begin{footnotesize}
\begin{align*}
&\E{m_3(O; k, \tau, g, q_k, q_\tau, \alpha_g)} - \theta_0 \\
&= \E{m_0(O; k) + q_k(Z) \left( g(Z) - \tau(Z) - k(A; \tau, g) \right) + q_\tau(W; q_k) \left( \tau(Z; g) - g(Z) \right) + \alpha_g(Z; q_k, q_\tau) \left( g_Y(Y) - g(Z) \right)} - \theta_0 \\
&= \E{\left( q_{k, 0}(Z) - q_k(Z) \right) \left( k(A; \tau, g) - k_0(A; \tau_0, g_0) \right)} \\
& \ \ \ \ + \E{\left( q_{\tau, 0}(W; q_k) - q_\tau(W; q_k) \right) \left( \tau_0(Z; g_0) - \tau(Z; g) \right)} \\
& \ \ \ \ + \E{\left( \alpha_{g, 0}(Z; q_k, q_\tau) - \alpha_g(Z; q_k, q_\tau) \right) \left( g(Z) - g_0(Z) \right)}  
\end{align*}
\end{footnotesize}

First, deal with $\E{m_0(O; k) + q_k(Z) \left( g(Z) - \tau(Z) - k(A) \right)} - \theta_0$.
From line one to two below use that $\theta_0 = \E{m(O; k_0)}$, and that $\E{m(O; k)} = \E{\alpha_{k, 0} k(A)}$ for any $k \in \mathcal{K}$ by the Riesz representation theorem when $k \mapsto \E{m(O; k)}$ is a linear and continuous functional. From line one to two also use that $\E{k_0(A; \tau, g) | Z} = g(Z) - \tau(Z)$ for any $\tau \in \mathcal{T}$ and $g \in \mathcal{G}$. From line two to three use that $\alpha_{k, 0}(A) = \E{q_{k, 0}(Z) | A}$ for $q_{k, 0}(Z) \in \mathcal{Q}_{k, 0}$ under strong instrument relevance \ref{a:iv-strong-relevance}.
\begin{footnotesize}
\begin{align*}
&\E{m_0(O; k(\tau, g)) + q_k(Z) \left( g(Z) - \tau(Z) - k(A; \tau, g) \right)} - \theta_0 \\
&= \E{\alpha_{k, 0}(A) k(A; \tau, g)} - \E{\alpha_{k, 0}(A) k_0(A; \tau_0, g_0)} \\
& \ \ \ \ + \E{q_k(Z) \left( g(Z) - \tau(Z; g) - k(A; \tau, g) \right)} - \E{q_k(Z) \left( g_0(Z) - \tau_0(Z; g_0) - k_0(A; \tau_0, g_0) \right)} \\
&= \E{q_{k, 0}(Z) \left( k(A; \tau, g) - k_0(A; \tau_0, g_0) \right)} \\
& \ \ \ \ + \E{q_k(Z) \left( \left( g(Z) - \tau(Z; g) - k(A; \tau, g) \right) - \left( g_0(Z) - \tau_0(Z; g_0) - k_0(A; \tau_0, g_0) \right) \right)}  \\
&= \E{\left( q_{k, 0}(Z) - q_k(Z) \right) \left( k(A; \tau, g) - k_0(A; \tau_0, g_0) \right)} \\
& \ \ \ \ + \E{q_k(Z) \left( \left( g(Z) - \tau(Z; g) \right) - \left( g_0(Z) - \tau_0(Z; g_0) \right) \right)}  
\end{align*}
\end{footnotesize}

Now, deal with the terms involving $\tau_0$ and $\tau$ from the previous step, and $\E{q_\tau(W; q_k) \left( \tau(Z; g) - g(Z) \right)}$ from $m_3$.
From line one to two use theorem \ref{th:nonlinear-strong-id}, which implies $\E{q_k(Z) \tau(Z; g)} = \E{\alpha_{\tau, 0}(Z; q_k) \tau(Z; g)}$ with $\alpha_{\tau, 0}(q_k) \in \mathcal{N}^\perp(P^{\mathcal{L}_2(W)}_{Z, \mathcal{T}} P_{\mathcal{L}_2(W)}^{Z, \mathcal{T}})$. Also, from line one to two use $\E{g(Z) | W} = \E{\tau_0(Z; g) | W}$ for any $g \in \mathcal{G}$.
\begin{footnotesize}
\begin{align*}
&\E{q_k(Z) \left( \tau_0(Z; g_0) - \tau(Z; g) \right)} + \E{q_\tau(W; q_k) \left( \tau(Z; g) - g(Z) \right)} \\
&= \E{q_{\tau, 0}(W; q_k) \left( \tau_0(Z; g_0) - \tau(Z; g) \right)} + \E{q_\tau(W; q_k) \left( \tau(Z; g) - g(Z) \right)} - \E{q_\tau(W; q_k) \left( \tau_0(Z; g_0) - g_0(Z) \right)} \\
&= \E{\left( q_{\tau, 0}(W; q_k) - q_\tau(W; q_k) \right) \left( \tau_0(Z; g_0) - \tau(Z; g) \right)} + \E{q_\tau(W; q_k) \left( g_0(Z) - g(Z) \right)} 
\end{align*}
\end{footnotesize}

Finally, deal with the remaining terms involving $g$ and $g_0$. From line to two use that $\alpha_{g, 0}(Z; q_k, q_\tau) = \E{\left(q_k(Z) - q_\tau(W; q_k)\right) | Z}$. Also use that $\E{g_Y(Y) | Z} = g_0(Z)$.
\begin{footnotesize}
\begin{align*}
&\E{\left(q_k(Z) - q_\tau(W; q_k)\right) \left( g(Z) - g_0(Z) \right)} + \E{\alpha_g(Z; q_k, q_\tau) \left( g_Y(Y) - g(Z) \right)} \\
&= \E{\alpha_{g, 0}(Z; q_k, q_\tau) \left( g(Z) - g_0(Z) \right)} + \E{\alpha_g(Z; q_k, q_\tau) \left( g_Y(Y) - g(Z) \right)} - \E{\alpha_g(Z; q_k, q_\tau) \left( g_Y(Y) - g_0(Z) \right)}  \\
&= \E{\alpha_{g, 0}(Z; q_k, q_\tau) \left( g(Z) - g_0(Z) \right)} + \E{\alpha_g(Z; q_k, q_\tau) \left( g_0(Z) - g(Z) \right)}  \\
&= \E{\left( \alpha_{g, 0}(Z; q_k, q_\tau) - \alpha_g(Z; q_k, q_\tau) \right) \left( g(Z) - g_0(Z) \right)}   
\end{align*}
\end{footnotesize}

Join all of the above to complete the proof:
\begin{footnotesize}
\begin{align*}
\E{m_3(O; k, \tau, g, q_k, q_\tau, \alpha_g)} - \theta_0 
&= \E{\left( q_{k, 0}(Z) - q_k(Z) \right) \left( k(A) - k_0(A) \right)} \\
& \ \ \ \ + \E{\left( q_{\tau, 0}(W; q_k) - q_\tau(W; q_k) \right) \left( \tau_0(Z) - \tau(Z) \right)} \\
& \ \ \ \ + \E{\left( \alpha_{g, 0}(Z; q_k, q_\tau) - \alpha_g(Z; q_k, q_\tau) \right) \left( g(Z) - g_0(Z) \right)}.
\end{align*}
\end{footnotesize}
\end{proof}

\begin{proof}[Proof of corollary \ref{co:double-robustness}]
The conditions of theorem hold, \ref{th:error-decomposition} hold, so
\begin{align*}
\E{m_3(O; k, \tau, g, q_k, q_\tau, \alpha_g)} - \theta_0 
&= \underbrace{\E{\left( q_{k, 0}(Z) - q_k(Z) \right) \left( k(A; \tau, g) - k_0(A; \tau_0, g_0) \right)}}_{\text{Term 1}} \\
& \ \ \ \ + \underbrace{\E{\left( q_{\tau, 0}(W; q_k) - q_\tau(W; q_k) \right) \left( \tau_0(Z; g_0) - \tau(Z; g) \right)}}_{\text{Term 2}} \\
& \ \ \ \ + \underbrace{\E{\left( \alpha_{g, 0}(Z; q_k, q_\tau) - \alpha_g(Z; q_k, q_\tau) \right) \left( g(Z) - g_0(Z) \right)}}_{\text{Term 3}}.
\end{align*}
\begin{enumerate}
  \item Term 1 is zero whenever $\left( (k = k_0(\tau_0, g_0)) \lor (q_k = q_{k, 0}) \right)$, irrespective of the other two terms.
  \item Term 2 is zero whenever $\left( (\tau = \tau_0(g_0)) \lor (q_\tau = q_{\tau, 0}) \right)$, irrespective of the other two terms.
  \item Term 3 is zero whenever $\left( (g = g_0) \lor (\alpha_g = \alpha_{g, 0}) \right)$, irrespective of the other two terms.
\end{enumerate}
This completes the proof.
\end{proof}

\begin{proof}[Proof of corollary \ref{co:neyman-orthogonality}]
The conditions of theorem hold, \ref{th:error-decomposition} hold, so
\begin{align*}
\E{m_3(O; k, \tau, g, q_k, q_\tau, \alpha_g)} - \theta_0 
&= \E{\left( q_{k, 0}(Z) - q_k(Z) \right) \left( k(A; \tau, g) - k_0(A; \tau_0, g_0) \right)} \\
& \ \ \ \ + \E{\left( q_{\tau, 0}(W; q_k) - q_\tau(W; q_k) \right) \left( \tau_0(Z; g_0) - \tau(Z; g) \right)} \\
& \ \ \ \ + \E{\left( \alpha_{g, 0}(Z; q_k, q_\tau) - \alpha_g(Z; q_k, q_\tau) \right) \left( g(Z) - g_0(Z) \right)}.
\end{align*}

Let $\partial_g f(g, h) \coloneqq \left. \frac{\partial}{\partial t} f(.; g_0 + t g, h) \right|_{t=0}$ for any functions $f$, $g$, and $h$.
\begin{small}
\begin{align*}
&\left. \frac{\partial}{\partial t}  \E{m_3(O; k_0 + t k, \tau, g, q_k, q_\tau, \alpha_g)} \right|_{t=0} = \E{ q_{k, 0}(Z) - q_k(Z) } = 0 \text{ if } q_k = q_{k, 0}, \\
&\left. \frac{\partial}{\partial t}  \E{m_3(O; k, \tau_0 + t \tau, g, q_k, q_\tau, \alpha_g)} \right|_{t=0} \\
 & \ \ \ \ = \E{ \left( q_{k, 0}(Z) - q_k(Z) \right) \partial_\tau k(A; \tau, g) + \left( q_{\tau, 0}(W; q_k) - q_\tau(W; q_k) \right) } \\
 & \ \ \ \ = 0 \text{ if } \left( q_\tau = q_{\tau, 0} \right) \land \left( \left( \partial_\tau k(\tau, g) = 0 \right) \lor \left( q_k = q_{k, 0} \right) \right), \\
&\left. \frac{\partial}{\partial t}  \E{m_3(O; k, \tau, g_0 + t g, q_k, q_\tau, \alpha_g)} \right|_{t=0} \\
 & \ \ \ \ = \E{ \left( q_{k, 0}(Z) - q_k(Z) \right) \partial_g k(A; \tau, g) - \left( q_{\tau, 0}(W; q_k) - q_\tau(W; q_k) \right) \partial_g \tau(Z; g) + \left( \alpha_{g, 0}(Z; q_k, q_\tau) - \alpha_g(Z; q_k, q_\tau) \right) } \\
 & \ \ \ \ = 0 \text{ if } \left( \alpha_g = \alpha_{g, 0} \right) \land \left( (\partial_g k(\tau, g) = 0) \lor (q_k = q_{k, 0}) \right) \land \left( (\partial_g \tau(g) = 0) \lor (q_\tau = q_{\tau, 0}) \right), \\
&\left. \frac{\partial}{\partial t}  \E{m_3(O; k, \tau, g, q_{k, 0} + t q_k, q_\tau, \alpha_g)} \right|_{t=0} \\
 & \ \ \ \ = \E{- \left( k(A; \tau, g) - k_0(A; \tau_0, g_0) \right) + \left( \partial_{q_k} \left( q_{\tau, 0}(W; q_k) - q_\tau(W; q_k) \right) \right) \left( \tau_0(Z; g_0) - \tau(Z; g) \right)} \\
 & \ \ \ \ \ \ \ \ + \E{ \left( \partial_{q_k} \left( \alpha_{g, 0}(Z; q_k, q_\tau) - \alpha_g(Z; q_k, q_\tau) \right) \left( g(Z) - g_0(Z) \right) \right) } \\
 & \ \ \ \ = 0 \text{ if } \left( k = k_0(\tau_0, g_0) \right) \land \left( (\tau = \tau_0(g_0)) \lor (q_\tau = q_{\tau, 0}) \right) \land \left( (g = g_0) \lor (\alpha_g = \alpha_{g, 0}) \right), \\
&\left. \frac{\partial}{\partial t}  \E{m_3(O; k, \tau, g, q_k, q_{\tau, 0} + t q_\tau, \alpha_g)} \right|_{t=0} \\
 & \ \ \ \ = \E{ - \left( \tau_0(Z; g_0) - \tau(Z; g) \right) + \partial_{q_\tau} \left( \alpha_{g, 0}(Z; q_k, q_\tau) - \alpha_g(Z; q_k, q_\tau) \right) \left( g(Z) - g_0(Z) \right) } \\
 & \ \ \ \ = 0 \text{ if } \left( \tau = \tau_0(g_0) \right) \land \left( (g = g_0) \lor (\alpha_g = \alpha_{g, 0}) \right), \\
&\left. \frac{\partial}{\partial t}  \E{m_3(O; k, \tau, g, q_k, q_\tau, \alpha_{g, 0} + t \alpha_g)} \right|_{t=0} = \E{ - \left( g(Z) - g_0(Z) \right) } = 0 \text{ if } (g = g_0).
\end{align*}
\end{small}
\end{proof}

\begin{proof}[Proof of corollary \ref{co:error-bounds}]
The conditions of theorem hold, \ref{th:error-decomposition} hold, so
\begin{align*}
\E{m_3(O; k, \tau, g, q_k, q_\tau, \alpha_g)} - \theta_0 
&= \underbrace{\E{\left( q_{k, 0}(Z) - q_k(Z) \right) \left( k(A; \tau, g) - k_0(A; \tau_0, g_0) \right)}}_{\text{Term 1}} \\
& \ \ \ \ + \underbrace{\E{\left( q_{\tau, 0}(W; q_k) - q_\tau(W; q_k) \right) \left( \tau_0(Z; g_0) - \tau(Z; g) \right)}}_{\text{Term 2}} \\
& \ \ \ \ + \underbrace{\E{\left( \alpha_{g, 0}(Z; q_k, q_\tau) - \alpha_g(Z; q_k, q_\tau) \right) \left( g(Z) - g_0(Z) \right)}}_{\text{Term 3}}.
\end{align*}
Each of these terms can be bounded using the Cauchy-Schwarz inequality.

Term 1:
\begin{small}
\begin{align*}
&\E{\left( q_{k, 0}(Z) - q_k(Z) \right) \left( k(A; \tau, g) - k_0(A; \tau_0, g_0) \right)} \\
 & \ \ \ \ = \E{\left( q_{k, 0}(Z) - q_k(Z) \right) P^{A, \mathcal{K}}_{\mathcal{L}_2(Z)} \left[ \left( k(A; \tau, g) - k_0(A; \tau_0, g_0) \right) \right]} \\
 &\ \ \ \ \ \ \ \ = \left\langle (q_{k, 0} - q_k), P^{A, \mathcal{K}}_{\mathcal{L}_2(Z)} \left[ k(\tau, g) - k_0(\tau_0, g_0) \right] \right\rangle \\
  & \ \ \ \ = \E{ P_{A, \mathcal{K}}^{\mathcal{L}_2(Z)} \left[ \left( q_{k, 0}(Z) - q_k(Z) \right) \right]  \left( k(A; \tau, g) - k_0(A; \tau_0, g_0) \right)} \\
 &\ \ \ \ \ \ \ \ = \left\langle P_{A, \mathcal{K}}^{\mathcal{L}_2(Z)} \left[ q_{k, 0} - q_k \right], \left( k(\tau, g) - k_0(\tau_0, g_0) \right) \right\rangle \\
  & \ \ \ \ \leq \min \left\{ \left\lVert P^{A, \mathcal{K}}_{\mathcal{L}_2(Z)} \left(k_0(\tau_0, g_0) - k(\tau, g)\right) \right\rVert_2 \left\lVert q_{k, 0} - q_k \right\rVert_2, \ \left\lVert k_0(\tau_0, g_0) - k(\tau, g) \right\rVert_2 \left\lVert P_{A, \mathcal{K}}^{\mathcal{L}_2(Z)} \left( q_{k, 0} - q_k \right) \right\rVert_2  \right\} 
\end{align*}
\end{small}

Term 2:
\begin{small}
\begin{align*}
&\E{\left( q_{\tau, 0}(W; q_k) - q_\tau(W; q_k) \right) \left( \tau_0(Z; g_0) - \tau(Z; g) \right)} \\
 & \ \ \ \ = \E{\left( q_{\tau, 0}(W; q_k) - q_\tau(W; q_k) \right) P^{\mathcal{L}_2(W)}_{Z, \mathcal{T}} \left[ \tau_0(Z; g_0) - \tau(Z; g) \right]} \\
  & \ \ \ \ \ \ \ = \left\langle \left( q_{\tau, 0}(q_k) - q_\tau(q_k) \right), P^{\mathcal{L}_2(W)}_{Z, \mathcal{T}} \left[ \tau_0(g_0) - \tau(g) \right] \right\rangle \\
  & \ \ \ \ = \E{ P_{\mathcal{L}_2(W)}^{Z, \mathcal{T}} \left[ q_{\tau, 0}(W; q_k) - q_\tau(W; q_k) \right]  \left( \tau_0(Z; g_0) - \tau(Z; g) \right)} \\
  & \ \ \ \ \ \ \ = \left\langle P_{\mathcal{L}_2(W)}^{Z, \mathcal{T}}  \left[ q_{\tau, 0}(q_k) - q_\tau(q_k) \right], \left( \tau_0(g_0) - \tau(g) \right) \right\rangle \\
  & \ \ \ \ \leq \min \left\{ \left\lVert P^{\mathcal{L}_2(W)}_{Z, \mathcal{T}} \left(\tau_0(g_0) - \tau(g) \right) \right\rVert_2 \left\lVert q_{\tau, 0}(q_k) - q_\tau(q_k) \right\rVert_2, \ \left\lVert \tau_0(g_0) - \tau(g) \right\rVert_2 \left\lVert P_{\mathcal{L}_2(W)}^{Z, \mathcal{T}}  \left( q_{\tau, 0}(q_k) - q_\tau(q_k) \right) \right\rVert_2  \right\}
\end{align*}
\end{small}

Term 3:
\begin{small}
\begin{align*}
&\E{\left( \alpha_{g, 0}(Z; q_k, q_\tau) - \alpha_g(Z; q_k, q_\tau) \right) \left( g(Z) - g_0(Z) \right)} \\
 & \ \ \ \ = \left\langle \left( \alpha_{g, 0}(Z; q_k, q_\tau) - \alpha_g(Z; q_k, q_\tau) \right), \left(g - g_0 \right) \right\rangle \\
 & \ \ \ \ \leq \left\lVert \left(g - g_0 \right) \right\rVert_2 \left\lVert \alpha_{g, 0}(Z; q_k, q_\tau) - \alpha_g(Z; q_k, q_\tau) \right\rVert_2
\end{align*}
\end{small}

Joining these terms, we obtain
\begin{small}
\begin{align*}
&\E{m_3(O; k, \tau, g, q_k, q_\tau, \alpha_g)} - \theta_0 \\
&\ \ \ \ \leq  \min \left\{ \left\lVert P^{A, \mathcal{K}}_{\mathcal{L}_2(Z)} \left(k_0(\tau_0, g_0) - k(\tau, g)\right) \right\rVert_2 \left\lVert q_{k, 0} - q_k \right\rVert_2, \ \left\lVert k_0(\tau_0, g_0) - k(\tau, g) \right\rVert_2 \left\lVert P_{A, \mathcal{K}}^{\mathcal{L}_2(Z)} \left( q_{k, 0} - q_k \right) \right\rVert_2  \right\}  \\
&\ \ \ \ \ \  + \min \left\{ \left\lVert P^{\mathcal{L}_2(W)}_{Z, \mathcal{T}} \left(\tau_0(g_0) - \tau(g) \right) \right\rVert_2 \left\lVert q_{\tau, 0}(q_k) - q_\tau(q_k) \right\rVert_2, \ \left\lVert \tau_0(g_0) - \tau(g) \right\rVert_2 \left\lVert P_{\mathcal{L}_2(W)}^{Z, \mathcal{T}}  \left( q_{\tau, 0}(q_k) - q_\tau(q_k) \right) \right\rVert_2  \right\} \\
&\ \ \ \ \ \  + \left\lVert \left(g - g_0 \right) \right\rVert_2 \left\lVert \alpha_{g, 0}(Z; q_k, q_\tau) - \alpha_g(Z; q_k, q_\tau) \right\rVert_2.
\end{align*}
\end{small}
\end{proof}

\begin{proof}[Proof of theorem \ref{th:dr-est}]

Let $\hat{\theta}_k = \frac{1}{|\mathcal{I}_j|} \sum_{i \in \mathcal{I}_j} m_3 \left(O; \hat{k}^{(j)}, \hat{\tau}^{(j)}, \hat{g}^{(j)}, \hat{q}_k^{(j)}, \hat{q}_\tau^{(j)}, \hat{\alpha}_g^{(j)} \right)$. Also, let $H = (k, \tau, g, q_k, q_\tau, \alpha_g)$, and $H_0 = (k_0, \tau_0, g_0, q_{k,0}, q_{\tau,0}, \alpha_{g,0})$ as well as $\hat{H}^{(j)} = (\hat{k}^{(j)}, \hat{\tau}^{(j)}, \hat{g}^{(j)}, \hat{q}_k^{(j)}, \hat{q}_\tau^{(j)}, \hat{\alpha}_g^{(j)})$, so $\hat{\theta}_k =  \mathbb{E}_{n, k} \left[ m_3(O; \hat{H}^{(j)}) \right]$.
Expand the above to get
\begin{small}
\begin{align*}
\hat{\theta}_k - \theta_0 &= \mathbb{E}_{n, k} \left[ m_3(O; \hat{H}^{(j)}) - \theta_0 \right] \\
 &= \mathbb{E}_{n, k} \left[ m_3(O; H_0) - \theta_0 \right] + \mathbb{E}_{n, k} \left[ m_3(O; \hat{H}^{(j)}) - m_3(O; H_0) \right] \\
 &= \mathbb{E}_{n, k} \left[ m_3(O; H_0) - \theta_0 \right] + \left( \mathbb{E}_{n, k} - \mathbb{E} \right) \left[ m_3(O; \hat{H}^{(j)}) - m_3(O; H_0) \right] + \mathbb{E} \left[ m_3(O; \hat{H}^{(j)}) - m_3(O; H_0) \right].
\end{align*}
\end{small}

By assumption \ref{a:rate-conditions}, $\lVert \hat{\eta}^{(j)} - \eta_0 \rVert_2^2 = o_p(1)$ for any $\eta \in H$. It follows from simple algebra that there exists a universal constant such that 
\begin{align*}
\mathbb{E} \left[ \left( m_3(O; \hat{H}^{(j)}) - m_3(O; H_0) \right)^2 \right] \leq c \sum_{\eta \in H} \lVert \hat{\eta}^{(j)} - \eta_0 \rVert_2^2 = o_p(1).
\end{align*}
Then, using the Markov inequality,
\begin{align*}
\left| \left( \mathbb{E}_{n, k} - \mathbb{E} \right) \left[ m_3(O; \hat{H}^{(j)}) - m_3(O; H_0) \right] \right| = o_p \left( | \mathcal{I}_j |^{-1/2} \right) = o_p \left(n^{-1/2}\right).
\end{align*}
Dealing with $ \mathbb{E} \left[ m_3(O; \hat{H}^{(j)}) - m_3(O; H_0) \right]$ requires assumptions on the rates on the components of its error decomposition. By corollary \ref{co:error-bounds},
\begin{tiny}
\begin{align*}
&\left| \mathbb{E} \left[ m_3(O; \hat{H}^{(j)}) - m_3(O; H_0) \right] \right| \\
&\  \leq  \min \left\{ \left\lVert P^{A, \mathcal{K}}_{\mathcal{L}_2(Z)} \left(k_0(\tau_0, g_0) - \hat{k}^{(j)}(\hat{\tau}^{(j)}, \hat{g}^{(j)})\right) \right\rVert_2 \left\lVert q_{k, 0} - \hat{q}_k^{(j)} \right\rVert_2, \ \left\lVert k_0(\tau_0, g_0) - \hat{k}^{(j)}(\hat{\tau}^{(j)}, \hat{g}^{(j)}) \right\rVert_2 \left\lVert P_{A, \mathcal{K}}^{\mathcal{L}_2(Z)} \left( q_{k, 0} - \hat{q}_k^{(j)} \right) \right\rVert_2  \right\}  \\
&\ \  + \min \left\{ \left\lVert P_{\mathcal{L}_2(W)}^{Z, \mathcal{T}} \left(\tau_0(g_0) - \hat{\tau}^{(j)}(\hat{g}^{(j)}) \right) \right\rVert_2 \left\lVert q_{\tau, 0}(\hat{q}_k^{(j)}) - \hat{q}_\tau^{(j)}(\hat{q}_k^{(j)}) \right\rVert_2, \ \left\lVert \tau_0(g_0) - \hat{\tau}^{(j)}(\hat{g}^{(j)}) \right\rVert_2 \left\lVert P^{\mathcal{L}_2(W)}_{Z, \mathcal{T}}  \left( q_{\tau, 0}(\hat{q}_k^{(j)}) - \hat{q}_\tau^{(j)}(\hat{q}_k^{(j)}) \right) \right\rVert_2  \right\} \\
&\ \  + \left\lVert \left(\hat{g}^{(j)} - g_0 \right) \right\rVert_2 \left\lVert \alpha_{g, 0}(Z; \hat{q}_k^{(j)}, \hat{q}_\tau^{(j)}) - \hat{\alpha}_g^{(j)}(Z; \hat{q}_k^{(j)}, \hat{q}_\tau^{(j)}) \right\rVert_2.
\end{align*}
\end{tiny}
Using definition \ref{d:l2-convergence-rates}, the above can be written as
\begin{align*}
&\left| \mathbb{E} \left[ m_3(O; \hat{H}^{(j)}) - m_3(O; H_0) \right] \right| \\
&\  \leq \min \left\{ \epsilon_{n, k, \mathcal{L}_2(Z)} \epsilon_{n, q_k}, \epsilon_{n, k} \epsilon_{n, q_k, \mathcal{K}}  \right\} + \min \left\{ \epsilon_{n, \tau, \mathcal{L}_2(W)} \epsilon_{n, q_\tau}, \epsilon_{n, \tau} \epsilon_{n, q_\tau, \mathcal{T}}  \right\} + \epsilon_{n, g} \epsilon_{n, \alpha_g} = o_p \left( n^{-1/2} \right),
\end{align*}
where the final equality holds by the assumption on L2 rate conditions, assumption \ref{a:rate-conditions}.2.

This simplifies the problem to 
\begin{align*}
\hat{\theta}_k - \theta_0 &= \mathbb{E}_{n, k} \left[ m_3(O; H_0) - \theta_0 \right] + o_p \left( n^{-1/2} \right) \\
\sqrt{n} \left( \hat{\theta}_k - \theta_0 \right) &= \frac{1}{\sqrt{n}} \sum_{i=1}^n \left( m_3(O_i; H_0) - \theta_0 \right) + o_p(1).
\end{align*}
By the Central Limit Theorem, as $n \rightarrow \infty$,
\begin{align*}
\sqrt{n} \left( \hat{\theta}_n - \theta_0  \right) &\rightarrow N(0, \sigma_0^2), & \sigma_0^2 &= \E{\left( m_3(O; H_0) - \theta_0 \right)^2}.
\end{align*}
By Slutsky's theorem, as long as $\hat{\sigma}^2_n \rightarrow \sigma^2_0$ in probability for an estimator $\hat{\sigma}^2_n$ of $\sigma^2_0$, it holds that as $n \rightarrow \infty$,
\begin{align*}
\frac{\sqrt{n}}{\hat{\sigma}^2_n} \left( \hat{\theta}_n - \theta_0  \right) &\rightarrow N(0, 1).
\end{align*}

\end{proof}

\section{Data Description} \label{sec:data-desc}
The sample consists of 1,983 individuals.

\subsubsection*{$Y$: Household net worth at 35 (Z9141400)}
Household net worth was top-coded at 600,000\$ and bottom-coded at -300,000\$. 7.0\% of individuals were top-coded, 0.3\% bottom-coded.

\begin{figure}[h]
\begin{center}
\caption{Histogram for $Y$}  \label{f:yhist}
\includegraphics[width=0.786\textwidth]{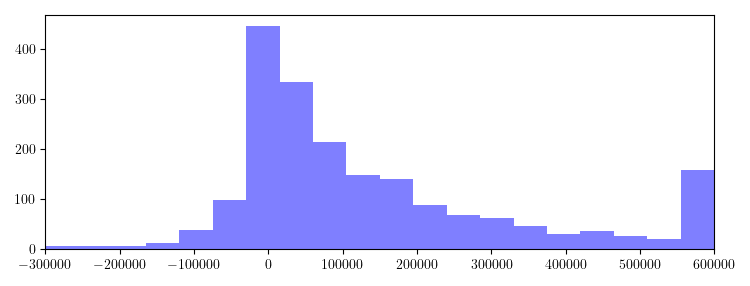}
\end{center}
\begin{footnotesize}
\textit{Notes}:
Distribution of household net worth.
\end{footnotesize}
\end{figure}

\subsubsection*{$A$: Bachelor's degree obtained (Z9084400)}
If there is a date of obtaining a bachelor's degree (Z9084400 $\geq 0$ or invalid skip $-3$), $A=1$. 50.0\% of individuals in the sample have obtained a BA degree.

\subsubsection*{$Z$: Pre-college ability measures}
Instruments are credit-weighted high-school GPAs in English (R9872000), Math (R9872200), Social Sciences (R9872300) and Life Sciences (R9872400), as well as the ASVAB percentile in each individual's respective age group.

\begin{figure}[h]
\begin{center}
\caption{Histogram for high-school GPA}  \label{f:gpahist}
\includegraphics[width=0.786\textwidth]{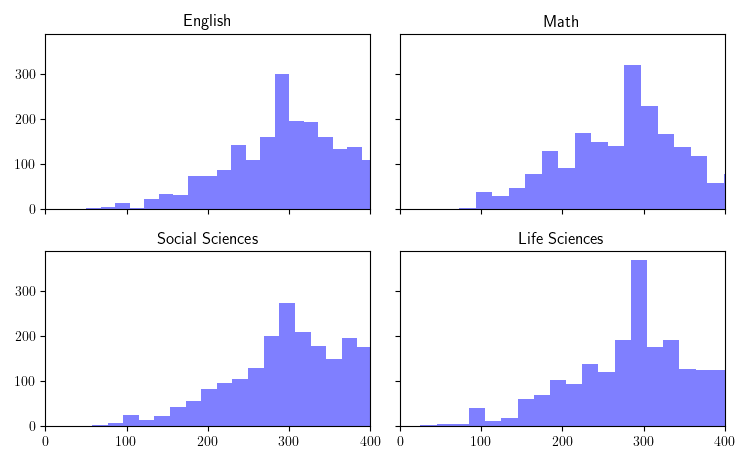}
\end{center}
\begin{footnotesize}
\textit{Notes}:
Distribution of household high-school GPAs.
\end{footnotesize}
\end{figure}

\begin{figure}[h]
\begin{center}
\caption{Histogram for ASVAB percentile in 3-month age cohort}  \label{f:asvabhist}
\includegraphics[width=0.786\textwidth]{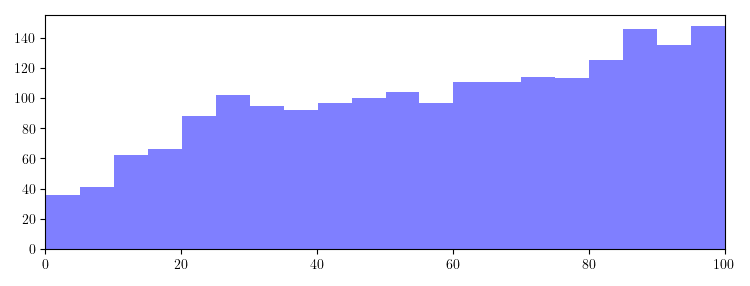}
\end{center}
\begin{footnotesize}
\textit{Notes}:
Distribution of ASVAB percentiles.
\end{footnotesize}
\end{figure}

\subsubsection*{$W$: Pre-college risky behaviour}
The proxies consist of risky behaviour dummies. They equal one when an individual has engaged in behaviour considered "risky" by age 17 or earlier if missing.

\begin{table}[h]
\centering
\caption{Probability of engaging in risky behaviour $W$ by 17 (in \%)}  \label{t:w-pr}
          \begin{tabular}{c | c c c c c c c c c} & & & mari- & run & attack & sell & destroy & steal & steal \\ 
                                                         & drink & smoke & juana & away & someone & drugs & property & $<50$\$ & $>50$\$ \\ 
                                                         \midrule
                                                         $\Pr$   &  65.5 &  47.1 &         29.8 &    10.9 &          19.0 &       9.1 &            32.0 &       38.1 &        8.0 \\
          \end{tabular}
\end{table}

\begin{table}[h]
\centering
\caption{Correlation matrix of pre-college risky behaviour $W$ (in \%)}  \label{t:w-pr}
          \begin{tabular}{l | c c c c c c c c c} & & & mari- & run & attack & sell & destroy & steal & steal \\ 
                                                         & drink & smoke & juana & away & someone & drugs & property & $<50$\$ & $>50$\$ \\ 
\midrule
drink            &            100 &              49 &             41 &                18 &              15 &                  21 &                        24 &                 28 &                 14 \\
smoke          &             49 &             100 &             49 &                21 &              20 &                  28 &                        27 &                 33 &                 17 \\
marijuana            &             41 &              49 &            100 &                26 &              24 &                  43 &                        29 &                 35 &                 21 \\
run away         &             18 &              21 &             26 &               100 &              27 &                  22 &                        20 &                 17 &                 23 \\
attack           &             15 &              20 &             24 &                27 &             100 &                  27 &                        28 &                 24 &                 18 \\
sell drugs       &             21 &              28 &             43 &                22 &              27 &                 100 &                        30 &                 27 &                 27 \\
destroy property &             24 &              27 &             29 &                20 &              28 &                  30 &                       100 &                 39 &                 21 \\
steal $<50\$$        &             28 &              33 &             35 &                17 &              24 &                  27 &                        39 &                100 &                 29 \\
steal $>50\$$        &             14 &              17 &             21 &                23 &              18 &                  27 &                        21 &                 29 &                100 \\
\bottomrule
\end{tabular}
\end{table}

\subsubsection*{$X$: Individual, family, and regional covariates}
The proxies consist of risky behaviour dummies. They equal one when an individual has engaged in behaviour considered "risky" by age 17 or earlier if missing.

\begin{table}[h]
\centering
\caption{Covariates in first stage}  \label{t:na}
          \begin{tabular}{l | l l l l} 
          variable & NLS97 basis & info & replace invalid by  \\ 
          \midrule
          hh net worth 1997 & R1204700 &  & 35\% quantile (8834.35) \\
          male & R0536300 & 46.6\% &  \\
          citizen birth non us & R1201300 & 2.3\% (not US-born)  &   \\ 
          citizen birth other & R1201300 & 13.2\% (not det.)  &   \\ 
          bio mom age first birth & R1200100 &  & median (23)  \\ 
          bio mom age subject birth & R1200100 & & median (26) \\
          bio mom educ & R1302500 & & median (13) \\
          bio dad educ & R1302400 & & median (12) \\
          relation parent figure & R1205300 & & both biological parents (1) \\
          parent religious & R1486900 & & median (3) \\
          n siblings & T6745900 & & median (2) \\
          region north central at 17 & R1200300 (i.a.) & 27.8\% & \\
          region south at 17 & R1200300 (i.a.) & 36.3\% & \\
          region west at 17 & R1200300 (i.a.) & 21.6\% & \\
          urban at 17 & R1217500 (i.a.) & 69.8\% & \\
          other non rural at 17 & R1217500 (i.a.) & 0.3\% & \\
          non english home & R0551900 & 14.1\% & 0  \\
          poor english interview & R2394903 & 0.3\% & 0  \\
          \end{tabular}
\end{table}

\begin{table}[h]
\centering
\caption{Additional covariates in outcome model}  \label{t:na}
          \begin{tabular}{l | l l l l} 
          variable & NLS97 basis & info & replace invalid by  \\ 
          \midrule
          net worth age 20 & Z9048900 &  & 35\% quantile (2737.7) \\
          gpa college & B0004600 &  & 400+ as 400 \\
          region north central at 35 & U0001900 (i.a.) & 24.5\% & \\
          region south at 35 & U0001900 (i.a.) & 38.7\% & \\
          region west at 35 & U0001900 (i.a.) & 23.2\% & \\
          urban at 35 & U0015000 (i.a.) & 81.9\% & \\
          other non rural at 35 & U0015000 (i.a.) & 0.4\% & \\
          \end{tabular}
\end{table}

\clearpage

\begin{table}
\caption{OLS coefficients key variables}
\label{}
\begin{center}
\begin{tabular}{lllll}
\hline
                              & OLS       & NC        & IV        & ICC        \\
\hline
const                         & -32.36    & -30.21    & 46.85     & 13.52      \\
                              & (37.41)   & (36.21)   & (43.38)   & (45.10)    \\
$T$                    &           & 27.76***  &           & 16.05**    \\
                              &           & (4.82)    &           & (7.37)     \\
$A$                       & 59.18***  & 30.90***  & 222.97*** & 125.15**   \\
                              & (9.12)    & (10.40)   & (34.74)   & (52.93)    \\
\hline
\multicolumn{5}{c}{ } \\[-1ex]
\multicolumn{5}{p{11cm}}{
\begin{minipage}{1 \linewidth}
\begin{footnotesize}
\textit{Notes}:
The table contains estimates and their standard errors (in parantheses) for $\beta$ in the $A$ row, and the linear parameter on $T$ if used in the method from four estimators: Ordinary Least Squares (OLS), Proximal Learning (NC), Instrumental Variables (IV), and Instrumented Common Confounding (ICC). Asterisks indicate significance at the 1\% (***), 5\% (**) and 10\% (*) level.
\end{footnotesize}
\end{minipage}}
\end{tabular}
\end{center}
\end{table}

\begin{table}
\caption{OLS coefficients individual- and family-covariates}
\label{}
\begin{center}
\begin{tabular}{lllll}
\hline
                              & OLS       & NC        & IV        & ICC        \\
\hline
hh net worth 1997          & 0.15***   & 0.16***   & 0.12***   & 0.14***    \\
                              & (0.03)    & (0.03)    & (0.03)    & (0.03)     \\
net worth age 20           & 0.35***   & 0.36***   & 0.35***   & 0.36***    \\
                              & (0.09)    & (0.08)    & (0.09)    & (0.08)     \\
gpa college                  & 16.00***  & 11.58***  & -6.01     & 2.86       \\
                              & (3.44)    & (4.00)    & (6.19)    & (6.94)     \\
male                          & 30.19***  & 24.33***  & 40.32***  & 32.06***   \\
                              & (8.39)    & (7.94)    & (9.14)    & (9.05)     \\
n siblings                   & -6.29***  & -6.78***  & -6.08**   & -6.47***   \\
                              & (1.73)    & (2.16)    & (2.37)    & (2.21)     \\
citizen birth non us       & -5.50     & 1.61      & -16.48    & -7.42      \\
                              & (27.65)   & (28.22)   & (30.87)   & (29.11)    \\
citizen birth other         & 34.53***  & 37.24***  & 29.66**   & 33.18**    \\
                              & (13.24)   & (12.88)   & (14.09)   & (13.27)    \\
bio mom age first birth   & 2.28**    & 2.60**    & 1.11      & 1.80       \\
                              & (1.00)    & (1.03)    & (1.15)    & (1.13)     \\
bio mom age subject birth & -0.08     & -0.06     & -0.01     & 0.00       \\
                              & (0.55)    & (0.57)    & (0.62)    & (0.58)     \\
bio mom educ                & 2.91      & 3.21*     & -0.16     & 1.53       \\
                              & (1.95)    & (1.77)    & (2.04)    & (2.03)     \\
bio dad educ                & 1.84      & 3.22*     & -2.52     & 0.38       \\
                              & (1.84)    & (1.71)    & (2.06)    & (2.32)     \\
relation parent figure      & -8.81***  & -10.29*** & -5.41**   & -7.78***   \\
                              & (2.22)    & (2.44)    & (2.75)    & (2.82)     \\
parent religious             & -0.03     & -0.03     & -0.04     & -0.03      \\
                              & (0.03)    & (0.03)    & (0.03)    & (0.03)     \\
non english home            & -10.64    & -12.37    & -9.47     & -11.23     \\
                              & (13.01)   & (13.20)   & (14.40)   & (13.43)    \\
poor english interview      & 75.41     & 56.96     & 78.91     & 65.72      \\
                              & (112.29)  & (78.33)   & (85.58)   & (79.79)    \\
\hline
\multicolumn{5}{c}{ } \\[-1ex]
\multicolumn{5}{p{15cm}}{
\begin{minipage}{1 \linewidth}
\begin{footnotesize}
\textit{Notes}:
The table contains estimates and their standard errors (in parantheses) for slope coefficients on individual and family covariates using four estimators: Ordinary Least Squares (OLS), Proximal Learning (NC), Instrumental Variables (IV), and Instrumented Common Confounding (ICC). Asterisks indicate significance at the 1\% (***), 5\% (**) and 10\% (*) level.
\end{footnotesize}
\end{minipage}}
\end{tabular}
\end{center}
\end{table}

\begin{table}
\caption{Slope coefficients regional covariates}
\label{}
\begin{center}
\begin{tabular}{lllll}
\hline
                              & OLS       & NC        & IV        & ICC        \\
\hline
urban 17                     & 1.08      & -1.48     & 5.09      & 1.59       \\
                              & (10.04)   & (9.68)    & (10.64)   & (10.01)    \\
urban 35                     & -30.42*** & -28.21*** & -38.53*** & -32.79***  \\
                              & (11.48)   & (10.89)   & (12.00)   & (11.47)    \\
other non rural 17         & 24.13     & 19.02     & 42.62     & 30.34      \\
                              & (23.31)   & (24.07)   & (26.55)   & (25.25)    \\
other non rural 35         & -92.49    & -82.51    & -48.20    & -66.35     \\
                              & (72.08)   & (66.52)   & (73.14)   & (68.57)    \\
region north central 17    & 43.59*    & 42.51**   & 38.27*    & 40.92**    \\
                              & (22.91)   & (20.27)   & (22.13)   & (20.68)    \\
region north central 35    & -28.57    & -27.57    & -11.03    & -19.45     \\
                              & (24.64)   & (21.07)   & (23.28)   & (22.11)    \\
region south 17             & 16.45     & 14.74     & 8.67      & 11.96      \\
                              & (18.82)   & (18.22)   & (19.92)   & (18.66)    \\
region south 35             & -22.18    & -20.87    & -1.67     & -11.54     \\
                              & (19.98)   & (18.40)   & (20.49)   & (19.65)    \\
region west 17              & 27.39     & 25.91     & 24.62     & 25.23      \\
                              & (22.03)   & (19.97)   & (21.79)   & (20.34)    \\
region west 35              & 16.58     & 15.37     & 34.70     & 24.20      \\
                              & (22.98)   & (20.06)   & (22.19)   & (21.22)    \\
\hline
\multicolumn{5}{c}{ } \\[-1ex]
\multicolumn{5}{p{13cm}}{
\begin{minipage}{1 \linewidth}
\begin{footnotesize}
\textit{Notes}:
The table contains estimates and their standard errors (in parantheses) for slope coefficients on regional covariates using four estimators: Ordinary Least Squares (OLS), Proximal Learning (NC), Instrumental Variables (IV), and Instrumented Common Confounding (ICC). Asterisks indicate significance at the 1\% (***), 5\% (**) and 10\% (*) level.
\end{footnotesize}
\end{minipage}}
\end{tabular}
\end{center}
\end{table}

\begin{table}
\caption{Slope coefficients risky behaviour proxies}
\label{}
\begin{center}
\begin{tabular}{lrrrr}
\toprule
                              & OLS   &       & IV   &         \\
\midrule
ever drank 17               & 14.44     &           & 12.56     &            \\
                              & (10.10)   &           & (10.63)   &            \\
ever smoked 17              & -2.40     &           & 7.39      &            \\
                              & (10.01)   &           & (10.90)   &            \\
ever marijuana 17               & -9.03     &           & -0.38     &            \\
                              & (10.90)   &           & (12.10)   &            \\
ever ran away 17           & -8.52     &           & -12.39    &            \\
                              & (12.91)   &           & (14.74)   &            \\
ever attack 17              & -1.71     &           & 15.39     &            \\
                              & (10.46)   &           & (12.54)   &            \\
ever sell drugs 17         & -22.44    &           & -19.73    &            \\
                              & (15.87)   &           & (17.27)   &            \\
ever destroy property 17   & 3.86      &           & 2.57      &            \\
                              & (9.89)    &           & (10.59)   &            \\
ever steal bit 17          & -1.29     &           & -4.15     &            \\
                              & (9.50)    &           & (10.30)   &            \\
ever steal lot 17          & -20.98    &           & -11.98    &            \\
                              & (15.45)   &           & (17.12)   &            \\
\bottomrule
\multicolumn{5}{c}{ } \\[-1ex]
\multicolumn{5}{p{10cm}}{
\begin{minipage}{1 \linewidth}
\begin{footnotesize}
\textit{Notes}:
The table contains estimates and their standard errors (in parantheses) for slope coefficients on risky behaviour proxies using four estimators: Ordinary Least Squares (OLS), Proximal Learning (NC), Instrumental Variables (IV), and Instrumented Common Confounding (ICC). Asterisks indicate significance at the 1\% (***), 5\% (**) and 10\% (*) level.
\end{footnotesize}
\end{minipage}}
\end{tabular}
\end{center}
\end{table}

\end{document}